\documentclass[reprint,superscriptaddress]{revtex4-2}
\usepackage[english]{babel}
\makeatletter
\@namedef{l@en}{\l@english}
\makeatother

\usepackage{physics,amsmath,amsfonts,amssymb,dsfont,mathtools,amsthm}
\usepackage{caption}
\usepackage{ragged2e}

\makeatletter
\long\def\@makecaption#1#2{%
  \par\begingroup
    \small 
    \justifying
    \noindent #1\quad #2\par
  \endgroup
}
\makeatother

\usepackage{graphicx,xcolor,subcaption}
\usepackage{dcolumn}
\usepackage{bm}
\usepackage{comment}
\usepackage{enumerate}
\usepackage{hyperref}
\usepackage{cleveref} 

\newtheorem{lemma}{Lemma}

\begin{document}

\preprint{}

\title{Port-based teleportation under pure-dephasing decoherence}

\author{Rajendra S.~Bhati}
\email{rsbhati@cft.edu.pl}
\affiliation{Center for Theoretical Physics, Polish Academy of Sciences, Aleja Lotników 32/46, 02-668 Warszawa, Poland}

\author{Michał Studziński}
\email{michal.studzinski@ug.edu.pl}
\thanks{corresponding author}
\affiliation{International Centre for Theory of Quantum Technologies, University of Gdańsk,
Marii Janion 4, 80-309 Gdańsk, Poland}

\author{Jarosław K.~Korbicz}
\email{jkorbicz@cft.edu.pl}
\affiliation{Center for Theoretical Physics, Polish Academy of Sciences, Aleja Lotników 32/46, 02-668 Warszawa, Poland}


\begin{abstract}
We study deterministic port-based teleportation in the presence of noise affecting both the entangled resource state and the measurement process. We focus on a physically motivated model in which each Bell pair constituting the resource interacts with an identical local environment, corresponding to independently distributed entangled links.
Two noisy scenarios are analysed: one with decoherence acting solely on the resource state and ideal measurements, and another with noisy, noise-adapted measurements, which in our analysis are chosen as square-root measurements constructed from the noisy ensemble. In the first case, we derive an analytical lower bound and later a closed-form expression for the entanglement fidelity of the teleportation channel and analyse its asymptotic behaviour. In the second, we combine semi-analytical and numerical methods.
Surprisingly, we find that such noise-adapted measurements perform worse than the noiseless ones.
To connect the abstract noise description with microscopic physics, we embed the protocol in a spin–boson model and investigate the influence of bath memory and temperature on the teleportation fidelity, highlighting qualitative differences between different environments. We also briefly discuss the algorithmic aspects of noisy port-based teleportation,
highlighting how noise modifies the status of efficient measurement implementations.
\end{abstract}

\maketitle

\section{Introduction}
Port-based teleportation (PBT), introduced by Ishizaka and Hiroshima~\cite{ishizaka_asymptotic_2008,ishizaka_quantum_2009}, is a variant of the standard quantum teleportation protocol~\cite{bennett_teleporting_1993} in which the receiver applies no unitary correction to the output system. Instead, the sender (Alice) and the receiver (Bob) share a multipartite resource state; Alice performs a single joint measurement on her share of the resource and the input state, and Bob retrieves the teleported state simply by selecting one of his subsystems (ports) according to the classical message sent by Alice. We distinguish two variants of PBT: deterministic non-exact, where the state is always transmitted, but is noisy; probabilistic exact, where if the transmission is successful, then teleportation is perfect.  Due to the non-programming theorem~\cite{Nielsen1997,PhysRevLett.125.210501}, both variants cannot be implemented perfectly when the size of the resource state is finite -- the teleported state is always noisy, or the probability of failure is always non-zero. This limitation makes it important to understand how the efficiency of the teleportation procedure depends on the dimension of the teleported state, the number of shared ports (size of the resource state), and what the optimal settings of the protocol are. The full answer was given in a series of papers, where the analysis was made possible by the use of tools from representation theory and semidefinite programming~\cite{Studzinski2017,StuNJP,wang_higher-dimensional_2016, majenz2, leditzky2020optimality}. 

The absence of the unitary correction step in PBT has deep structural implications. This feature endows PBT with a unique operational flexibility that makes it a natural building block for many quantum information tasks: Bob can process the teleported system even before receiving the classical information about which port is correct (delayed-input formalism). 
In this way, PBT underlies a variety of higher-order quantum operations (HOQO) primitives~\cite{taranto2025higherorderquantumoperations}, including universal programmable quantum processors~\cite{ishizaka_asymptotic_2008,sim} and their hybrid versions~\cite{Muguruza2024portbasedstate,brzic2025}, schemes of unitary learning~\cite{yoshida2024onetoone}, transposing and inverting unitary operations~\cite{PhysRevLett.123.210502,PhysRevA.100.062339,Quintino2022deterministic}, as well as protocols for storage and retrieval of unknown operations~\cite{Stroing,PhysRevA.81.032324}. Beyond these applications, PBT has also proved influential in more traditional quantum information scenarios, where its structure has enabled progress in instantaneous non-local quantum computation, position-based cryptography~\cite{beigi_konig}, and the study of interaction complexity and entanglement in non-local and holographic models~\cite{may2022complexity}. It has further contributed to clarifying the connection between quantum communication complexity advantages and violations of Bell inequalities~\cite{buhrman_quantum_2016} and  others~\cite{pereira2023continuous,limit,10315956}. In the last years, progress in the area of PBT has put it on a much firmer mathematical and algorithmic footing: its performance is now fully characterized in terms of representation theory~\cite{Studzinski2017,StuNJP,majenz2,grinko2023gelfandtsetlinbasispartiallytransposed}, and efficient quantum circuits and algorithms for optimal PBT in various regimes have been constructed~\cite{fei2023efficientquantumalgorithmportbased,Grinko2024,PRXQuantum.5.030354}.

 All of these developments, however, rely on essentially ideal assumptions: carefully optimized resource state and perfect measurements. In any realistic architecture, the multipartite resource state will be imperfect, affected by local or correlated noise, and the implementation of the PBT measurement itself will suffer from errors. From the operational point of view, a noisy PBT scheme is simply an effective quantum channel acting on the input state; understanding this effective channel is crucial if PBT is to be used as a building block of the above-mentioned primitives.  Studying PBT in noisy conditions is therefore not merely a technical refinement of the ideal theory, but a necessary step in assessing its true operational power.  This direction is beginning to be explored: for instance, recent work~\cite{Kim} has analyzed PBT when the shared ports are subjected to local Pauli noise, fully characterizing the induced teleportation channel, and using this to study port-based entanglement teleportation as a tool for entanglement distribution in quantum networks.

Motivated by these developments, we study PBT in an open-quantum-systems setting. Rather than treating noise only at the level of an effective teleportation channel, we explicitly model the interaction between each entangled resource pair and its environment, allowing us to derive the resulting noisy PBT channel from the underlying dynamics. In particular:

\begin{enumerate}[(a)]

\item We focus on the deterministic PBT, where each Bell pair forming the resource state is coupled to its own local environment, with all environments assumed to be identical, reflecting a scenario in which the parties share (N) identical entangled `wires'. While more general models with a common environment are possible, they are technically more involved and are not considered here. The used model is discussed in more detail in Section~\ref{Sec:noisyPBT_model}.

Throughout this work, we consider pure-dephasing coupling, motivated by its ubiquitous occurrence in realistic physical noise models \cite{schlosshauer2007} and by the fact that it allows for a tractable yet physically relevant analysis. Other important noise models, such as amplitude damping and depolarizing noise, can also affect the resource state. However, as the present treatment of pure dephasing already demonstrates, the analysis of each noise model is technically demanding and therefore deserves a separate, dedicated study (see, for example, Ref.~\cite{Kim} for an information-theoretic treatment of general Pauli noise). The same consideration applies to the probabilistic PBT protocol, whose analysis is left for future work.

\item We consider two scenarios of the noisy PBT protocol. In the first scenario, we consider a noisy resource state and noiseless measurements. We assume perfect control over the measurements--taken to be square-root measurements (SRMs)~\footnote{In the literature, square-root measurements (SRMs) are sometimes referred to as pretty-good measurements (PGMs). We will use the two terms interchangeably.}--in Alice’s laboratory, and uncontrolled interaction of the resource state distributed among the parties. In this scenario, we first derive a lower bound for the entanglement fidelity of the teleportation channel using techniques introduced in~\cite{beigi_konig}. Then we develop a fully analytical approach to the problem, presenting a closed-form expression for the entanglement fidelity, together with a discussion on its asymptotic behaviour. This is contained in Section~\ref{BK} and Section~\ref{Sec noiseless}.
    In the second case, we assume a scenario with noisy measurements. In particular, our measurements are noise-adapted ones, which in our analysis correspond to SRMs constructed from the noisy ensemble. Here, we mostly rely on a semi-analytical approach, supported strongly by a numerical approach. However, for the special case of $N=2$, we derive a fully analytical lower bound on the entanglement fidelity relying on Helstrom's bound. This is considered in Section~\ref{noisyPOVMs}.  In addition, in Section~\ref{sec:simulability-noisyPBT}, we discuss the algorithmic implications of these two scenarios. When the noiseless measurements are used, noise affects only the achieved fidelity and the required number of ports, while existing efficient implementations of the measurement remain applicable. By contrast, for noise-adapted POVMs constructed from the noisy ensemble, the loss of joint unitary symmetry means that efficient implementation can no longer be inferred from symmetry alone, and becomes a quantitative question tied to spectral properties of the noisy ensemble.
    \item  In order to illustrate the impact of realistic environments. 
    on the performance of the noisy PBT protocol, we apply our general results to the concrete example of  
the spin--boson model \cite{10.1098/rspa.1996.0029, PhysRevA.65.032326, Schlosshauer2007-tx}, which has been widely used to model interactions of physical qubits with their environments, see e.g. \cite{ruggiero2006quantum, predojevic2015engineering, gaitan2008quantum, weiss2021quantum} and references therein. The model describes a central spin-register, in our case two spins of each of the Bell pair, interacting with a thermal bosonic bath. The model is parametrized
in terms of three dimensionless quantities: the Ohmicity exponent $s$,
the ratio of the thermal energy to the cutoff energy $T/\Lambda$, and the
dimensionless time $\tau = t\Lambda$ (we are using units in which $\hbar=1$ for simplicity).
 We illustrate the impact of microscopic noise properties by analysing two
representative environmental regimes: modelled by a super-Ohmic spectral density with $s=2$, which for a single qubit corresponds to a Markovian regime with weak memory
effects, and with $s=3$,  corresponding to a
non-Markovian regime with enhanced environmental memory.
In both cases, we study the time dependence of the teleportation fidelity as a
function of time to show how the duration of the decoherence affects the teleportation process. 
This is considered in Section~\ref{Sec:spin-boson}.
\end{enumerate}
\begin{figure}[t]
    \centering    \includegraphics[width=0.95\linewidth]{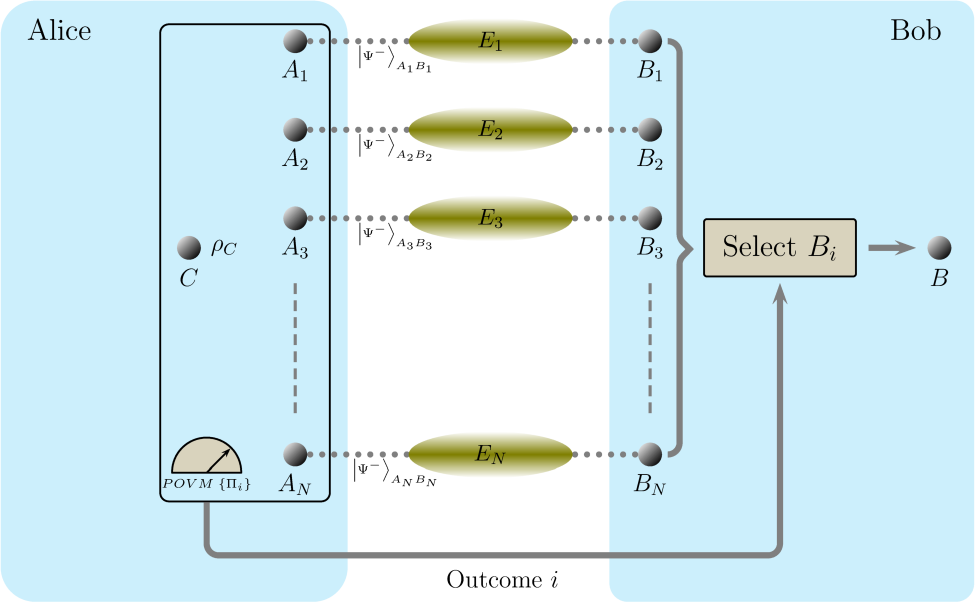}
    \caption{Port-based teleportation under decoherence. Alice and Bob share \(N\) copies of the Bell state \(\ket{\Psi^-}\), each transmitted through a decoherence channel and interacting with its own environment \(\{E_i\}_{i=1}^N\). Alice performs an \(N\)-outcome positive operator-valued measure (POVM) \(\{\Pi_i\}_{i=1}^N\)  on her \(N\) qubits together with the qubit \(C\) to be teleported. Bob then selects the port corresponding to Alice's measurement outcome.
    }
    \label{fig:PBT_scheme}
\end{figure}
The paper additionally contains Section~\ref{Sec:infos}, where all preliminary information is collected, and several appendices, Appendix~\ref{Fij}--Appendix~\ref{app:spin-boson}, where, for the reader's convenience, we collect detailed calculations and proofs.

 \section{Preliminaries and deterministic noiseless PBT}
 \label{Sec:infos}

In this section, we collect the definitions used to quantify the performance of deterministic PBT and recall the noiseless protocol. These ingredients are needed throughout the technical analysis below. For further details, we refer the reader to previous works.

 We start from the basic definitions of the average fidelity and
entanglement fidelity. The average fidelity of a quantum channel $\mathcal{C}$ (here teleportation channel) is given by
\begin{equation}
    f(\mathcal{C}) := \int \langle \psi | \mathcal{C}(|\psi\rangle\langle\psi|) | \psi \rangle \, d\psi,
\end{equation}
where the integral is performed with respect to the uniform
distribution $d\psi$ over all $d$-dimensional pure states~\cite{HorodeckiMPRFidelity}.
The entanglement fidelity of a channel $\mathcal{C}$ is given by
\begin{equation}
\label{eq:entFid}
    F(\mathcal{C}) := \Tr\!\left[ P^+ \, (\mathcal{C} \otimes \mathbf{1}) P^+ \right],
\end{equation}
where
$\quad  P^+ := \frac{1}{d} \sum_{i,j=1}^{d} |i\rangle\langle j| \otimes |i\rangle\langle j|$
is the projector on the bipartite maximally entangled state $|\psi^+\rangle=(1/\sqrt{d})\sum_i |ii\rangle$, and $\mathbf{1}$ denotes the identity channel.
Our channel in this context is a teleportation channel. There exists a universal relation between the average fidelity $f$ of a teleportation channel and its entanglement fidelity $F$~\cite{HorodeckiMPRFidelity} given by 
\begin{equation}
    f(\mathcal{C}) = \frac{dF(\mathcal{C}) + 1}{d+1}.
    \label{teleportfid}
\end{equation}
Throughout this work, we use the terms `teleportation fidelity' and `average fidelity' interchangeably to denote the average fidelity of the teleportation channel.

 In this paper, we consider the standard PBT scheme first proposed
in~\cite{ishizaka_asymptotic_2008}, and restrict our discussion to qubit systems, i.e., when $d=2$.
In the standard configuration of any PBT protocol, the sender (Alice) and the receiver (Bob) share a $2N$-qubit pure state (resource state):
\begin{equation}
    \Psi_{\vec{A}\vec{B}}^-
    := \psi^-_{A_1 B_1} \otimes \cdots \otimes \psi^-_{A_N B_N},
\end{equation}
where $\vec{A}=A_1A_2\cdots A_N, \vec{B}=B_1B_2\cdots B_N$, and $P_{A_iB_i}^-\equiv \psi_{A_iB_i}^- :=|\psi^-\rangle \langle \psi^-|_{A_iB_i}$ is the projector on the bipartite singlet state of two qubits on the $i-$th port:
\begin{equation}
    |\psi^-\rangle_{A_iB_i}:=\frac{1}{
    \sqrt{2}}(|01\rangle-|10\rangle).
\end{equation}

Alice teleports an unknown qubit state $\rho_C$ by performing a joint POVM
measurement $\{\Pi_i\}_{i=1}^N$ on the systems $\vec{A}$ and $C$, yielding
an outcome $i \in \{1,\dots,N\}$.
 After the measurement, she sends her outcome $i$ to Bob using a classical
communication channel. Bob, to recover the teleported state, selects qubit $B_i$, and removes the
rest of his qubits. The PBT channel is expressed as
\begin{equation}
    \mathcal{C}(\rho_C)
    := \sum_{i=1}^N \operatorname{tr}_{\vec{A}\bar{B}_i C}
    \Big[
        \big(\Pi_{i}^{\vec{A}C} \otimes \mathbf{1}_{\vec{B}} \big)
        \big( \Psi_{\vec{A} \vec{B}}^- \otimes \rho_C \big)
    \Big]_{B_i \rightarrow B},
\end{equation}
where $\operatorname{tr}_{\bar{B}_i}$ denotes tracing out all systems but $B_i$. The selected port is regarded as a common qubit $B$ with the teleported state, which is just a renaming, denoted by $B_i\rightarrow B$. Since we work with deterministic and non-exact transmission to quantify performance of the teleportation channel, we use the entanglement fidelity $F(\mathcal{C})$ from~\eqref{eq:entFid}, but for $P^-$ instead of $P^+$. It is because the entanglement fidelity is invariant under local unitaries, and the considered states are connected by a unitary of the form $Y \otimes \mathbf{1}$, where $Y= \bigl(\begin{smallmatrix} 0 & -i \\ i & 0 \end{smallmatrix}\bigr)$ is a Pauli matrix.
The exact form of the entanglement fidelity
$F(\mathcal{C})$  for the teleportation channel $\mathcal{C}$ is given by~\cite{ishizaka_asymptotic_2008,ishizaka_quantum_2009}:
\begin{equation}
   F(\mathcal{C})= \frac{1}{2^{N+3}}
    \sum_{k=0}^N
    \left(
        \frac{N - 2k - 1}{\sqrt{k+1}}
        + \frac{N - 2k + 1}{\sqrt{N - k + 1}}
    \right)^2
    \binom{N}{k}.
\end{equation}
For further purposes, we present derivations of the above expression in Appendix~\ref{derivationPRA}. In the asymptotic limit \(N\to\infty\), the entanglement fidelity and the average fidelity behave respectively as 
\begin{equation} \label{F_IH_inf}
F(\mathcal{C}) \to 1 - \frac{3}{4N},\qquad f(\mathcal{C}) \to 1 - \frac{1}{2N}. 
\end{equation}
 The asymptotic of the deterministic and probabilistic PBT is also known for the case when $d>2$ and presented in  papers~\cite{majenz2,Studzinski2017}. 
The maximal value of the entanglement fidelity from~\eqref{eq:entFid} is achieved by the square-root measurement (SRM), equivalently known as the pretty-good measurement (PGM). These measurements have the following explicit form for every \(1\leq i\leq N\):
\begin{equation}
    \Pi_{i} := \widetilde{\Pi}_{i} + \Delta,
    \qquad
    \Delta := \frac{1}{N} \sum_{i=1}^N \widetilde{\Pi}_{i},
\end{equation}
with
\begin{equation}
\label{eq:POVMs0}
    \widetilde{\Pi}_{i}^{\vec{A}C} := \rho^{-1/2}\, \sigma_{i} \, \rho^{-1/2},
    \qquad
    \rho:= \sum_{i=1}^N \sigma_{i}.
\end{equation}
The operators $\sigma_{i}$ are called signal states, expressed
as
\begin{equation}
    \sigma_i
    := \frac{1}{2^{N-1}} \psi^-_{A_i C} \otimes \mathbf{1}_{\bar{A}_i},
\end{equation}
where $\mathbf{1}_{\bar{A}_i}$ represents the identity operator acting
on all $\vec{A}$ systems excluding system $A_i$.  The inverse in~\eqref{eq:POVMs0} is taken on the support of $\rho$ and thus the POVM must be supplemented by the operator $\Delta $ to have $\sum_i \Pi_i^{\vec{A}C}=\mathbf{1}_{\vec{A}C}$, which projects on the zero-eigenvalue subspace and thus has no contribution to the entanglement fidelity. The explicit proof of optimality of the SRM measurements for the deterministic PBT has been shown in~\cite{leditzky2020optimality}.

\section{The noisy teleportation channel}
\label{Sec:noisyPBT_model}
 In this work, we assume that the environment affects only the Bell pairs, while the input qubit to be teleported remains unaffected. This is an approximation to a real situation, motivated by the fact that entangled states decohere faster than the states of a single system \cite{PhysRevLett.93.140404,dur2004stability,zurek2003decoherence,SCHLOSSHAUER20191}. We will assume the following form of the interaction between the Bell pairs $AB$ and the environment $E$:  
\begin{equation}\label{Hint}
    H_{int} = \sum_{i,j\in\{0, 1\}}\ketbra{ij}_{AB}\otimes V^E_{ij},
\end{equation}
which has the form of a general pure dephasing in the product basis $\ket{ij} = \ket i \otimes \ket j$ of the qubits $AB$, i.e. it destroys coherence in the basis $\ket{ij}$, in particular it decoheres Bell states. The basis $\{\ket {ij}\}$ is called the pointer basis and without a loss of generality we can choose it to be the eigenbasis of $\sigma_z$. The environmental coupling observables $V^E_{ij}$ are at this point left arbitrary. The form of the interaction \eqref{Hint} is motivated by physical considerations as it is often encountered in open systems. For example, a broad class of dynamics of a type $H=S_A\otimes V_E+S_B\otimes V'_E+H_E$, where $S_{A}, S_B$ are the coupling observables of the systems $A$ and $B$ side, $V_E,V_E'$ the corresponding environmental coupling observables and the $H_E$ internal dynamics of the environment can be transformed to form \eqref{Hint}. The internal dynamics of the qubits is assumed to be suppressed or at most commuting with the interaction $[H_{AB},S_A]=[H_{AB},S_B]=0$. Although this can be criticized from the open system's perspective as a rather restrictive condition, it can be nevertheless justified here as one would like to prepare the communication qubits as close to the Bell state as possible and keep them in that state.  Another reason for this choice is the fact that it allows for analytical calculations in principle without the need of the tracing out of the environment, which may be kept as a bona fide quantum system throughout the analysis. This in turn may reveal surprising decoherence suppression effects like e.g. the purifying teleportation \cite{Roszak2023purifying}. Although we will not study such scenarios here due to their complication in the case of the PBT scheme, our approach can be in principle directly applied to it in the future. A concrete example of the dynamics \eqref{Hint} is the popular spin-boson model with multiple spins (a spin register)  \cite{10.1098/rspa.1996.0029, PhysRevA.65.032326,Tuziemski2018, Schlosshauer2007-tx}, which we will study in Section \ref{Sec:spin-boson}.

The Hamiltonian \eqref{Hint} generates the following unitary evolution, known as pure dephasing evolution in the open systems community, coherent control or generalized C-NOT evolution in the quantum information community, or generalized von Neumann measurement in the measurement theory:
\begin{equation}
    U_{AB:E} = \sum_{i,j\in\{0,1\}}\ketbra{ij}\otimes\exp(-i V_{ij}t).
\end{equation}
The state of any of the Bell pairs after the decoherence is given by
\begin{align}
        \tilde{P}^-_{A B} & =\tr_E U_{A B:E}\left(\ketbra{\Psi^-}_{A B}\otimes\rho_E\right)U_{AB:E}^\dagger \label{dephasing_channel} \\
        & = \frac{1+\vert\Gamma\vert\cos{\theta}}{2}\ketbra{\Psi^-}+\frac{1-\vert\Gamma\vert\cos{\theta}}{2}\ketbra{\Psi^+} \\
        &+\frac{i\vert\Gamma\vert\sin{\theta}}{2}\left(\ketbra{\Psi^+}{\Psi^-}-\ketbra{\Psi^-}{\Psi^+}\right) \label{P_tilde_1}\\
        & = \frac{1+|\Gamma|}{2}\mathds{1}\otimes R(\theta)\ketbra{\Psi^-}_{A B}\mathds{1}\otimes R(\theta)^\dagger \\
        &+ \frac{1-|\Gamma|}{2}\mathds{1}\otimes  R(\theta)\sigma_z \ketbra{\Psi^-}_{A B}\mathds{1}\otimes \sigma_z R(\theta)^\dagger \label{P_tilde} \\
        & \equiv \mathds{1}\otimes\mathcal{E}_{\Gamma}\left(\ketbra{\Psi^-}_{A B}\right).
\end{align}
Here 
\begin{align}\label{Gamma}
\Gamma = \tr_E[e^{-i V_{01}t}\rho_E e^{i V_{10}t}] \equiv \vert\Gamma\vert e^{i\theta},
\end{align}
$|\Gamma|\leq 1$,  is the only non-trivial decoherence factor, $R(\theta)$ is a relative phase rotation in the pointer basis:
\begin{equation}
\label{eq:Rotation}
    R(\theta) = e^{-i\theta}\ketbra{0}+\ketbra{1},
\end{equation}
and $\mathcal{E}_{\Gamma}$ 
is a single qubit decoherence channel, defined as:
\begin{align} 
\label{single_q_channel}
    \mathcal{E}_{\Gamma}(\rho)= R(\theta)\left[\frac{1+|\Gamma|}{2}\rho  + \frac{1-|\Gamma|}{2} \sigma_z \rho \sigma_z \right] R(\theta)^\dagger .
\end{align}
The high symmetry of $\ket{\Psi^-}$ allows to move the effect of the environment to a single qubit only, despite the correlated character of the interaction \eqref{Hint}. If the decoherence factor was real, i.e. $\theta=0$, this would be a Pauli noise channel with operators $\mathds 1, Z $ \cite{Kim}. The decoherence factor \eqref{Gamma} depends on the interaction time $t$, but for the moment we neglect that dependence and study the efficiency of the protocol as a function of $\Gamma$. We note that $\Gamma=1$ corresponds to no decoherence, while $\Gamma=0$ to a fully decohered, separable resource, which is no different from using classical randomness:
\begin{align}
        \tilde{P}^-_{A B}(\Gamma=0) = \frac{1}{2} \ketbra{01} + \frac{1}{2} \ketbra{10}.
        \label{strong_decoh}
\end{align}

Let us now briefly describe the PBT procedure \cite{ishizaka_asymptotic_2008,ishizaka_quantum_2009, Kim}. Consider that Alice and Bob share $N$ singlet pairs, which were affected by the decoherence ~\eqref{single_q_channel}. The shared resource state becomes: 
\begin{equation}
\rho_{\vec{A}\vec{B}}=\bigotimes_{i=1}^N\ketbra{\Psi^-}_{A_i B_i}\equiv \bigotimes_{i=1}^N{{\tilde{P}^-}}_{A_i B_i},
\end{equation}
where Alice and Bob have $N$ qubits each $A_1, A_2,\cdots, A_N$ and $B_1, B_2, \cdots, B_N$, respectively. Here onwards, we use the following notations through out the paper: Alice's $N$ qubits are denoted by $\vec{A}$, all except $A_i$ is denoted by $\bar{A}_i$, all except $A_i$ and $A_j$ are denoted by $\bar{A}_{ij}$, and so on. Similar notations are adopted for Bob's qubits as well. The resource state $\rho_{\vec{A}\vec{B}}$ is employed for teleportation of an unknown state of another qubit, say $C$, from $A$ to $B$. Alice performs a joint measurement with $N$ possible outcomes on the $\vec{A}$ and $C$ qubits, described by a POVM $\{\Pi_i\}$. Alice then sends the information regarding the index of the outcome, for example the $i$-th outcome, and Bob selects the corresponding port and discards the rest. The selected port is regarded as a common qubit $B$ with the teleported state, which is just a renaming, denoted by $B_i\rightarrow B$. The corresponding teleportation channel, which maps the state acting on the Hilbert space $\mathcal{H}_C$ to those on $\mathcal{H}_B$ and assuming corrupted resource \eqref{dephasing_channel} then reads (cf. \cite{Kim}):
\begin{equation}
    \begin{aligned}
        \mathcal{C}(\sigma^{in}_C) & = \sum_{i=1}^N \left[\tr_{\vec{A}\bar{B}_i C}\sqrt{\Pi_i}\left(\rho_{\vec{A}\vec{B}}\otimes\sigma_C^{in}\right)\sqrt{\Pi_i}^\dagger\right]_{B_i\rightarrow B} \\
        &= \sum_{i=1}^N\tr_{\vec{A} C}\Pi_i\left[\eta_{\vec{A} B}^{(i)}\otimes\sigma_{C}^{in}\right],\label{Lambda}
    \end{aligned}
\end{equation}
where we introduced the following, normalized, states:
\begin{equation}
\label{eta_i}
        \eta_{\vec{A} B}^{(i)}  = \left[\tr_{\bar{B}_i}(\rho_{\vec{A} \vec{B}})\right]_{B_i\rightarrow B}  =\frac{1}{2^{N-1}}\tilde{P}^-_{A_i B} \otimes \mathds{1}_{\bar{A}_i} 
\end{equation}
The state $\eta_{\vec{A} B}^{(i)}$ can be expressed as a mixture of signal states corresponding to the scenarios when a noiseless PBT is performed with $\ket{\Psi^-}^{\otimes N}$ and $\ket{\Psi^+}^{\otimes N}$ states as the resources and a subsequent phase rotation at the Bob's qubit:
\begin{equation}
    \eta_{\vec{A} B}^{(i)} = R(\theta)_B\left[\frac{1+|\Gamma|}{2}\sigma^{(i)}_{\vec{A}B} + \frac{1-|\Gamma|}{2}\omega^{(i)}_{\vec{A}B}\right] R(\theta)_B^\dagger,\label{diag}
\end{equation}
where:
\begin{align}
    \sigma^{(i)}_{\vec{A}B} = &  \frac{1}{2^{N-1}}\ketbra{\Psi^-}_{A_i B} \otimes \mathds{1}_{\bar{A}_i}, \label{sigma_i}\\
       \omega^{(i)}_{\vec{A}B} = &  \frac{1}{2^{N-1}}\ketbra{\Psi^+}_{A_i B} \otimes \mathds{1}_{\bar{A}_i}.\label{omega_i}
\end{align} 
The entanglement fidelity $F$ of $\mathcal{C}$ is given by:

\begin{equation}
   \begin{aligned}
        F &= \tr P^-_{BD}\left[\left(\mathcal{C}\otimes\mathds{1}\right)P^-_{CD}\right] \\
        &= \tr\sum_{i = 1}^N P^-_{BD}\Pi_{i}^{\vec{A} B}\left\{\eta^{(i)}_{\vec{A}B}\otimes P^{-}_{CD}\right\}  \\
         & = \frac{1}{2^2}\sum_{i=1}^N\tr \Pi_{i}^{\vec{A} B}\eta^{(i)}_{\vec{A} B} \equiv \frac{N}{4} P_{succ} \label{F}.
   \end{aligned}
\end{equation}
    
As the last equality above shows, it is proportional to the success probability of the state discrimination task defined by the ensemble $\left\{\frac{1}{N}, \eta^{(i)}_{\vec{A} B}\right\}_{i=1}^N$, $i=1,2,\dots,N$, of equiprobable signal states $\eta^{(i)}_{\vec{A} B}$ of \eqref{eta_i}. 
The problem of finding the teleportation fidelity then reduces to the state discrimination task, as in the original PBT protocol. Following the original idea~\cite{ishizaka_asymptotic_2008,ishizaka_quantum_2009}, we use the PGMs~\cite{beigi_konig,leditzky2020optimality} as a standard tool for state-discrimination tasks. For an ensemble \(\{p_i,\sigma_i\}_{i=1}^N\), it is defined as \begin{align} \Pi_i \equiv p_i\bar\sigma^{-1/2}\sigma_i\bar\sigma^{-1/2}, \label{PGM} \end{align} where \(\bar\sigma\equiv\sum_i p_i\sigma_i\) is the ensemble average state. If \(\bar\sigma\) has zero eigenvalues, the inverse is taken on its support, and the POVM is supplemented by the operator \(\Delta\equiv\mathds{1}-\sum_i\Pi_i\), which projects onto the zero-eigenvalue subspace and does not contribute to the success probability.

However, the states \eqref{eta_i} to be discriminated are now more complicated than in the original PBT scheme, i.e. for $\Gamma=1$. First, their non-trivial part, the reduction to the $A_iB$ subspace is not pure but mixed, supported on a $2D$ subspace spanned by $\ket{\Psi^-},\ket{\Psi^+}$, cf. \eqref{P_tilde}. Second, and more importantly, the original symmetry $U^{\otimes N}\otimes Y U^*Y$ of the non-decohered ensemble $\{\sigma^{(i)}_{\vec{A}B}\}$ \eqref{sigma_i} is lost. The latter allowed for a great simplification: using the spin representation theory, the authors of \cite{ishizaka_asymptotic_2008,ishizaka_quantum_2009} were able to simultaneously diagonalize the ensemble states, their average, and the pretty good measurement $\{\Pi_i\}$ and as a result calculate the fidelity analytically. Under the action of decoherence, the unitary symmetry is lost since 
$|\Psi^{+}\rangle$ has a different symmetry structure than $|\Psi^{-}\rangle$. 
In particular, \((U \otimes X U^{*} X)\,|\Psi^{+}\rangle = |\Psi^{+}\rangle \), where $X$ denotes the Pauli $X$ matrix. Although each ensemble state~\eqref{eta_i} still possesses the $U^{\otimes (N-1)}$ symmetry, their average does not 
inherit this property, since the operator $\mathbf{1}_{\bar A_i}$ is supported  on different subspaces for different values of $i$. As a remark, we note that one would be tempted to use the Knill-Barnum bound \cite{barnum2002reversing, Montanaro2019,Audenaert2014} to estimate $P_{succ}$ as in fact the former was derived with the help of the PGM. The bound reads $P_{succ}\geq 1-\sum_{i\ne j}\sqrt{p_ip_j}F(\sigma_i,\sigma_j)$, where $F(\sigma_i,\sigma_j)\equiv \tr\sqrt{\sqrt{\sigma_i}\sigma_j\sqrt{\sigma_i}}$ the pairwise state fidelity of the ensemble. In our case, we obtain from \eqref{eta_i}, see Appendix ~\ref{Fij}
\begin{align}\label{overlap}
    F\left(\eta^{(i)},\eta^{(j)}\right) = \frac{1}{2}
\end{align}
which is both $N$- and $\Gamma$-independent and leads to a bound $P_{succ}\geq 1 -\frac{1}{N}\frac{N(N-1)}{4}=1-\frac{N-1}{4}$, which trivializes already for $N>3$.

Below, we analyse three better strategies, of increasing complication, for estimating the PBT efficiency under the pure dephasing decoherence.

\section{The Beigi-Koenig bound}\label{BK}

The easiest way to obtain a bound on the fidelity \eqref{F} is to use the main result of \cite{beigi_konig}:  The success probability of the PGM for the equiprobable ensemble $\{\frac{1}{N}, \eta^{(i)}\}^{N}_{i=1}$ is bounded by:
\begin{equation}\label{BKbound}
    P^{pgm}_{succ}\geq \frac{1}{N\bar{r}\tr\bar\eta^2},
\end{equation}
where $\bar{\eta}=(1/N)\sum_{i=1}^N\eta^{(i)}$ is the ensemble average state, $\tr\bar\eta^2$ is its purity, and $\bar{r}=(1/N)\sum_{i=1}^N \mbox{rank}\:\eta^{(i)}$ is the average rank. From \eqref{eta_i} we obtain:
\begin{equation}
    \tr\bar{\eta}^2 = \frac{1}{2^N N}\left(\vert \Gamma \vert^2+\frac{N+1}{2}\right).
\end{equation}
where we have used the individual state purities:
    \begin{align} \label{purity}
        \tr[(\eta^{(i)})^2]= \frac{1}{2^{N-1}}\tr(\tilde{P}^-)^2 = \frac{1+\vert\Gamma \vert^2}{2^N},
    \end{align}
and $\tr{\eta^{(i)} \eta^{(j)}}_{i\neq j} =2^{-(N+1)}$. Since $\mbox{rank}\:\eta^{(i)}= 2^{N-1} \cdot \mbox{rank}\: \tilde{P}^- =2^N$ for every $i$, cf. \eqref{eta_i}, we obtain:
\begin{equation}
    \begin{aligned}
        P^{pgm}_{succ} & 
        = \frac{2}{N}\left(1+\frac{1+2\vert \Gamma\vert^2}{N}\right)^{-1}
        \geq \frac{2}{N}\left( 1-\frac{1+2\vert \Gamma\vert^2}{N} \right),
    \end{aligned}
\end{equation}
and from this the bound on the entanglement fidelity:
    \begin{align}\label{BKbound}
        F \geq \frac{1}{2}\left(1-\frac{1+2\vert \Gamma\vert^2}{N} \right).
    \end{align}
This formula looks almost correct in the absence of decoherence, \(\Gamma=1\), apart from the fact that the bound is one half of what it should be~\cite{beigi_konig}. This is a consequence of the fact that, in the noisy case, \(\operatorname{rank}\tilde P^-=2\), whereas for the ideal, uncorrupted resource one has \(\operatorname{rank}P^-=1\). Thus the Beigi--Koenig estimate suffers an immediate drop whenever the rank of the signal states increases. It is useful to distinguish this estimate from the pairwise-fidelity bound discussed above. The latter becomes trivial already for \(N>3\), whereas the Beigi--Koenig bound remains positive but is operationally uninformative for our purpose, since it certifies at most fidelities of order \(1/2\) in the large-\(N\) regime and therefore cannot capture the high-fidelity behaviour expected for weak dephasing. We thus conclude that Eq.~\eqref{BKbound} is of limited use for assessing noisy PBT in the regime of interest.

\section{PBT using the noiseless Ishizaka-Hiroshima measurements}\label{Sec noiseless}

The next strategy is to use the original Ishizaka-Hiroshima deterministic PGM, constructed from the noiseless ensemble, i.e. from $\eta_{\vec{A} B}^{(i)}$ with $\Gamma=1$, cf. \cite{Kim}. We will call it the noiseless measurement and its elements the noiseless POVMs. The motivation is that one can imagine a situation, where the noise estimation, i.e. the estimation of $\Gamma$, is impossible or even if possible, it cannot be used to fine-tune the measurements. The natural choice is then to use the noiseless measurement as at least for weak decoherence it will perform close to optimal. As we will show, this strategy performs surprisingly well not only for the weak decoherence. 

There are at least two possible approaches: A direct calculation of the entanglement fidelity  \eqref{F} or a calculation using the fact that the PBT channel \eqref{Lambda} is a depolarizing channel for noiseless measurements \cite{ishizaka_asymptotic_2008,ishizaka_quantum_2009}. The latter was studied in \cite{Kim} and here we perform a direct calculation for completeness, presenting below only the main steps while the detailed calculations can be found in the Appendix~\ref{App_noiseless}. 

For simplicity, throughout this Section we will use the un-normalized ensemble average and POVMs, i.e. we consider
the ensemble $\left\{\frac{1}{N}, \sigma^{(i)}_{\vec{A} B}\right\}$, with $\sigma^{(i)}_{\vec{A} B}$ given by \eqref{sigma_i}, and define the un-normalized average state by:
\begin{align}\label{avrho}
    \rho=\sum_{i=1}^N\sigma^{(i)}_{\vec{A} B},
\end{align}
where we used the $\rho$ rather than $\bar\sigma$ to comply with the original notation of \cite{ishizaka_asymptotic_2008,ishizaka_quantum_2009}. As it turns out, the original analysis based on the symmetry and the spin decomposition can be still applied here with a slight modification. 

Since all the $\sigma^{(i)}_{\vec{A} B}$ share the same symmetry, the average state \eqref{avrho} shares it also and can be shown to be block-diagonal with respect to the total spin $s$ of the system of $(N+1)$ spins $1/2$, given by the qubits $\vec{A} B$:
\begin{equation}\label{rhos}
    \rho = \sum_{s = s_{min}}^{(N-1)/2}\rho_-(s)\oplus \rho_+(s),
\end{equation}
where $s_{min}=0$ for $N$ even and $1/2$ for $N$ odd; $\rho_-(s)$ and $\rho_+(s)$ are two families of the block states expressed in terms of eigenstates of $\rho$ as:
\begin{widetext}
    \begin{equation}
    \begin{aligned}
        \rho_\mp(s) = & \lambda^{\mp}_{s\mp 1/2}\sum_{m=-s}^s\left[\sum_{\beta}\ketbra{\Psi_{\text{I}}(\lambda^\mp_{s\mp 1/2};m,\beta)} +\sum_{\beta}\ketbra{\Psi_{\text{II}}(\lambda^\mp_{s\mp 1/2};m,\beta)} \right],
    \end{aligned}
\end{equation}
\end{widetext}
where
\begin{equation}
    \rho\ket{\Psi_{\text{I(II)}}(\lambda^\mp_{s\mp 1/2};m, \beta)} = \lambda^\mp_{s\mp 1/2}\ket{\Psi_{\text{I(II)}}(\lambda^\mp_{s\mp 1/2};m, \beta)},
\end{equation}
and the above states are explicitly written in Appendix~\ref{App_noiseless}. The spin index $s$ is implicit there in terms of $j$, $m=-s,\dots ,s$ is the usual magnetic number for spin $s$, and the index beta $\beta$ corresponds to the irrep multiplicity and characterizes degeneracy of the corresponding eigenavlues: For $\lambda^-_{s-1/2}$ the degeneracy is $g^{[N-1]}(s)$ and $g^{[N-1]}(s-1)$ for the eigenstates $\ket{\Psi_{\text{I}}(\lambda^-_{s- 1/2};m, \beta)}$ and $\ket{\Psi_{\text{II}}(\lambda^-_{s- 1/2};m, \beta)}$, respectively. Similarly, the degeneracy of $\lambda^+_{s+1/2}$ is $g^{[N-1]}(s+1)$ and $g^{[N-1]}(s)$ for the eigenstates $\ket{\Psi_{\text{I}}(\lambda^+_{s+1/2};m, \beta)}$ and $\ket{\Psi_{\text{II}}(\lambda^+_{s+1/2};m, \beta)}$, respectively, where    
\begin{equation}
    g^{[N]}(s) = \frac{N!(2s+1)}{(N/2-s)!(N/2+1+s)!}.
\end{equation}
The noiseless POVMs can be now constructed from $\rho$ as:
\begin{align}
  \forall_{1\leq i \leq N} \quad  \Pi_i  = \rho^{-1/2}\sigma^{(i)}\rho^{-1/2},
\end{align} 
where we have dropped the $1/N$ normalization factor and the inverse is on the support of $\rho$. 
Using \eqref{P_tilde_1} and \eqref{eta_i}, we can write down the success probability of detecting the noisy state $\eta^{(i)}_{\Vec{A}B}$ by the noiseless POVM defined above. This is the main objective of study in this Section:
\begin{widetext}
    \begin{align}
\label{noiseless_succ_prob}
        \forall_{1\leq i \leq N} \quad \tr\Pi^{(i)}_{\Vec{A}B}\eta^{(i)}_{\Vec{A}B} & = \frac{1+\vert\Gamma\vert\cos{\theta}}{2}\tr\Pi^{(i)}_{\Vec{A}B}\sigma^{(i)}_{\Vec{A}B} + \frac{1-\vert\Gamma\vert\cos{\theta}}{2}\tr\Pi^{(i)}_{\Vec{A}B}{\omega}^{(i)}_{\Vec{A}B} \\
        &+ \frac{i\vert\Gamma\vert\sin{\theta}}{2}\tr\Pi^{(i)}_{\Vec{A}B}\left[\left(\ketbra{\Psi^+}{\Psi^-}-\ketbra{\Psi^-}{\Psi^+}\right)_{A_i  B}\otimes\frac{\mathds{1}_{\bar{A}_i}}{2^{N-1}}\right] \label{mixedterm}
\end{align}
\end{widetext}

The first term above is the original Ishizaka-Hiroshima, calculated in \cite{ishizaka_asymptotic_2008,ishizaka_quantum_2009}. The second term differs in that instead of $\ket{\Psi^-}$, the ensemble states are constructed using $\ket{\Psi^+}$, cf. \eqref{omega_i}, and the last term is a mixed term. This allows to easily diagonalize all the operators in \eqref{noiseless_succ_prob}.  Building on the original works, let us introduce two families of states:
\begin{equation}
    \forall_{1\leq i \leq N} \quad\ket{\xi_{\mp}^{(i)}\left(s,m,\beta\right)} = \ket{\Psi^{\mp}}_{A_i B}\ket{\Phi^{[N-1]}(s,m,\beta)}_{\bar{A}_i},
\end{equation}
where $\ket{\Phi^{[N-1]}(s,m,\beta)}_{\bar{A}_i}$ are appropriately chosen basis states in the space of $N-1$ qubits. States $\ket{\xi_{-}^{(i)}\left(s,m,\beta\right)}$ were shown to diagonalize $\sigma^{(i)}_{\Vec{A}B}$ \cite{ishizaka_asymptotic_2008,ishizaka_quantum_2009}, and by the same argument states $\ket{\xi_{+}^{(i)}\left(s,m,\beta\right)}$ diagonalize ${\omega}^{(i)}_{\Vec{A}B}$. We note that $\ket{\xi_{+}^{(i)}\left(s,m,\beta\right)}$ vectors are just the $\sigma_z$ rotated Ishizaka-Hiroshima vectors (see Secion II in~\cite{ishizaka_quantum_2009}), where the rotation can be conveniently applied to the qubit $B$:  $\ket{\xi_{+}^{(i)}\left(s,m,\beta\right)} =\sigma_z^B \ket{\xi_{+}^{(i)}\left(s,m,\beta\right)}$, cf. \eqref{P_tilde}. To complete the calculation, we need the following overlaps, which can be calculated  for every $1\leq i \leq N$:
\begin{widetext}
    \begin{align}
        \braket{\xi^{(i)}_{+}(s,m,\beta)}{\Psi_{\text{I}}\left(\lambda^-_{s-1/2};m^\prime,\beta^\prime\right)} & = \frac{m}{\sqrt{s(2s+1)}} \delta_{m,m^\prime}\delta_{\beta, \beta^\prime} \\
        \braket{\xi^{(i)}_{+}(s,m,\beta)}{\Psi_{\text{II}}(\lambda^{+}_{s+1/2};m^\prime,\beta^\prime)} & = \frac{m}{\sqrt{(s+1)(2s+1)}} \delta_{m,m^\prime}\delta_{\beta, \beta^\prime}, \\
    \end{align}
\end{widetext}
the other two being zero. We can now evaluate \eqref{noiseless_succ_prob}, using the fact that due to the unitary symmetry, every object has a block diagonal structure like \eqref{rhos}. It is thus enough to calculate $ \tr\Pi^{(i)}_{\Vec{A}B}\eta^{(i)}_{\Vec{A}B}$ in the subspace of a fixed total spin $s$. Dropping the qubit indices for simplicity we obtain for the second term of \eqref{noiseless_succ_prob}:
\begin{widetext}
    \begin{align}
       \forall_{1\leq i \leq N} \quad  \tr\Pi^{(i)}{(s)}\omega^{(i)}(s) & = \tr\rho(s)^{-1/2}\sigma^{(i)}(s)\rho(s)^{-1/2}\omega^{(i)}(s) \label{eq:48}\\
        & = \frac{1}{2^{2N-2}}\cdot \frac{1}{3} g^{[N-1]}(s) \frac{s(s+1)}{2s+1} \left[\left(\lambda^-_{s-1/2}\right)^{-1/2} -\left(\lambda^+_{s+1/2}\right)^{-1/2}\right]^2, \label{eq:49}
\end{align}
\end{widetext}
and last term in \eqref{mixedterm}, which contains $\ketbra{\Psi^+}{\Psi^-}$. Vanishing  term~\eqref{mixedterm} and expression~\eqref{eq:49} are derived in Appendix~\ref{App_noiseless}. Finally, putting all the steps together, we obtain the following entanglement fidelity: 

\begin{align}
        F = \frac{1+\vert\Gamma\vert\cos\theta}{2}F_{IH} + \frac{1-|\Gamma|\cos\theta}{2}F_{\mathrm{corr}} \label{ent_fid_NL}.
    \end{align} 
In the above expression $F_{IH}$ is the Ishizaka-Hiroshima entanglement fidelity \cite{ishizaka_asymptotic_2008,ishizaka_quantum_2009}:
\begin{equation}
\label{eq:F_ih_noiseless}
    F_{IH} = \frac{1}{2^{N+3}}\sum_{k=0}^N\left(\frac{N-2k-1}{\sqrt{k+1}}+\frac{N-2k+1}{\sqrt{N-k+1}}\right)^2\binom{N}{k}.
\end{equation}
For the reader's convenience the term $F_{IH}$ is calculated in Appendix~\ref{derivationPRA}.
The second term $F_{corr}$ in~\eqref{ent_fid_NL} is the new entanglement fidelity correction term:
\begin{widetext}
    \begin{equation}
\label{eq:F_corr_term0}
F_{\mathrm{corr}}=\frac{1}{3}\frac{1}{2^{N+3}}
\sum_{k=0}^{N}
\binom{N}{k}
\left[(N-2k)^2-1\right]
\left(
\frac{1}{\sqrt{k+1}}
-
\frac{1}{\sqrt{N-k+1}}
\right)^2.
\end{equation}
\end{widetext}
The large-\(N\) behaviour of \(F_{\mathrm{corr}}\) can be understood
without going through the full derivation given in
Appendix~\ref{app:F_corr_term}. Indeed, the binomial sum in
Eq.~\eqref{eq:F_corr_term0} may be interpreted as an expectation with
respect to a random variable \(K\sim \mathrm{Bin}(N,1/2)\), with the
overall prefactor \(1/24\). Writing \(K=N/2+X\), the relevant values of
\(X\) are typically of order \(\sqrt{N}\). In this central region we have
\begin{equation}
    N-2K=-2X,
\end{equation}
and
\begin{equation}
    \frac{1}{\sqrt{K+1}}
    -
    \frac{1}{\sqrt{N-K+1}}
    =
    -\frac{X}{(N/2+1)^{3/2}}
    +
    O\!\left(\frac{|X|^3}{N^{7/2}}\right).
\end{equation}
Consequently, the summand in Eq.~\eqref{eq:F_corr_term0} is governed, to
leading order, by
\begin{align}
    &\bigl[(N-2K)^2-1\bigr]
    \left(
    \frac{1}{\sqrt{K+1}}
    -
    \frac{1}{\sqrt{N-K+1}}
    \right)^2  \\
    &=
    \frac{32X^4}{N^3}
    +
    O\!\left(
    \frac{X^2}{N^3}
    +
    \frac{X^6}{N^5}
    \right).
\end{align}
Using the fourth central moment
\begin{align}
    \mathbb{E}[X^4]
    =
    \frac{3N^2}{16}
    +
    O(N),
\end{align}
we obtain
\begin{align}
    F_{\mathrm{corr}}(N)
    &=
    \frac{1}{24}\,
    \frac{32}{N^3}\,
    \mathbb{E}[X^4]
    +
    O\!\left(\frac{1}{N^2}\right)  \\
    &=
    \frac{1}{4N}
    +
    O\!\left(\frac{1}{N^2}\right).
\end{align}
Thus the correction term vanishes in the limit \(N\to\infty\), while
the Ishizaka--Hiroshima contribution remains the dominant term.

The above derivation show that the correction term vanishes when $N\rightarrow \infty$ and the total entanglement fidelity~\eqref{ent_fid_NL} reads $F=\frac{1}{2}(1+|\Gamma|\cos \theta)$ having maximum value 1 for $\theta=0$ and $|\Gamma|=1$ (no decoherence). Notice, that the function $F_{\mathrm{corr}}$ is not monotonic for small values of $N$. We observe monotonicity for larger values of $N$. A more detailed discussion on behaviour of $F_{\mathrm{corr}}$ is presented in Appendix~\ref{app:F_corr_term}.  

\begin{figure}[t]
    \centering \includegraphics[width=0.9\linewidth]{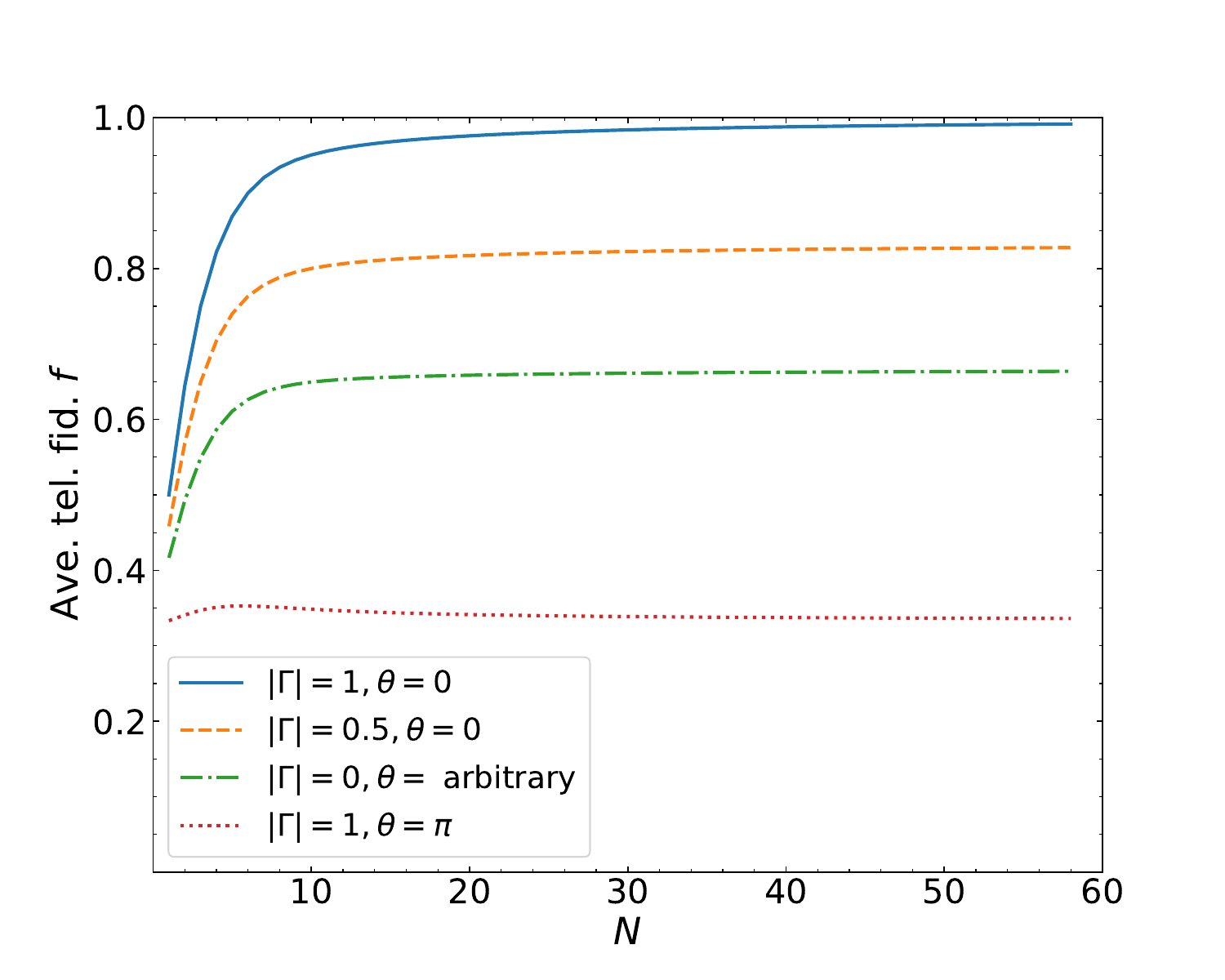}
    \caption{
        Teleportation fidelity 
        $f=(2F+1)/3$ as a function of $N$ for four various choices of the noise parameters $(|\Gamma|,\theta)$. The blue line reproduces the noiseless case from the Hiroshima-Ishizaka protocol.
    }
    \label{fig:fidelity_contour_N9}
\end{figure}

A plot of the total resulting teleportation fidelity, $f = (2F+1)/3$, as a function of $|\Gamma|$ and $\theta$ is presented in Fig.~\ref{3d_plot} for $N=9$ (there is no qualitative change for different $N$). The no decoherence case corresponds on this plot to the corner $|\Gamma|=1$ and $\theta=0$. There is a pronounced minimum for $\theta=\pi$, which corresponds to the zeroing of the Ishizaka-Hiroshima term in \eqref{ent_fid_NL}, which is the dominating term. The maximum, of the other hand is obtained for $\theta=0$, i.e. for a real decoherence factor $\Gamma$. This maximum can be obtained by a semi-adaptive procedure, where instead of using the original PGM, one rotates them by $\theta$: $\Pi_i(\theta)=R^\dagger(\theta)_B \Pi_i^{\vec{A} B} R(\theta)_B$. One can the easily show, starting from \eqref{P_tilde}, that this leads to the fidelity \eqref{ent_fid_NL} with $\theta=0$. Of course, this strategy assumes a partial knowledge about the noise, the phase of the decoherence factor $\theta$, and the ability to use it for the POVM optimization. As we mentioned, for $\theta =0$, the above strategy reduces to a special case of the Pauli noise channel, cf. \eqref{single_q_channel}, studied in \cite{Kim}.  Importantly, this qualitative picture persists in the asymptotic regime of
large~$N$.
While the Ishizaka--Hiroshima contribution approaches unity as
$N \to \infty$, the correction term $F_{\mathrm{corr}}$ vanishes as
$\mathcal{O}(1/N)$. This is illustrated in Figure~\ref{fig:fidelity_contour_N9}.
 The fidelity obtained from Eq.~\eqref{ent_fid_NL} differs from that of Ref.~\cite{Kim}: it is always lower, by at most approximately \(10\%\) for \(\Gamma=0\), and becomes equal only for \(\Gamma=1\), where all fidelities reduce to the Ishizaka--Hiroshima value. A detailed comparison with the approach of Kim and Jeong~\cite{Kim} is presented in Appendix~\ref{app:Kim}.

To clarify the behaviour of our expression in the many-port regime, and
to make explicit the part of the fidelity loss which cannot be removed
by increasing \(N\), we now analyse the large-\(N\) limit of Eq.~\eqref{ent_fid_NL}.
Using Eq.~\eqref{ent_fid_NL} with
\[
    \alpha:=\frac{1+|\Gamma|\cos\theta}{2},
\]
together with the asymptotic expansions
\begin{align}
    F_{\mathrm{IH}}(N)
    &=
    1-\frac{c_{\mathrm{IH}}}{N}
    +
    O(N^{-2}),
    \\
    F_{\mathrm{corr}}(N)
    &=
    \frac{1}{4N}
    +
    O(N^{-2}),
\end{align}
where $c_{IH}=3/4$ (see. Eq.~\eqref{F_IH_inf}), we obtain, for fixed decoherence parameter \(\Gamma=|\Gamma|e^{i\theta}\),
\begin{align}
    F(N;\Gamma,\theta)
    =
    \alpha
    -
    \frac{\alpha c_{\mathrm{IH}}-\frac{1}{4}(1-\alpha)}{N}
    +
    O(N^{-2}).
\end{align}
Hence,
\begin{align}
    \lim_{N\to\infty}F(N;\Gamma,\theta)
    =
    \frac{1+|\Gamma|\cos\theta}{2}.
\end{align}
This limit holds for any pure-dephasing model of the form considered in
Sec.~\ref{Sec:noisyPBT_model}, provided that the microscopic noise parameters are
kept fixed as \(N\to\infty\).  Equivalently, in the spin--boson realization of
Sec.~\ref{Sec:spin-boson}, the parameters \(\tau,\ell,T/\Lambda\), and \(s\) are
fixed, so that \(\Gamma(\tau,\ell)\) is fixed while the number of ports is increased.
If the microscopic noise parameters themselves scale with \(N\), then the limit should
instead be evaluated along the corresponding sequence \(\Gamma_N\).

This shows that increasing the number of ports removes only the finite-\(N\) PBT
error.  It cannot remove the decoherence-induced ceiling
\begin{align}
    F_\infty(\Gamma,\theta)
    =
    \frac{1+|\Gamma|\cos\theta}{2}.
\end{align}
Therefore, a target entanglement fidelity \(F_\ast\) can be reached only if
\begin{align}
    F_\ast < F_\infty(\Gamma,\theta).
\end{align}
Within the leading-order approximation, the required number of ports is estimated by
\begin{align}
    N
    \gtrsim
    \frac{
        \alpha c_{\mathrm{IH}}-\frac{1}{4}(1-\alpha)
    }{
        \alpha-F_\ast
    },
\end{align}
whenever the numerator is positive.  Thus the port overhead diverges as
\begin{align}
    N=O\!\left(\frac{1}{\alpha-F_\ast}\right)
\end{align}
when the target fidelity approaches the noise-dependent ceiling.  If the numerator is
non-positive, then the first-order correction does not impose an additional large-\(N\)
lower bound, and the target fidelity condition is governed mainly by
\(F_\ast<\alpha\), together with subleading corrections.

\begin{figure*}[t]
\centering

\begin{minipage}{0.51\linewidth}
    \centering    \includegraphics[width=\linewidth]{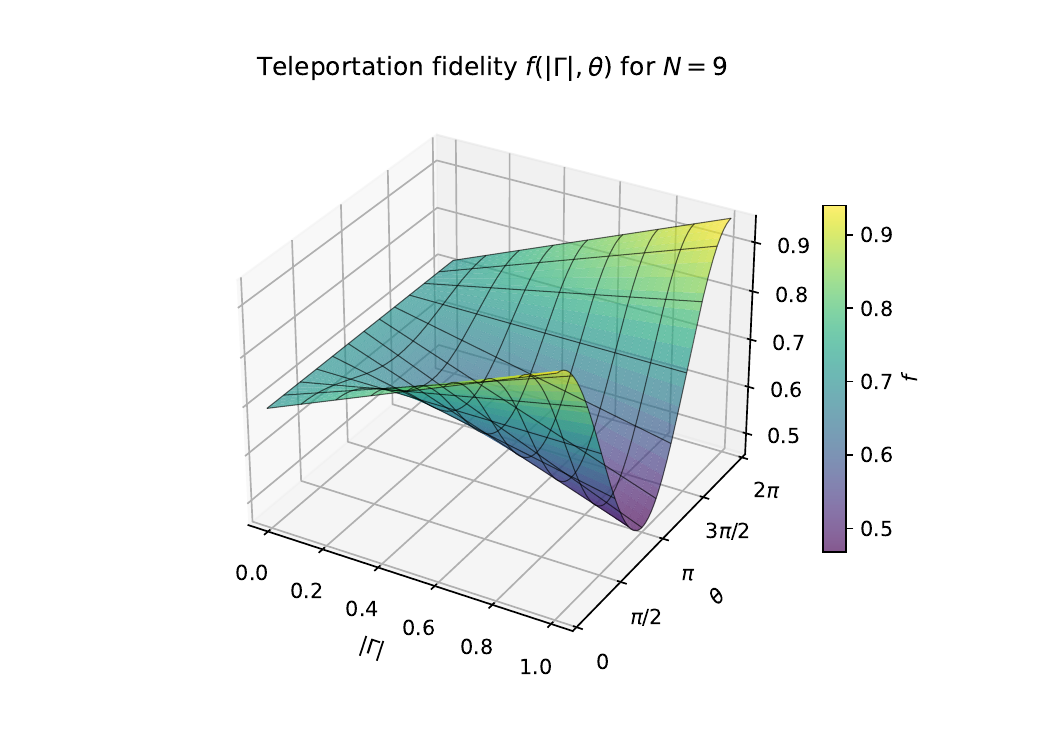}
    
    \vspace{-2mm}
    {\small (a)}
\end{minipage}
\hfill
\begin{minipage}{0.48\linewidth}
    \centering    \includegraphics[width=\linewidth]{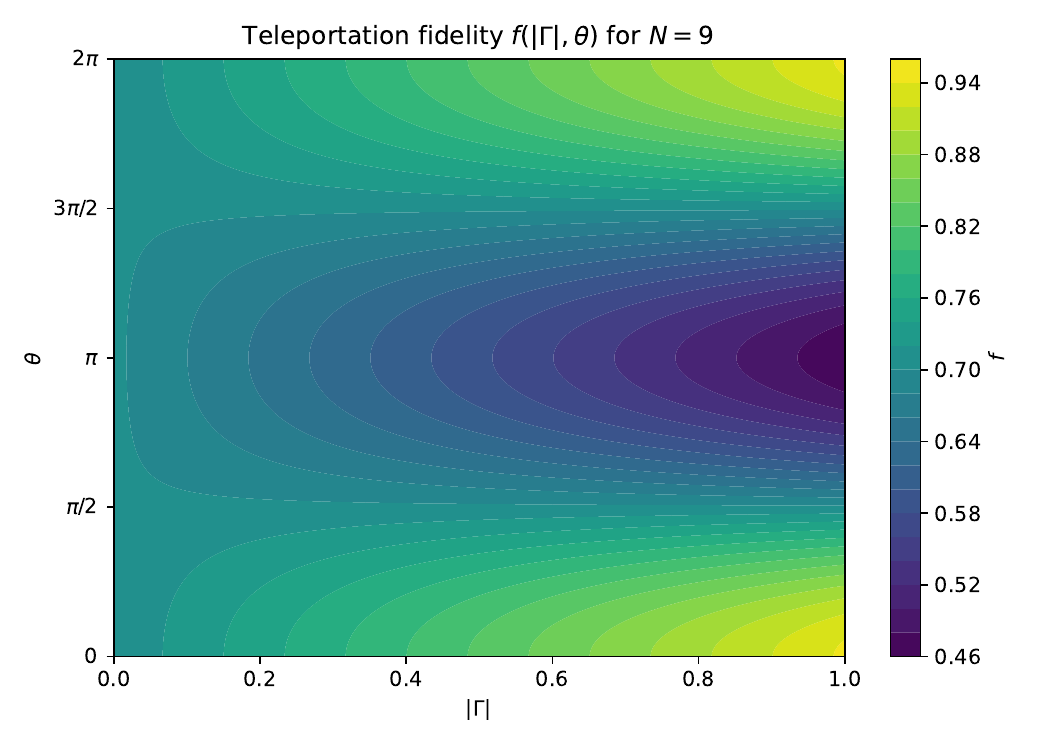}
    
    \vspace{-2mm}
    {\small (b)}
\end{minipage}
\caption{The average teleportation fidelity for $N=9$, calculated using the entanglement fidelity  Eq.~\eqref{ent_fid_NL}, is plotted against $|\Gamma|$ and $\theta$. The results are shown as
(a) a three-dimensional surface plot and (b) a contour plot.}
\label{3d_plot}
\end{figure*}

\section{PBT using noise-adapted measurements}
\label{noisyPOVMs}

The final, and the most complicated, scenario is that of discriminating measurements adapted to the noisy signal states \eqref{eta_i}. The natural choice here is of course again the PGM, but this time constructed from the noisy states $\eta_{\vec{A} B}^{(i)}$. Following the general PGM recipe  \eqref{PGM}, one first has to find the average noisy ensemble state (we keep using un-normalized states):
\begin{align}\label{noisyav}
    \bar\eta= \sum_{i=1}^N \eta_{\vec{A} B}^{(i)}.
\end{align}
Unlike for the noiseless average \eqref{avrho}, there does not seem to be an obvious diagonalizing procedure for this state. First, the unitary symmetry of the noiseless case is lost whenever the decoherence is present, since the $A_iB$ part of each $\eta_{\vec{A} B}^{(i)}$, given by $\tilde{P}^-_{A B}$, cf. \eqref{dephasing_channel}, \eqref{eta_i}, does not have a definite symmetry under the product action of $SU(2)$, as we have explained below \eqref{PGM}. Second, although each $\eta_{\vec{A} B}^{(i)}$ can be easily diagonalized, cf. \eqref{diag}, they overlap on the qubit $B$ making the diagonalization of the sum non-trivial. We thus have to resort to the numerics.  We have employed in total three numerical methods, specifying to the case $\theta=0$ expecting this is the best case scenario, cf. Fig.~\ref{3d_plot}, and cross checking between the methods whenever possible:

(1) In the first method, the signal states $\eta_{\vec{A} B}^{(i)}$ and the ensemble average \eqref{noisyav} were constructed with a python program. 
The noisy PGM POVMs $\{\Pi_{i}\}$ were generated using the eigensolver of numpy library \cite{harris2020numpy}, where we first diagonalized \eqref{noisyav}, discarded the zero eigenvalues, then constructed $\bar\eta^{-1/2}$ and from this the POVMs, checking numerically that the obtained operators were positive. The final POVM set also included the projector on the zero-eigenvalue subspace completeness. We finally computed \eqref{F} and from it the teleportation fidelity \eqref{teleportfid}. 

(2) In the second method, the starting point was the same, but the $\bar\eta^{-1/2}$ was computed differently, using the Taylor expansion:
\begin{align}
    \bar{\eta}^{-1/2} = \mathds{1}-\frac{1}{2}(\bar{\eta}-\mathds{1})+\frac{3}{8}(\bar{\eta}-\mathds{1})^2
    -\frac{5}{16}(\bar{\eta}-\mathds{1})^3+\cdots
\end{align}
and including the terms up to the order of 4000. The non-zero eigenspace gets perfectly inverted while the zero eigenspaces obviously diverge with the increased number of iterations. However, those terms disappear when the POVMs are constructed as the signal states $\eta^{(i)}$ turned out to be supported on subspaces orthogonal to the kernel of $\bar{\eta}$, which is rather a non-trivial observation, revealed by the numerics. The POVMs were again checked for the positivity, which was not an issue within the numerical precision due to the very high order of the expansion.

(3) The third method was a direct use of the Wolfram Mathematica program, which was feasible only for the simplest case of $N=2$. We inserted the signal states by hand and then used the Mathematica eigensolver to diagonalize $\bar\eta$. From this, the POVMs were constructed and the teleportation fidelity computed.   

All the three methods gave the same results within the numerical precision and whenever applicable. As a final cross-check, all the methods we also applied to the noiseless case of Section \eqref{Sec noiseless}  and the confirmed the Eq.~\eqref{ent_fid_NL} for $\theta=0$. The results are presented in Fig.~\ref{fid_2_5_9}(a) for $N=2,5,9$ together with the noiseless curves and the Beigi-Koenig bound \eqref{BKbound} for $N=9$ for a comparison.  Due to the exponential increase of the dimension, $d = 2^{N+1}$, $N=9$ was the maximum available within a normal laptop capabilities but it turned out to be perfectly enough to reveal the behaviour of the considered strategies.

The motivation behind the comparative analysis presented in Fig.~\ref{fid_2_5_9}(a) is to determine whether a noise-adapted PBT protocol, employing the PGM associated with the noisy signal ensemble, can achieve a higher teleportation fidelity than the original deterministic PBT protocol, whose POVM is constructed under the assumption of a noiseless shared resource state. Such an advantage might naturally be anticipated from a conventional viewpoint, since the adapted measurement explicitly incorporates the effects of noise present in the signal ensemble. Additionally, the comparison is well motivated from an operational standpoint as well: while the steps, such as encoding, transmission, measurements and correction procedures, of a quantum communication protocol are typically predetermined and standardized, the channels used to distribute the entangled resource states are generally exposed to environmental noise and other imperfections.

Quite surprisingly, the noise-adapted PGM performs worse than the noiseless PGM apart from $N=2$. The noise-adapted strategy is better only in the region of strong decoherence, close to $\Gamma=0$, but even that region is shrinking with increasing $N$, which is a bit of a surprise in its own right. The optimality of PGM for PBT \cite{leditzky2020optimality} was shown only for the ideal case of uncorrupted Bell states, using the unitary symmetry. The optimality criteria for the PGM in the case of mixed-state signal ensembles with 
$N>2$ are highly nontrivial. Nevertheless, among the known results, there exist optimality conditions for geometrically uniform ensembles, in which the signal states can be generated from a reference state through the action of unitary operators belonging to a symmetry group. In this setting, the PGM is known to be optimal for pure state ensembles and for mixed states with some additional condition (Theorem 3 of \cite{eldar2004optimal}; see also \cite{PhysRevA.81.012315}).
It is therefore not surprising that the PGM may cease to be optimal in the presence
of decoherence.  What is surprising, however, is that it is beaten by the noiseless
PGM, which is completely agnostic to the changes in the ensemble introduced by the
decoherence. Apart from geometrically uniform ensemble based criterion, there exists a more general Gram-matrix-based optimality criterion for PGM \cite{ItenRenesSutter}. A PGM is optimal if and only if the Gram-matrix, $G$, of the ensemble commutes with block-diagonal operator constructed from the diagonal blocks of $\sqrt{G}$. However, due to the complexity and the size of the ensemble, establishing such general results for our PBT scenario stays outside the scope of this paper.

As a remark, there is of course a case when the optimal, noisy measurement is known: For $N=2$ there is the famous Helstrom optimality bound \cite{helstrom1976quantum}, given in the terms of the trace distance between the states $P_{\rm succ}\leq 1/2+(1/4)\Vert\eta^{(1)}-\eta^{(2)}\Vert_1$. A direct calculation shows that (see Appendix~\ref{Helstrom}):
\begin{align}
    \Vert\eta^{(1)}-\eta^{(2)}\Vert_1 = \sqrt{1+2|\Gamma|^2}, 
\end{align}
and therefore the entanglement fidelity is upper-bounded for $N=2$ by:
\begin{align}
\label{eq:Hel_bound}
    F\leq\frac{1}{4}\left(1+\frac{\sqrt{1+2|\Gamma|^2}}{2}\right).
\end{align}
The plot of this bound together with the PGM fidelity for both noiseless and noise-adapted scenarios is presented in Fig.~\ref{fid_2_5_9}(b). There is a substantial gain w.r.t. the noiseless measurement and only a slight one w.r.t. the noise-adapeted one. 

\begin{figure*}[t] 
	\centering
	\begin{minipage}{0.48\textwidth}
		\centering    \includegraphics[width=\linewidth]{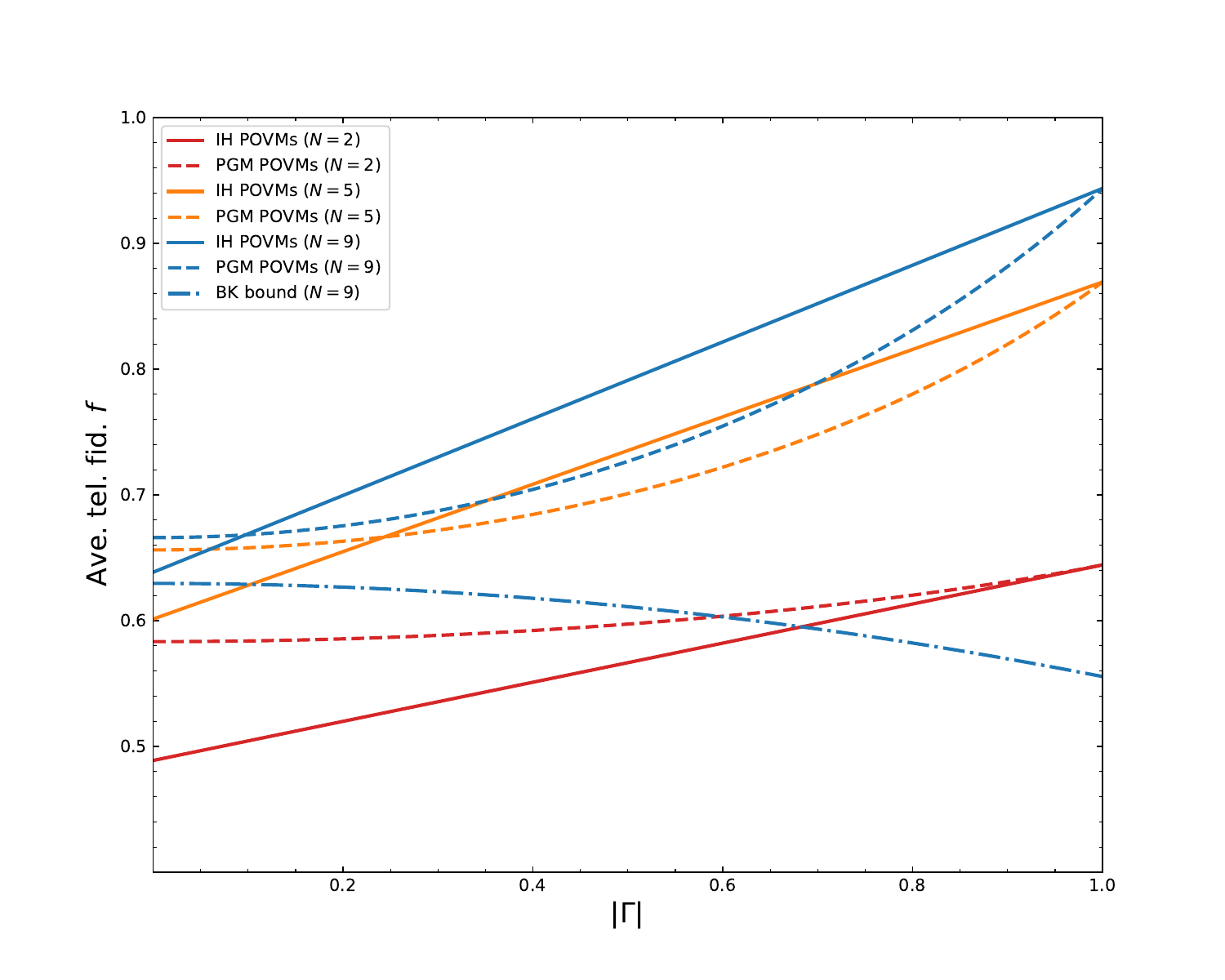} 
		\\[4pt] (a)
	\end{minipage}\hfill
	\begin{minipage}{0.48\textwidth}
		\centering		\includegraphics[width=\linewidth]{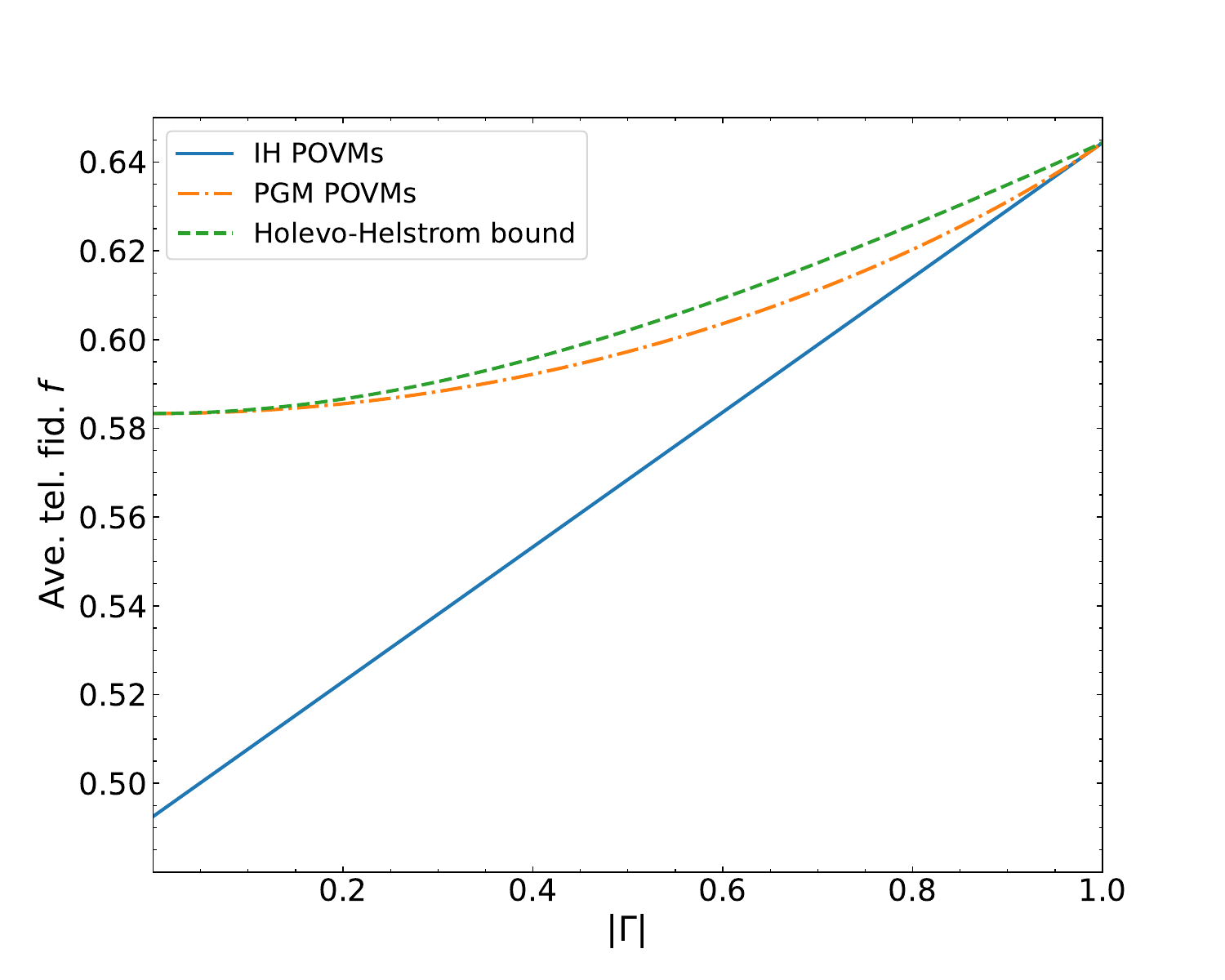} 
		\\[4pt] (b)
	\end{minipage}
		\caption{(a) The average teleportation fidelities for noiseless Ishizaka--Hiroshima POVMs (solid lines), computed using Eq.~\eqref{ent_fid_NL}, and the average teleportation fidelities for PGMs (dashed lines), computed using the Python eigensolver program, are plotted for \(N=2,5,9\). The Beigi--K\"onig bound for \(N=9\) is plotted with a blue dash-dotted line. (b) Average teleportation fidelities for noiseless POVMs (blue solid line), PGMs (orange dash-dotted line), and the maximum fidelity allowed by the Helstrom bound (green dashed line) are plotted for \(N=2\).}
	\label{fid_2_5_9}
\end{figure*}

At this stage, it is instructive to relate the strongly decohered regime, \textit{i.e.}~\( |\Gamma| \rightarrow 0 \) for arbitrary \(\theta\), to the standard teleportation benchmark. According to the Horodecki criterion~\cite{HorodeckiMPRFidelity}, a qubit teleportation protocol is nonclassical only when the average teleportation fidelity exceeds the classical threshold \(f>2/3\), corresponding to an entanglement fidelity \(F>1/2\) in Eq.~\eqref{teleportfid}. In the strong decoherence regime $|\Gamma|=0$, Eq.~\eqref{ent_fid_NL} gives \(F = (F_{IH}+F_{corr})/2\), implying \(F \leq 1/2\) for all \(N\), see Fig.\ref{fig:fidelity_vs_N_Fcorr}(b). Hence, the shared resource ceases to be useful for quantum teleportation. Equivalently, when \( |\Gamma|=0 \), the shared bipartite state becomes separable (see Eq.~\eqref{strong_decoh}), indicating that the decoherence channel is entanglement breaking~\cite{ent_breaking_channels}. Consequently, the teleportation fidelity reduces to, or falls below, the classical threshold.

The observation that the noise-adapted PGM is generally outperformed by the noiseless PGM for finite $N>2$, except in a narrow region of strong decoherence, is somewhat surprising and deserves additional emphasis. While we do not currently have a definitive explanation for this behavior, one plausible interpretation is related to the symmetry structure of the signals and the ensemble states. In the absence of noise, signals and the ensemble possess a high degree of unitary symmetry, and the corresponding PGM is supported on the associated symmetry-protected subspace which is relevant for the teleportation. Decoherence partially breaks these symmetries and enlarges the support of the noisy signal states, leading to a modified PGM adapted to a broadened noise-dominated support which might not necessarily be useful for the teleportation. At present, a more complete understanding of this phenomenon remains an interesting open question.

\section{Case study: Spin-boson model}
\label{Sec:spin-boson}

In this section, we provide a concrete realization of the pure-dephasing model used throughout the paper. We will use the well-known spin-boson model, which is one of the standard models describing decoherence effects in qubits. We will assume the self-Hamiltonian of the qubits is suppressed, which is motivated by the fact that one would like to produce Bell states for each pair and keep the qubits in those states. We consider Hamiltonian of the form~\cite{Tuziemski2018}
\begin{equation}   H=\sum_{m=A,B}\frac{\sigma_z^{(m)}}{2}\otimes \sum_{k}(g_{mk}b_k^\dagger+g_{mk}^*b_k) + \sum_k H_k,
\end{equation}
where $\sigma_z^{(m)}$ are the $m$-th qubit $z$-axis Pauli matrices, $b_k,b_k^\dagger$ are the $k$-th mode annihilation and creation operators, $g_{mk}$ are complex coupling constants, and  $H_k=\omega_k b_k^\dagger b_k$. One can easily check it is of the form \eqref{Hint} with  $V_{ij}=\sum_k(\epsilon_i g_{1k} b_k^\dagger +\epsilon_j g_{2k}b^\dagger_k+h.c.+H_k)$, where $\epsilon_i=\mp1/2$ for $i=0,1$ respectively.
We assume position-dependent coupling to the bath,
depending explicitly on the spatial separation $r$ between
the qubits forming each Bell pair, i.e.  ${g}_{mk}=g_ke^{-kr_m}$, $r=|r_A-r_B|$.  This describes e.g. the interaction with the electromagnetic field in the dipole approximation. All the information about the interaction with the environment is encoded in the complex decoherence parameter:
\begin{equation}
\Gamma(t,r)=|\Gamma(t,r)|e^{i\theta(t,r)}.
\end{equation}

We further assume the environment to be thermal, which is the most physically relevant example, i.e. the initial state of the  environment is $R(0)=1/Ze^{-\beta\sum_k H_k}$, with the inverse temperature $\beta$. 
Following the standard procedure, we incorporate the coupling constants $g_k$ into a single spectral density function. The most popular choice of the latter is the generalized Ohmic spectral density with an exponential cutoff:
\begin{equation}
J(\omega)=\frac{\omega^s}{\Lambda^{s-1}}e^{-\omega/\Lambda},
\end{equation}
where $\Lambda$ is a high-frequency cutoff and s is the so-called Ohmicity parameter. Interestingly, the latter is known to control the Markovianity properties of the model \cite{PhysRevLett.103.210401}: For $s\leq 2$ the model is Markovian, while for $s \geq 3$ it is not.  Physically, this means that for $s\leq 2$ the state of the qubits at each moment $t$ depends only on that moment. In the language of correlations, this implies that the correlations buildup between the qubits and the  environment monotonically destroys the initial coherences. One imagines this a monotonic flow of certain quibit-related information out into the environment. For $s \geq 3$, however, this picture no longer holds. The state of the qubits acquires a dependence on their previous evolution, indicating the presence of memory effects. As a result, the coherences no longer decay monotonically; in addition to their suppression, revivals of coherence may occur. From an information-theoretic perspective, this behavior can be interpreted as a bidirectional exchange of information between the qubits and the environment, with information flowing back from the environment to the system.

Introducing
dimensionless variables
\begin{equation}
\tau = t\Lambda,
\qquad
\vartheta = \frac{T}{\Lambda},
\end{equation}
it is convenient to parametrize the separation dependence by the dimensionless distance
parameter
\begin{equation}
\ell = \frac{r\Lambda}{c}.
\end{equation}
In terms of these variables, the decoherence factor takes the form
\begin{equation}
\Gamma(\tau,\ell)
=
\exp\!\left[-\chi(\tau,\ell)+i\theta(\tau,\ell)\right],
\label{eq:64}
\end{equation}
where the explicit expressions for $\chi(\tau,\ell)$ and $\theta(\tau,\ell)$ are derived in
Appendix~\ref{app:spin-boson} as,
\begin{widetext}
    \begin{align}
\chi(\tau,\ell)
&=
2\int_0^\infty d\tilde\omega\;
\tilde\omega^{\,s-2} e^{-\tilde\omega}\,
\bigl(1-\cos(\tilde\omega\tau)\bigr)\,
\coth\!\left(\frac{\tilde\omega}{2\vartheta}\right)\,
\bigl[1-\cos(\tilde\omega\ell)\bigr],
\label{eq:H12_newa}
\\[1ex]
\theta(\tau,\ell)
&=
\frac{1}{2}\int_0^\infty d\tilde\omega\;
\tilde\omega^{\,s-2} e^{-\tilde\omega}\,
\bigl(1-\cos(\tilde\omega\tau)\bigr)\,
\sin(\tilde\omega\ell).
\label{eq:chi_theta}
\end{align}
\end{widetext}
The integrals can be calculated explicitly for $s>1$, see e.g. \cite{Tuziemski2018}. More details are recalled in Appendix~\ref{app:spin-boson}. We also recall that independently of the microscopic structure of the bath, the reduced state
of each Bell pair after tracing out the environment can always be written in the form
\begin{equation}
\tilde P^-_{AB}(t,r)
=
(\mathds{1}\otimes\mathcal E_{\Gamma(t,r)})
\bigl(P^-_{AB}\bigr),
\end{equation}
so that all results derived in Secs.~\ref{Sec noiseless} and~\ref{noisyPOVMs} remain valid after the substitution
$\Gamma\rightarrow\Gamma(t,r)$.

We first consider noiseless Ishizaka--Hiroshima
measurements as they can be analyzed analytically. 
The entanglement fidelity of the
teleportation channel then reads
\begin{equation}
\begin{aligned}
    F(\tau,\ell) &=
\frac{1+|\Gamma(\tau,\ell)|\cos\theta(\tau,\ell)}{2}F_{\mathrm{IH}} \\
& +
\frac{1-|\Gamma(\tau,\ell)|\cos\theta(\tau,\ell)}{2}F_{\mathrm{corr}},
\label{eq:65}
\end{aligned}
\end{equation}
and the corresponding teleportation fidelity is
\begin{equation}
f(\tau,\ell)=\frac{2F(\tau,\ell)+1}{3}.
\label{eq:66}
\end{equation}
We plot the resulting teleportation fidelity as a function of the dimensionless time
\(\tau=t\Lambda\) in Fig.~\ref{Spin_Boson_Ohmics=2,3}(a,b) for the noiseless
Ishizaka-Hiroshima measurements (solid lines).
Two representative values of the temperature, $T/\Lambda=0.1$ and $T/\Lambda=0.9$, are chosen to
contrast low- and high-temperature environments relative to the cutoff. Increasing $T/\Lambda$ enhances thermal decoherence, leading to a faster
suppression of the fidelity at long times, while the short-time behaviour is
only weakly affected.

\begin{figure*}[htbp]
    \centering
    \begin{subfigure}[b]{0.48\textwidth}
        \centering
        \includegraphics[width=\linewidth]{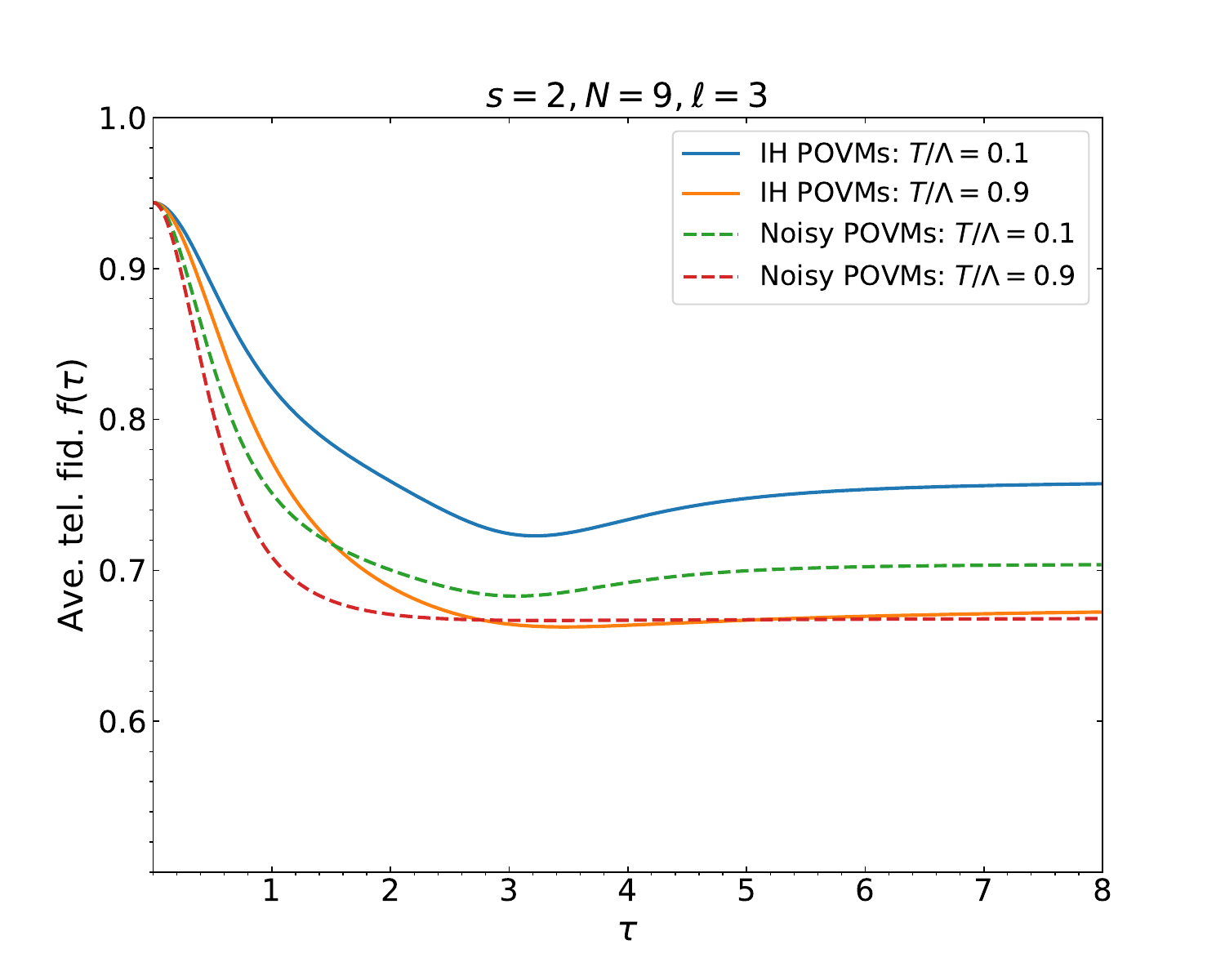}
        \caption{}
        \label{fig:sub1}
    \end{subfigure}
    \hfill
    \begin{subfigure}[b]{0.48\textwidth}
        \centering
        \includegraphics[width=\linewidth]{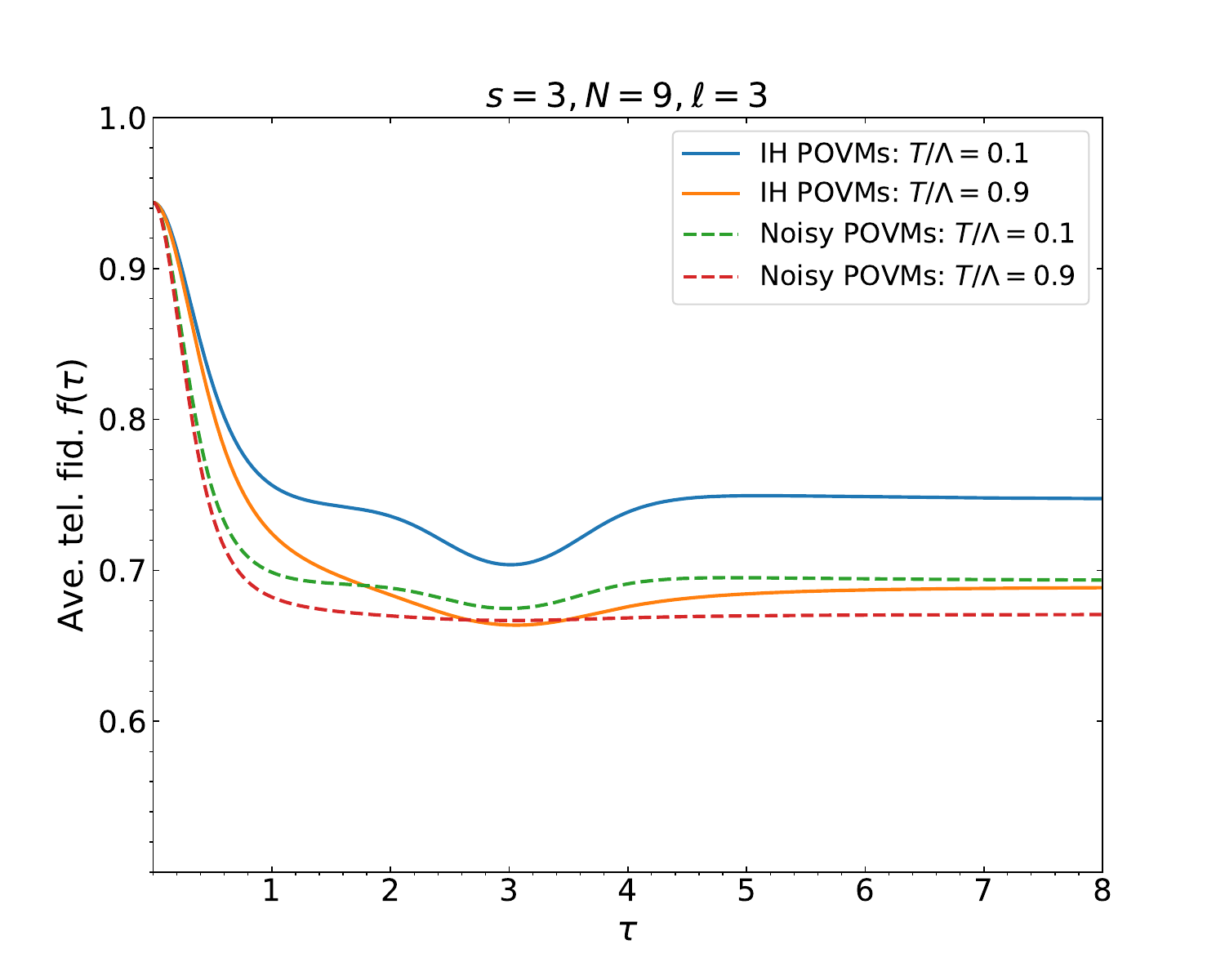}
        \caption{}
        \label{fig:sub2}
    \end{subfigure}

    \caption{Average teleportation fidelity $f(\tau)$ for $N=9$ as a function of the dimensionless time
    $\tau=t\Lambda$ for a generalized Ohmic bath with Ohmicity $s=2$ ($s=3$) and transit time
    $\ell=\Lambda\bar{t}=3$. Solid lines -- noiseless measurements; dashed lines -- noise-adapted measurements. Each curve corresponds to a fixed temperature ratio $T/\Lambda=0.1$ or $T/\Lambda=0.9$.}
    \label{Spin_Boson_Ohmics=2,3}
\end{figure*}

The case for noise-adapted measurements is studied numerically by taking $\Gamma$ as a function of $\tau$. Unlike Section~\ref{noisyPOVMs}, we consider the complex value of decoherence factor with $\chi(\tau,\ell)$ and $\theta(\tau,\ell)$ given in \eqref{eq:chi_theta}. The inverse square root of the ensemble state is computed using the eigensolver method. The corresponding results are also shown in Fig.~\ref{Spin_Boson_Ohmics=2,3}(a,b) for the noise-adapted measurements
(dashed lines). As we have predicted before, surprisingly, noise-adapted measurements perform worse than the noiseless ones, even with a large temperature difference.

We can see characteristic dips in the teleportation fidelity, so that it is not monotonically decreasing with time even for $s=2$. These dips are associated with a finite time-of-flight of the field quanta between the qubits' positions.  In the simulations shown in Fig.~\ref{Spin_Boson_Ohmics=2,3} the spatial separation is fixed to
$\ell=\Lambda r/c=3$. That distance enters the decoherence factor through the oscillatory terms in
\eqref{eq:H12_newa}\text{--}\eqref{eq:chi_theta}.
The dips are centered exactly at $\tau\simeq 3$, corresponding to the time scale at which the
distance-dependent phase $\theta(\tau,\ell)$ produces destructive interference in the coherence factor
$\Gamma(\tau,\ell)$ through the $\cos\theta(\tau,\ell)$ term.
Memory effects are reflected in the non-exponential short-time behaviour of $F(\tau,\ell)$, whereas the
long-time dynamics is controlled by the saturation of $\chi(\tau,\ell)$, explaining why increasing
temperature primarily lowers the asymptotic fidelity plateau observed in
Fig.~\ref{Spin_Boson_Ohmics=2,3}.

\section{Simulability of port-based teleportation with noisy resources}
\label{sec:simulability-noisyPBT}
A central motivation behind recent work on algorithmic implementations of port-based
teleportation~\cite{fei2023efficientquantumalgorithmportbased,Grinko2024,PRXQuantum.5.030354}
is that, although the square-root (pretty-good) measurement (SRM/PGM) is
information-theoretically optimal for distinguishing the PBT signal ensemble, its direct
implementation can be prohibitively expensive~\cite{PhysRevLett.128.220502,PRXQuantum.5.030354}.
In particular, generic Petz-type constructions may suffer from poor spectral conditioning
of the ensemble average operator, leading to exponential overheads unless additional
structure is exploited. Efficient implementations therefore rely crucially on symmetry
and representation-theoretic decompositions of the noiseless PBT ensemble. In this
subsection, we briefly discuss how noise in the PBT resource state affects the algorithmic
simulability of the protocol, focusing on the measurement stage performed by Alice.

\emph{Reference point: noiseless PBT---}
In noiseless deterministic PBT, the signal ensemble is covariant under both permutations of the ports and joint unitary transformations of the form
\( U^{\otimes N}_{\bar A}\otimes U^*_B \).
As a consequence, the ensemble average operator
\(
\bar\eta=\sum_{i=1}^N \eta^{(i)}
\)
lies in the double commutant of the natural actions of $S_N$ and $U(d)$.

By Schur--Weyl duality and joint $U^{\otimes N}\!\otimes U^*$ covariance, $\bar\eta$ decomposes as
\begin{equation}
\bar\eta
=
\bigoplus_{\lambda}
\bigl(
X_\lambda \otimes \mathds{1}_{\mathcal V_\lambda}
\bigr),
\end{equation}
acting trivially on the $S_N$ multiplicity spaces $\mathcal V_\lambda$.
Although the dimensions of the underlying representation spaces grow with the number of ports $N$, all nontrivial spectral structure of $\bar\eta$ is contained in the operators $X_\lambda$.

For fixed local dimension $d$, each $X_\lambda$ acts on
$\mathcal U_\lambda\otimes\mathbb C^d$ and decomposes into only finitely many irreducible components (at most $d$), independent of $N$.
Equivalently, $\bar\eta$ is diagonal in the Schur basis with eigenvalues depending only on coarse irrep labels and not on multiplicity indices.
This effective constant-size spectral structure—rather than the literal dimension of the representation spaces—is the representation-theoretic reason why $\bar\eta^{-1/2}$ can be evaluated efficiently and why noiseless PBT admits efficient quantum circuit implementations.

\emph{Noisy resources with noiseless PBT measurements---} We first consider the scenario in which the shared resource state is noisy, but Alice applies the same measurement as in noiseless PBT. In this case, the POVM itself is unchanged: noise affects only the quantum state entering the measurement. Algorithmically, this implies that any efficient circuit implementing the noiseless IH measurement remains valid without modification. In particular, symmetry-adapted constructions based on Schur--Weyl duality or spin-coupling techniques apply verbatim\cite{fei2023efficientquantumalgorithmportbased,Grinko2024,PRXQuantum.5.030354}. Noise degrades the teleportation fidelity but does not increase the circuit complexity of the measurement stage. Simulability in this scenario is therefore guaranteed and identical to that of noiseless PBT, independently of the noise model or strength.

\emph{Noisy resources with noise-adapted POVMs.—}
The situation changes qualitatively when the measurement is adapted to the noisy signal
ensemble. For the dephasing noise model studied here, the ensemble remains exactly
covariant under permutations of the ports,
\begin{equation}
[\bar\eta(\Gamma),\,U_\pi\otimes\mathds{1}_B]=0
\qquad \forall\,\pi\in S_N,
\end{equation}
and therefore admits a decomposition
\begin{equation}
(\mathbb C^d)^{\otimes N}\otimes\mathbb C^d
\;\cong\;
\bigoplus_{\lambda\vdash N}
\mathcal V_\lambda\otimes\mathbb C^{m_\lambda}\otimes\mathbb C^d,
\end{equation}
with $\bar\eta(\Gamma)$ block diagonal in the permutation label $\lambda$.

For any nonzero noise strength, however, the joint
$U^{\otimes N}\!\otimes U^*$ covariance of the noiseless protocol is broken, so
$\bar\eta(\Gamma)$ need no longer act trivially on the multiplicity spaces
$\mathbb C^{m_\lambda}$. Consequently, the nontrivial blocks on which
$\bar\eta(\Gamma)$ acts have dimensions
\(
D_\lambda = d\,m_\lambda,
\)
which typically grow polynomially with the number of ports $N$ (for fixed local
dimension $d$).

This loss of joint unitary symmetry removes the constant-size spectral structure that
underpins the guaranteed efficiency of noiseless PBT measurements. While permutation
symmetry still enables a blockwise treatment of the ensemble average operator, it no
longer enforces an efficient exact implementation of the noise-adapted PGMs. In
particular, the complexity of realizing the inverse square root
$\bar\eta(\Gamma)^{-1/2}$ becomes dependent on the growing block dimensions and on
quantitative spectral properties of the noisy ensemble, rather than being fixed by
symmetry alone.

\emph{Conditions for efficient circuits with noisy POVMs---}
The absence of symmetry-protected constant-size blocks does not, by itself, imply
inefficiency. Polynomial block sizes do not preclude efficient \emph{approximate}
implementations; instead, efficiency becomes contingent on quantitative spectral
properties of the ensemble average operator.

In broad classes of measurement-synthesis strategies based on functional calculus
(for example, polynomial or singular-value transformations approximating
$x\mapsto x^{-1/2}$~\cite{Gilyen2019QSVT,Low2019Hamiltonian}),
the cost of implementing the noise-adapted PGM
\begin{equation}
M_i(\Gamma)=\bar\eta(\Gamma)^{-1/2}\,\eta^{(i)}(\Gamma)\,\bar\eta(\Gamma)^{-1/2}
\end{equation}
is controlled by three parameters: the block dimensions $D_\lambda$, the blockwise
condition numbers
\begin{equation}
\kappa_\lambda(\Gamma)
=
\frac{\lambda_{\max}(\bar\eta(\Gamma)|_\lambda)}
     {\lambda_{\min}(\bar\eta(\Gamma)|_\lambda)},
\end{equation}
and the target approximation accuracy~\cite{Gilyen2019QSVT,Low2019Hamiltonian}.

Permutation symmetry ensures that the block dimensions \(D_\lambda\) grow at most polynomially with \(N\), for fixed local dimension \(d\). However, the behaviour of the condition numbers \(\kappa_\lambda(\Gamma)\) depends on the noise model and on the noise strength. For the standard measurement-synthesis strategies considered here, such as polynomial approximation or singular-value transformation of the function \(x\mapsto x^{-1/2}\), efficient approximate implementation requires these relevant blockwise condition numbers to remain polynomially controlled. If they grow too rapidly, these constructions may incur super-polynomial overhead. This statement should not be read as a general no-go result for all possible implementations of the noise-adapted measurement. Other approaches may exploit additional algebraic, approximate, or noise-specific structure. Our conclusion is therefore more modest: after the joint unitary symmetry of noiseless PBT is broken, permutation symmetry alone no longer guarantees efficient implementation of the noise-adapted PGM; the question becomes quantitative and depends on the block spectra and on the chosen synthesis method.

\section{Discussion and conclusion}
\label{Sec:discussion}
In this work we have analysed deterministic PBT in the
presence of pure-dephasing decoherence affecting both the entangled resource state and the
measurement process. Our primary objective was to understand how proposed noise models modify the performance of the PBT protocol and to
identify regimes in which analytical control can still be maintained. To this
end, we focused on a setting in which each Bell pair forming the resource state
interacts with its own identical local environment, corresponding to
independently distributed entangled links. This choice captures still an
experimentally relevant scenario while preserving sufficient symmetry to allow
for explicit analysis.

We considered two complementary noisy scenarios:
\begin{itemize}
    \item Decoherence acts only on the resource state while Alice’s measurement is assumed to be ideal.
In this case, we derived a lower bound on the teleportation fidelity and derived a closed-form
expression for the fidelity valid for arbitrary numbers of ports. This enabled
us to analyse the asymptotic behaviour of the protocol and to show that, while decoherence expectedly lowers the fidelity compared to the Ishizaka--Hiroshima result, 
the correction term vanishes when the resource state grows, leaving the overall qualitative
behaviour unchanged in the asymptotic regime.

\item In the second scenario, we addressed the problem of noisy measurements and
introduced noise-adapted POVMs tailored to the given decoherence model, which in this case means using the pretty good measurements, constructed from the noisy ensemble. While a
fully analytical treatment appears out of reach for a general~$N$, we
demonstrated that a semi-analytical approach combined with numerical
optimisation yields performance estimates. Moreover, for the special case
$N=2$, we derived an explicit analytical lower bound on the entanglement
fidelity based on the Helstrom optimality bound, providing a useful benchmark
for assessing the performance of noise-adapted measurements. Quite surprisingly, we find that noise-adapted measurements perform worse for finite $N$ than the ideal, noise-agnostic ones.
\end{itemize}

In both scenarios, we discussed the algorithmic aspects of PBT in the presence of noise,
focusing on how different noise scenarios affect the status of efficient quantum
implementations. When the Ishizaka--Hiroshima measurement is retained, noise alters
the achieved fidelity and the required number of ports but does not change the
structure of the measurement itself, so existing efficient implementations remain
applicable. In contrast, for noise-adapted POVMs constructed from the noisy signal
ensemble, the symmetries underlying noiseless PBT are partially broken. While
permutation symmetry still induces a block-diagonal structure, it no longer enforces
the constant-size blocks responsible for symmetry-protected efficiency. As a result,
the existence of efficient implementations for noise-adapted measurements cannot be
inferred from symmetry alone and becomes a quantitative question, dependent on
spectral properties of the noisy ensemble rather than addressed conclusively here.

Further, to connect  with microscopic physical models, we
embedded the protocol into a spin--boson framework and parametrised the
environment in terms of dimensionless quantities controlling bath memory,
temperature, cut-off, and the interaction time. By analysing representative super-Ohmic
regimes $s=2$ and $s=3$ and temperatures, we demonstrated that the temperature has predictably a strong effect on the fidelity while the Ohmicity in this case has only a moderate influence, unlike for a single spin models, where the transition from $s=2$ to $s=3$ corresponds to the change from Markovian to non-Markovian regimes.
We thereby illustrated the sensitivity of the PBT protocol to microscopic noise properties beyond depolarising models.

Let us also comment briefly on experimental relevance.  The local pure-dephasing
model should be understood as an effective description of stationary resource qubits
whose dominant error during storage or distribution is phase noise, with relaxation,
leakage, and control errors neglected on the time scale of the protocol.  In this
sense, solid-state spin platforms, such as NV centres in diamond, are a natural
setting, since they combine local spin-environment dephasing with spin--photon
interfaces and long-lived memory qubits~\cite{Bernien2013Heralded,ChildressHanson2013NV,Bradley2022Robust}.
Superconducting circuits are also relevant because local dephasing is a standard and
well-characterized noise mechanism in these systems~\cite{Siddiqi2021Engineering,PhysRevB.85.224505,Bylander2011Noise}.
Trapped ions provide highly coherent and controllable qubits, and may be useful for
engineered-noise tests, although their native noise is not generally a literal
realization of independent local spin--boson baths~\cite{Bruzewicz2019TrappedIons}.
We therefore do not claim that a large-\(N\) deterministic PBT implementation with
the full collective square-root measurement is currently within near-term experimental
reach.  A more realistic near-term objective would be a small-\(N\) proof-of-principle experiment, or a digital/analog simulation of the noisy signal-state discrimination problem, testing
the predicted dependence of the teleportation fidelity on the decoherence parameter
\(\Gamma\).

Several directions for future work naturally emerge from our analysis. An
important extension would be to consider common or correlated environments
acting on multiple Bell pairs, which may give rise to collective decoherence
effects not captured by the present model. Another open problem is the
systematic optimisation of noisy measurements for general~$N$, potentially
exploiting additional symmetries, like residual covariance with respect to the group
$U(1)_Z \times S_N$, or approximate covariance properties to reduce computational
complexity. From a physical perspective, it would also be interesting to investigate
other classes of noise models and environments beyond the pure-dephasing setting
considered here.  Examples include depolarizing or more general Pauli noise, which
may populate several Bell sectors, and amplitude damping, which is non-unital and
changes the population structure of the resource state.  Such models require a
separate analysis of the induced signal ensemble and of the associated PBT
measurements.  It would also be interesting to explore experimentally motivated
noise-mitigation strategies within the PBT framework, e.g. similar to the so-called
purifying teleportation~\cite{Roszak2023purifying}. 

We also note that deterministic port-based teleportation is known, in the
noiseless setting, to admit one-to-one correspondences with several other
higher-order quantum information primitives, including unitary estimation,
deterministic storage-and-retrieval, and parallel unitary inversion~\cite{yoshida2024onetoone}. From this
broader perspective, the present analysis of noise in PBT can be viewed as a
case study of how environmental decoherence deforms the symmetry structures
that underpin such equivalences. While our results do not directly extend to
these related protocols, they indicate that noise generically modifies the
tradeoffs between performance, symmetry, and algorithmic implementability.
Understanding whether analogous noise-dependent effects arise across the
corresponding unitary-processing tasks remains an open direction for future
work.

We believe that the methods developed here
provide a solid starting point for addressing these questions and for future work.

\begin{acknowledgments}
MS is supported by the National Science Centre, Poland, Grant Sonata Bis 14 no. 2024/54/E/ST2/00316. RSB and JKK acknowledge the support of  National Science Centre (NCN) through the QuantEra
project Qucabose 2023/05/Y/ST2/00139.
\end{acknowledgments}

\appendix

\onecolumngrid

\section{Fidelity between $\eta^{(i)}$ and $\eta^{(j)}$} \label{Fij}

The state fidelity between two signal states $\eta^{(i)}$ and $\eta^{(j)}$ is defined as
\begin{equation}
    F(\eta^{(i)}, \eta^{(j)}) = \tr\sqrt{\sqrt{\eta^{(i)}}\eta^{(j)}\sqrt{\eta^{(i)}}}.
\end{equation}
Using \eqref{eta_i}, \eqref{sigma_i}, and \eqref{omega_i} we immediately find that: 
\begin{align}
        \sqrt{\eta^{(i)}}
        =R(\theta)_B\left[\sqrt{\frac{1+|\Gamma|}{2}}\ketbra{\Psi^-} + \sqrt{\frac{1-|\Gamma|}{2}}\ketbra{\Psi^+}\right]R(\theta)_B^\dagger\otimes\frac{\mathds{1}_{\bar{A}_i}}{\sqrt{2^{N-1}}} .
    \end{align}
Lengthy but straightforward calculations show that:
\begin{align}
    \sqrt{\eta^{(i)}}\eta^{(j)}\sqrt{\eta^{(i)}}
        & = R(\theta)\nu_{A_i A_j B}R(\theta)^\dagger\otimes\frac{\mathds{1}_{\bar{A}_{ij}}}{2^{2(N-2)}}, 
\end{align}
where 
\begin{align}
    \nu_{A_i A_j B}  \equiv &\frac{1}{32}[(1+\sqrt{1-|\Gamma|^2})(\ketbra{001}+\ketbra{110}) + (1-\sqrt{1-|\Gamma|^2})(\ketbra{011}+\ketbra{100}) \\
         & - |\Gamma|\left(\ketbra{001}{100}+\ketbra{011}{110}+\ketbra{100}{001}+\ketbra{110}{011}\right)]. 
\end{align}
This is a $4\times 4$ matrix, which has a simple block diagonal form with the following matrix on the diagonal:
\begin{equation}
    \frac{1}{32}\begin{pmatrix}
1+\sqrt{1+|\Gamma|^2} & -|\Gamma| \\
-|\Gamma| & 1-\sqrt{1+|\Gamma|^2}
\end{pmatrix},
\end{equation}
which has eigenvalues $0$ and $2/32=1/16$. We thus obtain:
\begin{align}
   \tr\sqrt{\sqrt{\eta^{(i)}}\eta^{(j)}\sqrt{\eta^{(i)}}}= \tr\sqrt{R(\theta)_B \nu_{A_i A_j B} R(\theta)_B^\dagger}=\tr\sqrt{\nu_{A_i A_j B}} = \frac{1}{4}+\frac{1}{4}
\end{align}
 and therefore: 
\begin{align}
    F(\eta^{(i)}, \eta^{(j)}) = \frac{1}{2},
\end{align}
which is constant, independent of both $N$ and $\Gamma$.

\section{Calculations for the noiseless measurements} \label{App_noiseless}

Eigenstates of $\rho$, two families from the spin$-1/2$ addition to $N-2$ spins, one family corresponds to $j+1/2$, second to $j-1/2$ 
\begin{equation}
\label{Psi_I}
    \begin{aligned}
        \ket{\Psi_{\text{I}}(\lambda^{\mp}_{j};m,\beta)} 
        = & \ket{\Phi^{[N-1]}(j+\frac{1}{2},m+1,\beta)}_{\bar{A}_i}\ket{0}_{A_i}\ket{0}_B\braket{j,m+\frac{1}{2},\frac{1}{2},-\frac{1}{2}}{j\pm\frac{1}{2},m}
        \braket{j+\frac{1}{2},m+1,\frac{1}{2},-\frac{1}{2}}{j,m+\frac{1}{2}}\\
        +& \ket{\Phi^{[N-1]}(j+\frac{1}{2},m,\beta)}_{\bar{A}_i}\ket{1}_{A_i}\ket{0}_B\braket{j,m+\frac{1}{2},\frac{1}{2},-\frac{1}{2}}{j\pm\frac{1}{2},m}\braket{j+\frac{1}{2},m,\frac{1}{2},\frac{1}{2}}{j,m+\frac{1}{2}}  \\      
       +  & \ket{\Phi^{[N-1]}(j+\frac{1}{2},m,\beta)}_{\bar{A}_i}\ket{0}_{A_i}\ket{1}_B\braket{j,m-\frac{1}{2},\frac{1}{2},\frac{1}{2}}{j\pm\frac{1}{2},m}\braket{j+\frac{1}{2},m,\frac{1}{2},-\frac{1}{2}}{j,m-\frac{1}{2}}  \\
        + & \ket{\Phi^{[N-1]}(j+\frac{1}{2},m-1,\beta)}_{\bar{A}_i}\ket{1}_{A_i}\ket{1}_B\braket{j,m-\frac{1}{2},\frac{1}{2},\frac{1}{2}}{j\pm\frac{1}{2},m}
        \braket{j+\frac{1}{2},m-1,\frac{1}{2},\frac{1}{2}}{j,m-\frac{1}{2}}\\
    \end{aligned}
\end{equation}
and
\begin{equation}
\label{Psi_II}
    \begin{aligned}
        \ket{\Psi_{\text{II}}(\lambda^{\mp}_{j};m,\beta)} 
        = & \ket{\Phi^{[N-1]}(j-\frac{1}{2},m+1,\beta)}_{\bar{A}_i}\ket{0}_{A_i}\ket{0}_B\braket{j,m+\frac{1}{2},\frac{1}{2},-\frac{1}{2}}{j\pm\frac{1}{2},m}
        \braket{j-\frac{1}{2},m+1,\frac{1}{2},-\frac{1}{2}}{j,m+\frac{1}{2}}\\
        +& \ket{\Phi^{[N-1]}(j-\frac{1}{2},m,\beta)}_{\bar{A}_i}\ket{1}_{A_i}\ket{0}_B\braket{j,m+\frac{1}{2},\frac{1}{2},-\frac{1}{2}}{j\pm\frac{1}{2},m}\braket{j-\frac{1}{2},m,\frac{1}{2},\frac{1}{2}}{j,m+\frac{1}{2}}  \\      
       +  & \ket{\Phi^{[N-1]}(j-\frac{1}{2},m,\beta)}_{\bar{A}_i}\ket{0}_{A_i}\ket{1}_B\braket{j,m-\frac{1}{2},\frac{1}{2},\frac{1}{2}}{j\pm\frac{1}{2},m}\braket{j-\frac{1}{2},m,\frac{1}{2},-\frac{1}{2}}{j,m-\frac{1}{2}}  \\
        + & \ket{\Phi^{[N-1]}(j-\frac{1}{2},m-1,\beta)}_{\bar{A}_i}\ket{1}_{A_i}\ket{1}_B\braket{j,m-\frac{1}{2},\frac{1}{2},\frac{1}{2}}{j\pm\frac{1}{2},m}
        \braket{j-\frac{1}{2},m-1,\frac{1}{2},\frac{1}{2}}{j,m-\frac{1}{2}}\\
    \end{aligned}.
\end{equation}

\begin{equation}
    \sigma^{(i)} = \sum_{s=s_{min}}^{(N-1)/2}\sigma^{(i)}(s),
\end{equation}
with
\begin{equation}
    \sigma^{(i)}(s) = \frac{1}{2^{N-1}}\sum_{m=-s}^{s}\sum_{\beta}\ketbra{\xi^{(i)}_-(s,m,\beta)}.
\end{equation}
and similarly: 
\begin{equation}
    \omega^{(i)} = \sum_{s=s_{min}}^{(N-1)/2}\omega^{(i)}(s),
\end{equation}

\begin{equation}
    \omega^{(i)}(s) = \frac{1}{2^{N-1}}\sum_{m=-s}^{s}\sum_{\beta}\ketbra{\xi^{(i)}_+(s,m,\beta)}.
\end{equation}

Following Ishizaka and Hiroshima~\cite{ishizaka_asymptotic_2008,ishizaka_quantum_2009},
for comparison, in the IH calculation:
\begin{equation}\label{sigma_succ_prob}
\tr\Pi^{(i)}{(s)}\sigma^{(i)}(s)  =
         g^{[N-1]}(s) \frac{1}{2s+1} \left[s\left(\lambda^-_{s-1/2}\right)^{-1/2} +(s+1)\left(\lambda^+_{s+1/2}\right)^{-1/2}\right]^2,
\end{equation}
and
 \begin{align}
        \tr\Pi^{(i)}{(s)}\omega^{(i)}(s) & = \tr\rho(s)^{-1/2}\sigma^{(i)}(s)\rho(s)^{-1/2}\omega^{(i)}(s) \\
        & = \sum_{m, m'=-s}^s\sum_{\beta,\beta'}\left|\bra{\xi_{+}^{(i)}(s,m^\prime,\beta^\prime)}{\rho(s)}^{-1/2}\ket{\xi_{-}^{(i)}(s,m,\beta)}\right|^2 \\
        & = \sum_{\beta} \sum_{m=-s}^s \frac{m^2}{(2s+1)^2}\left[\left(\lambda^-_{s-1/2}\right)^{-1/2} -\left(\lambda^+_{s+1/2}\right)^{-1/2}\right]^2 \\
        & =  \sum_{\beta} \frac{s(s+1)(2s+1)}{3} \frac{1}{(2s+1)^2} \left[\left(\lambda^-_{s-1/2}\right)^{-1/2} -\left(\lambda^+_{s+1/2}\right)^{-1/2}\right]^2 \\
        & = \frac{1}{3} g^{[N-1]}(s) \frac{s(s+1)}{2s+1} \left[\left(\lambda^-_{s-1/2}\right)^{-1/2} -\left(\lambda^+_{s+1/2}\right)^{-1/2}\right]^2.
    \end{align}
Here, we have used $\sum_{m=-s}^{s} m^2 = \frac{1}{3}s(s+1)(2s+1)$ and $\sum_\beta = g^{[N-1]}(s)$ is the intrinsic degeneracy for $\ket{\Psi_{\text{I}}(\lambda^-_{s-1/2};m,\beta)}$ and $\ket{\Psi_{\text{II}}(\lambda^+_{s+1/2};m,\beta)}$.
Further simplification of the above term is presented in Appendix~\ref{app:F_corr_term}.

Now, we show vanishing of expression~\eqref{mixedterm}. We have the following:
\label{off_diag_terms_NL}
\begin{align} 
& \tr\Pi^{(i)}_{\Vec{A}B}\left[\left(\ketbra{\Psi^+}{\Psi^-}-\ketbra{\Psi^-}{\Psi^+}\right)_{A_i  B}\otimes\frac{\mathds{1}_{\bar{A}_i}}{2^{N-1}}\right] = \\
& =\frac{1}{2^{N-1}}\sum_s \sum_{m,\beta}\left\{\bra{\xi^{(i)}_{-}(s,m,\beta)} \Pi^{(i)}{(s)} \ket{\xi^{(i)}_{+}(s,m,\beta)} - c.c.\right\}\\
& = \frac{1}{2^{N-1}}\sum_s \sum_{m,m',\beta,\beta'} \left\{\bra{\xi^{(i)}_{-}(s,m,\beta)} \rho(s)^{-1/2} \ket{\xi^{(i)}_{-}(s,m',\beta')} \bra{\xi^{(i)}_{-}(s,m',\beta')} \rho(s)^{-1/2} \ket{\xi^{(i)}_{+}(s,m',\beta')} - c.c.\right\}\\
&= \frac{1}{2^{N-1}}\sum_s \sum_{m} g^{[N-1]}(s) \left\{ c_{IH}(s) \frac{m}{2s+1} \left[\left(\lambda^-_{s-1/2}\right)^{-1/2} -\left(\lambda^+_{s+1/2}\right)^{-1/2}\right] - c.c. \right\} = 0.
\end{align}
Where, $c_{IH}(s)=\bra{\xi^{(i)}_-(s,m,\beta)}\rho(s)^{-1/2}\ket{\xi^{(i)}_-(s,m',\beta')}$.
Since $\sum_m m=0$, actually not only all the overlaps above are real and hence their imaginary part is zero, but each term $\sum_m \bra{\xi^{(i)}_{-}(s,m,\beta)} \Pi^{(i)}{(s)} \ket{\xi^{(i)}_{+}(s,m,\beta)}$ vanishes on its own.

\section{Explicit derivations from the Helstrom bound}\label{Helstrom}
The goal of this Appendix is to prove expression~\eqref{eq:Hel_bound} from the main text, which reads:
\begin{align}
\label{eq:Hel_bound2}
    F\leq\frac{1}{4}\left(1+\frac{\sqrt{1+2|\Gamma|^2}}{2}\right).
\end{align}

For $N=2$, the noisy states entering the binary discrimination problem are
\begin{equation}
\eta^{(i)}_{A_1A_2B}
=
R_B(\theta)
\left[
\frac{1+|\Gamma|}{2}\,\sigma^{(i)}_{A_1A_2B}
+
\frac{1-|\Gamma|}{2}\,\omega^{(i)}_{A_1A_2B}
\right]
R_B(\theta)^\dagger,
\qquad i=1,2,
\end{equation}
with
\begin{align}
\sigma^{(1)}_{A_1A_2B}
&=
\frac12\,|\Psi^{-}\rangle\!\langle\Psi^{-}|_{A_1B}\otimes\mathbf{1}_{A_2},
&
\omega^{(1)}_{A_1A_2B}
&=
\frac12\,|\Psi^{+}\rangle\!\langle\Psi^{+}|_{A_1B}\otimes\mathbf{1}_{A_2},
\\
\sigma^{(2)}_{A_1A_2B}
&=
\frac12\,\mathbf{1}_{A_1}\otimes|\Psi^{-}\rangle\!\langle\Psi^{-}|_{A_2B},
&
\omega^{(2)}_{A_1A_2B}
&=
\frac12\,\mathbf{1}_{A_1}\otimes|\Psi^{+}\rangle\!\langle\Psi^{+}|_{A_2B}.
\end{align}

We define the difference operator
\begin{equation}
\Delta := \eta^{(1)}-\eta^{(2)}.
\end{equation}
Since both $\eta^{(1)}$ and $\eta^{(2)}$ are conjugated by the same unitary
$R_B(\theta)$, the trace norm is invariant under this transformation and
\begin{equation}
\|\Delta\|_1
=
\|\Delta_0\|_1,
\qquad
\Delta_0 :=
\frac{1+|\Gamma|}{2}\bigl(\sigma^{(1)}-\sigma^{(2)}\bigr)
+
\frac{1-|\Gamma|}{2}\bigl(\omega^{(1)}-\omega^{(2)}\bigr).
\end{equation}
Hence, the phase $\theta$ does not affect the Helstrom bound, and only
$|\Gamma|$ enters.

Notice that in this case the operators $\sigma^{(i)}$ and $\omega^{(i)}$ differ only by which
port is paired with the system $B$. It means that $\Delta_0$ is antisymmetric under the exchange of the two ports
$A_1\leftrightarrow A_2$.
It follows that $\Delta_0$ vanishes on the symmetric subspace of $A_1A_2$ and
acts non-trivially only on a two-dimensional sector.
In particular, $\Delta_0$ is Hermitian, traceless, and has rank at most $2$. In the   $B$ respecting these symmetries 
\begin{align}
\label{eq:basisBApp}
    B=\Bigl\{
|00\rangle|0\rangle,\ |00\rangle|1\rangle,\ 
|\psi_+\rangle|0\rangle,\ |\psi_+\rangle|1\rangle,\ 
|11\rangle|0\rangle,\ |11\rangle|1\rangle\ ;\
|\psi_-\rangle|0\rangle,\ |\psi_-\rangle|1\rangle
\Bigr\},
\end{align}
the term $\Delta_0$ has the following explicit form:
\begin{align}
\Delta_0=
\begin{pmatrix}
0_{6\times 6} & M\\[2pt]
M^\dagger & 0_{2\times 2}
\end{pmatrix},
\qquad
M=
\begin{pmatrix}
0 & 0\\
\sqrt{2}|\Gamma|/4 & 0\\
-1/4 & 0\\
0 & 1/4\\
0 & -\sqrt{2}|\Gamma|/4\\
0 & 0
\end{pmatrix}.
\end{align}
Notice that the first 6 vectors of the basis $B$ in~\eqref{eq:basisBApp} span the symmetric space $\mathrm{Sym}(A_1A_2)\otimes B$, while the last 2 span the antisymmetric sector $\mathrm{Asym}(A_1A_2)\otimes B$.

A direct diagonalization of $\Delta_0$ on this sector yields the nonzero
eigenvalues
\begin{equation}
\lambda_{\pm}=\pm\frac12\sqrt{1+2|\Gamma|^2},
\end{equation}
while all remaining eigenvalues are zero.
Since $\Delta$ is Hermitian, its trace norm equals the sum of absolute values
of its eigenvalues, and thus
\begin{equation}
\bigl\|\eta^{(1)}-\eta^{(2)}\bigr\|_1
=
|\lambda_{+}|+|\lambda_{-}|
=
\sqrt{1+2|\Gamma|^2},
\end{equation}
which proves Eq.~\eqref{eq:Hel_bound2}.

\section{Derivation of expression (29) from~\cite{ishizaka_quantum_2009}}
\label{derivationPRA}
For the Reader's convenience, we derive here expression~\eqref{eq:F_ih_noiseless} for the entanglement fidelity $F_{IH}$ for the noiseless deterministic PBT reported for the first time in~\cite{ishizaka_asymptotic_2008,ishizaka_quantum_2009}.
Starting from the expression for the entanglement fidelity under the square-root measurement (SRM), see the second line of equation (29) in~\cite{ishizaka_quantum_2009}, we have:
\begin{equation}
F_{IH} = \frac{N}{2^{2N}} 
\sum_{s=s_{\min}}^{(N-1)/2} 
(2s+1)\, g^{[N-1]}(s)\, [c(s,2)]^{2},
\label{eq:F-start}
\end{equation}
where
\begin{equation}
c(s,y) = \frac{s}{2s+1}\bigl(\lambda^{-}_{s-1/2}\bigr)^{-1/y}
+ \frac{s+1}{2s+1}\bigl(\lambda^{+}_{s+1/2}\bigr)^{-1/y},
\label{eq:c-def}
\end{equation}
and the eigenvalues of $\rho$ are
\begin{equation}
\lambda^{-}_{s-1/2} = \frac{N-2s+1}{2^{N+1}}, 
\qquad 
\lambda^{+}_{s+1/2} = \frac{N+2s+3}{2^{N+1}}.
\label{eq:lambdas}
\end{equation}

\paragraph{Step 1: Substitute eigenvalues.}
For $y=2$, Eq.~\eqref{eq:c-def} gives
\begin{equation}
c(s,2) = 
\frac{s}{2s+1}\sqrt{\frac{2^{N+1}}{N-2s+1}}
+ \frac{s+1}{2s+1}\sqrt{\frac{2^{N+1}}{N+2s+3}}.
\label{eq:c2}
\end{equation}

\paragraph{Step 2: Change of variables.}
Let $k = \frac{N-1}{2} - s$, so that
\[
N - 2s + 1 = 2(k+1), \qquad 
N + 2s + 3 = 2(N-k+1).
\]
Then Eq.~\eqref{eq:c2} becomes
\begin{equation}
c(s,2) = \frac{2^{(N-1)/2}}{2s+1}
\left(
\frac{s}{\sqrt{k+1}} + 
\frac{s+1}{\sqrt{N-k+1}}
\right).
\label{eq:c2k}
\end{equation}

\paragraph{Step 3: Insert into \eqref{eq:F-start}.}
Squaring and substituting into \eqref{eq:F-start} yields
\begin{equation}
F_{IH} = 
\frac{N}{2^{N+1}}
\sum_{s=s_{\min}}^{(N-1)/2}
\frac{g^{[N-1]}(s)}{2s+1}
\left(
\frac{s}{\sqrt{k+1}} +
\frac{s+1}{\sqrt{N-k+1}}
\right)^{\!2}.
\label{eq:F-mid}
\end{equation}

\paragraph{Step 4: Express multiplicity as a binomial coefficient.}
From Eq.~(9) of the paper (with $N\!\to\!N-1$),
\begin{equation}
g^{[N-1]}(s)
= \frac{(2s+1)(N-1)!}{\bigl(\tfrac{N-1}{2}-s\bigr)!\bigl(\tfrac{N+1}{2}+s\bigr)!}
= \frac{(2s+1)(N-1)!}{k!(N-k)!}.
\end{equation}
Hence
\[
\frac{g^{[N-1]}(s)}{2s+1}
= \frac{1}{N}\binom{N}{k}.
\]
Substituting into \eqref{eq:F-mid} cancels the prefactor $N$:
\begin{equation}
F_{IH} = 
\frac{1}{2^{N+1}}
\sum_{k=0}^{N}
\binom{N}{k}
\left(
\frac{s}{\sqrt{k+1}} +
\frac{s+1}{\sqrt{N-k+1}}
\right)^{\!2}.
\end{equation}

\paragraph{Step 5: Express $s$ in terms of $k$.}
Since $s = \frac{N-1}{2} - k$, we have
\[
\frac{s}{\sqrt{k+1}} = \frac{N-2k-1}{2\sqrt{k+1}},
\qquad
\frac{s+1}{\sqrt{N-k+1}} = \frac{N-2k+1}{2\sqrt{N-k+1}}.
\]
Plugging this in yields the final compact form:
\begin{equation}
\boxed{
F_{IH} =
\frac{1}{2^{N+3}}
\sum_{k=0}^{N}
\left(
\frac{N-2k-1}{\sqrt{k+1}} +
\frac{N-2k+1}{\sqrt{N-k+1}}
\right)^{\!2}
\binom{N}{k}
}.
\label{eq:F29}
\end{equation}

\section{Calculating the closed-form expression for $F_{\mathrm{corr}}$}
\label{app:F_corr_term}

We start from the expression for the correction contribution appearing in
Eq.~\eqref{ent_fid_NL}. Recall that the entanglement fidelity is related
to the corresponding state-discrimination success probability by
\begin{equation}
F=\frac{1}{4}\sum_{i=1}^{N}\tr\!\left[\Pi^{(i)}\eta^{(i)}\right].
\label{eq:Fcorr_app_F_def}
\end{equation}
Therefore, the correction term associated with $\omega^{(i)}$ is
\begin{equation}
F_{\mathrm{corr}}
=
\frac{1}{4}
\sum_{i=1}^N
\sum_{s=s_{\min}}^{(N-1)/2}
\tr\!\left[\Pi^{(i)}(s)\omega^{(i)}(s)\right].
\label{eq:Fcorr_start_corrected}
\end{equation}
By permutation symmetry of the ports, each value of $i$ gives the same
contribution. Hence
\begin{equation}
F_{\mathrm{corr}}
=
\frac{N}{4}
\sum_{s=s_{\min}}^{(N-1)/2}
\tr\!\left[\Pi^{(N)}(s)\omega^{(N)}(s)\right].
\label{eq:Fcorr_symmetry_corrected}
\end{equation}
Substituting the explicit expression for
$\tr[\Pi^{(N)}(s)\omega^{(N)}(s)]$ from Eq.~\eqref{eq:49}, we obtain
\begin{equation}
F_{\mathrm{corr}}
=
\frac{N}{4}
\frac{1}{2^{2N-2}}
\sum_{s=s_{\min}}^{(N-1)/2}
\frac{1}{3}\,
g^{[N-1]}(s)\,
\frac{s(s+1)}{2s+1}
\left[
\left(\lambda^-_{s-\frac12}\right)^{-1/2}
-
\left(\lambda^+_{s+\frac12}\right)^{-1/2}
\right]^2 .
\label{eq:S-start-corrected}
\end{equation}
Here the eigenvalues of $\rho$ are
\begin{equation}
\lambda^-_{s-\frac12}
=
\frac{N-2s+1}{2^{N+1}},
\qquad
\lambda^+_{s+\frac12}
=
\frac{N+2s+3}{2^{N+1}},
\label{eq:lambdas-corrected}
\end{equation}
and the multiplicity of the spin-$s$ irrep in the tensor product of
$N-1$ qubits is
\begin{equation}
g^{[N-1]}(s)
=
\frac{(2s+1)(N-1)!}
{
\left(\frac{N-1}{2}-s\right)!
\left(\frac{N+1}{2}+s\right)!
}.
\label{eq:g-corrected}
\end{equation}

\paragraph{Step 1: Substitution of the eigenvalues.}

From Eq.~\eqref{eq:lambdas-corrected}, we have
\begin{equation}
\left(\lambda^-_{s-\frac12}\right)^{-1/2}
-
\left(\lambda^+_{s+\frac12}\right)^{-1/2}
=
2^{\frac{N+1}{2}}
\left(
\frac{1}{\sqrt{N-2s+1}}
-
\frac{1}{\sqrt{N+2s+3}}
\right).
\label{eq:eigenvalue_difference_app}
\end{equation}

\paragraph{Step 2: Change of variables.}

Let
\begin{equation}
k=\frac{N-1}{2}-s,
\qquad
s=\frac{N-1}{2}-k.
\label{eq:k-def-corrected}
\end{equation}
Since $s$ labels irreducible representations of $N-1$ qubits, the range of
$k$ is
\begin{equation}
k=0,1,\ldots,\left\lfloor\frac{N-1}{2}\right\rfloor.
\label{eq:k-half-range}
\end{equation}
With this change of variables,
\begin{equation}
N-2s+1=2(k+1),
\qquad
N+2s+3=2(N-k+1),
\qquad
2s+1=N-2k.
\label{eq:k_relations_app}
\end{equation}
Moreover,
\begin{equation}
s(s+1)
=
\frac{(2s+1)^2-1}{4}
=
\frac{(N-2k)^2-1}{4}.
\label{eq:ss+1-corrected}
\end{equation}
Therefore,
\begin{align}
\left[
\left(\lambda^-_{s-\frac12}\right)^{-1/2}
-
\left(\lambda^+_{s+\frac12}\right)^{-1/2}
\right]^2
&=
2^{N}
\left(
\frac{1}{\sqrt{k+1}}
-
\frac{1}{\sqrt{N-k+1}}
\right)^2 .
\label{eq:diff-squared-corrected}
\end{align}

\paragraph{Step 3: Conversion of multiplicities to binomial coefficients.}

Using Eqs.~\eqref{eq:g-corrected} and \eqref{eq:k-def-corrected}, we obtain
\begin{equation}
g^{[N-1]}(s)
=
\frac{(2s+1)(N-1)!}{k!(N-k)!}
=
\frac{N-2k}{N}\binom{N}{k}.
\label{eq:g-binomial-corrected}
\end{equation}
Consequently,
\begin{align}
\frac{g^{[N-1]}(s)s(s+1)}{2s+1}
&=
\frac{1}{N}\binom{N}{k}s(s+1)
\nonumber\\
&=
\frac{1}{4N}
\binom{N}{k}
\left[(N-2k)^2-1\right].
\label{eq:prefactor-corrected}
\end{align}

\paragraph{Step 4: Closed-form expression.}

Substituting Eqs.~\eqref{eq:diff-squared-corrected} and
\eqref{eq:prefactor-corrected} into Eq.~\eqref{eq:S-start-corrected} gives
\begin{align}
F_{\mathrm{corr}}
&=
\frac{N}{4}
\frac{1}{2^{2N-2}}
\sum_{k=0}^{\lfloor (N-1)/2\rfloor}
\frac{1}{3}
\left[
\frac{1}{4N}
\binom{N}{k}
\left((N-2k)^2-1\right)
\right]
2^N
\left(
\frac{1}{\sqrt{k+1}}
-
\frac{1}{\sqrt{N-k+1}}
\right)^2
\nonumber\\[4pt]
&=
\frac{1}{12}\frac{1}{2^N}
\sum_{k=0}^{\lfloor (N-1)/2\rfloor}
\binom{N}{k}
\left[(N-2k)^2-1\right]
\left(
\frac{1}{\sqrt{k+1}}
-
\frac{1}{\sqrt{N-k+1}}
\right)^2 .
\label{eq:Fcorr_half_range}
\end{align}
The summand is invariant under $k\mapsto N-k$. Indeed,
\begin{equation}
\binom{N}{N-k}=\binom{N}{k},
\qquad
(N-2(N-k))^2=(N-2k)^2,
\label{eq:summand_symmetry_1}
\end{equation}
and
\begin{equation}
\left(
\frac{1}{\sqrt{N-k+1}}
-
\frac{1}{\sqrt{k+1}}
\right)^2
=
\left(
\frac{1}{\sqrt{k+1}}
-
\frac{1}{\sqrt{N-k+1}}
\right)^2 .
\label{eq:summand_symmetry_2}
\end{equation}
Thus the half-range sum in Eq.~\eqref{eq:Fcorr_half_range} can be written
as one half of the full binomial sum from $k=0$ to $k=N$. If $N$ is even,
the middle term $k=N/2$ vanishes because the squared difference is zero.
Therefore,
\begin{align}
F_{\mathrm{corr}}(N)
&=
\frac{1}{24}\frac{1}{2^N}
\sum_{k=0}^{N}
\binom{N}{k}
\left[(N-2k)^2-1\right]
\left(
\frac{1}{\sqrt{k+1}}
-
\frac{1}{\sqrt{N-k+1}}
\right)^2
\nonumber\\[4pt]
&=
\boxed{
\frac{1}{3}\frac{1}{2^{N+3}}
\sum_{k=0}^{N}
\binom{N}{k}
\left[(N-2k)^2-1\right]
\left(
\frac{1}{\sqrt{k+1}}
-
\frac{1}{\sqrt{N-k+1}}
\right)^2
.}
\label{eq:S-final-corrected}
\end{align}
The dependence $F_{\mathrm{corr}}=F_{\mathrm{corr}}(N)$ is depicted in
Figure~\ref{fig:fidelity_vs_N_Fcorr}. It is non-monotonic for small values
of $N$, reaches its maximum value
\begin{equation}
F_{\mathrm{corr}}(6)\approx 0.02909,
\label{eq:Fcorr_max_value}
\end{equation}
and then decays to zero as $N\rightarrow\infty$.

\begin{figure*}[t] 
	\centering
	\begin{minipage}{0.48\textwidth}
		\centering
    \includegraphics[width=\linewidth]{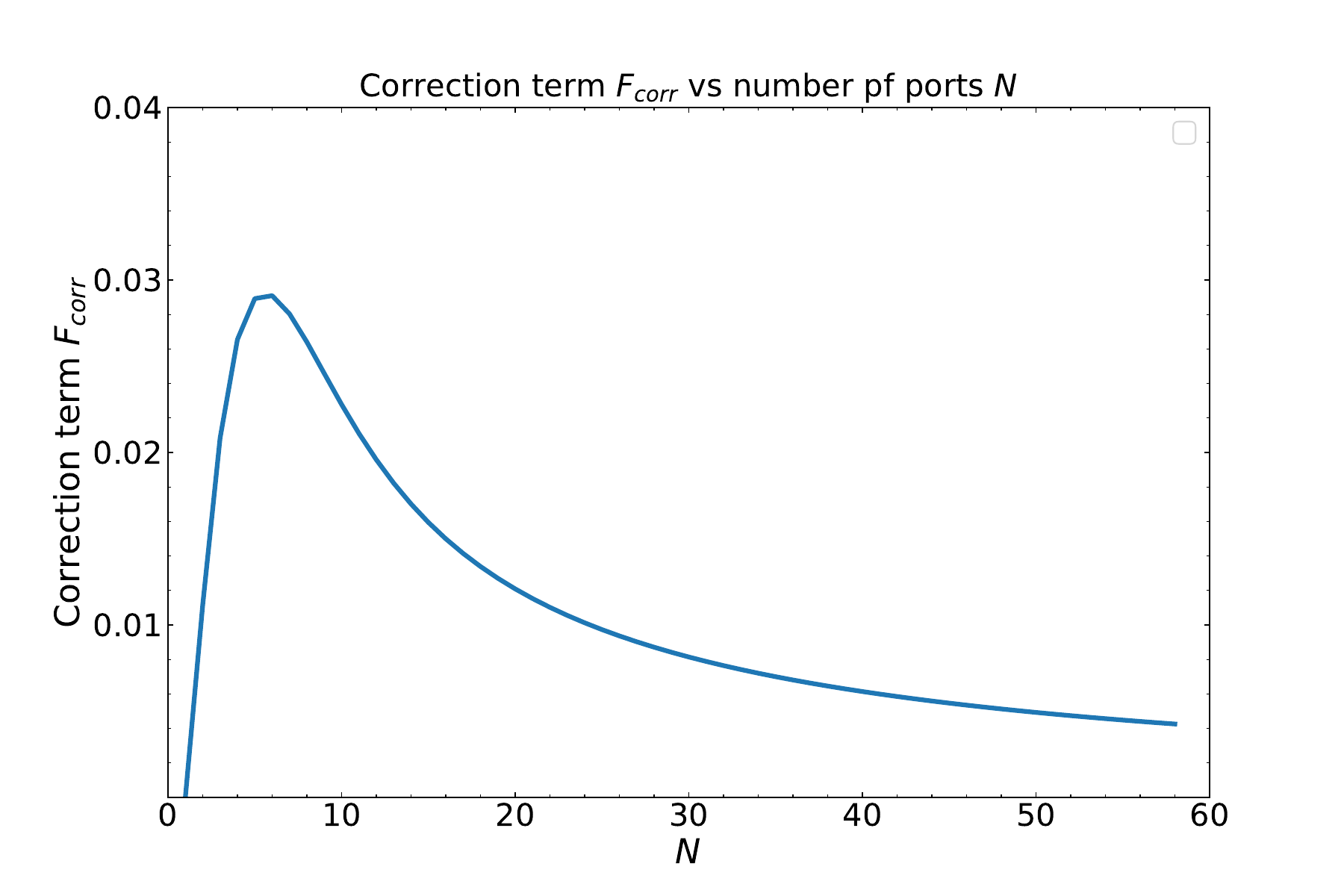} 
		\\[4pt] (a)
	\end{minipage}\hfill
	\begin{minipage}{0.48\textwidth}
		\centering
		\includegraphics[width=\linewidth]{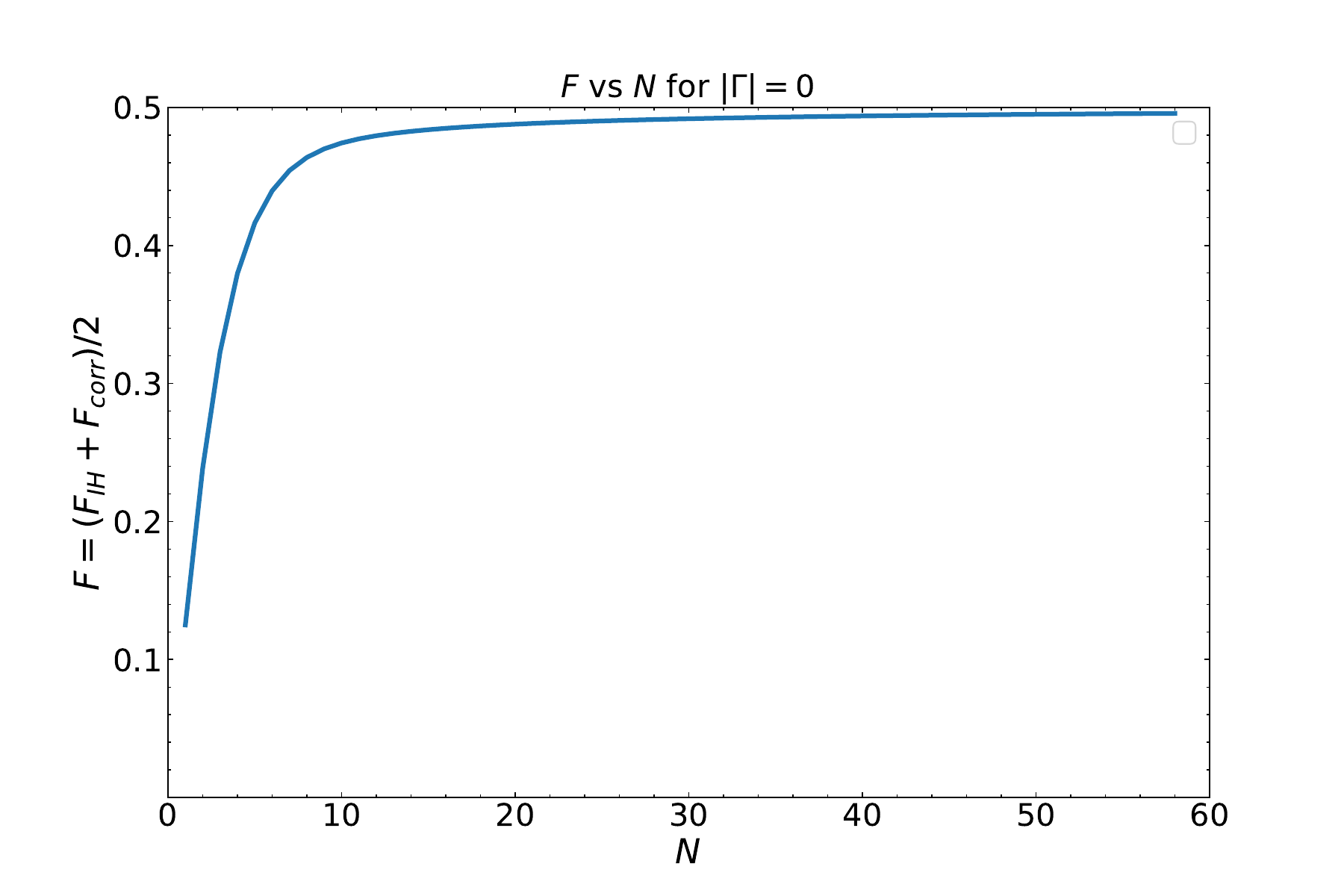}
		\\[4pt] (b)
	\end{minipage}
		\caption{(a) The correction term $F_{\mathrm{corr}}$ from
    Eq.~\eqref{eq:S-final-corrected}. The function is non-monotonic for
    small values of $N$ and achieves its maximum value
    $F_{\mathrm{corr}}\approx 0.02909$ at $N=6$. For $N>6$, the correction
    term decreases monotonically and vanishes in the limit
    $N\rightarrow\infty$. (b) The entanglement fidelity for $|\Gamma|=0$ and arbitrary value of $\theta$ from Eq.\eqref{ent_fid_NL}. The entanglement fidelity stays below $1/2$ for finite $N$ and approaches $1/2$ in the limit $N\rightarrow\infty$ implying that the shared resource is not useful for quantum teleportation.}
	\label{fig:fidelity_vs_N_Fcorr}
\end{figure*}

We now prove the asymptotic behaviour of $F_{\mathrm{corr}}$, used in the
discussion below Eq.~\eqref{eq:F_corr_term0} in the main text.

\begin{lemma}
\label{lem:Fcorr_asymptotics}
Let
\begin{equation}
F_{\mathrm{corr}}(N)
=
\frac{1}{3}\frac{1}{2^{N+3}}
\sum_{k=0}^{N}
\binom{N}{k}
\left[(N-2k)^2-1\right]
\left(
\frac{1}{\sqrt{k+1}}
-
\frac{1}{\sqrt{N-k+1}}
\right)^2 .
\label{eq:Fcorr_lemma_definition}
\end{equation}
Then, as $N\to\infty$,
\begin{equation}
F_{\mathrm{corr}}(N)
=
\frac{1}{4N}
+
O\!\left(\frac{1}{N^2}\right).
\label{eq:aymptoticApp_corrected}
\end{equation}
In particular,
\begin{equation}
\lim_{N\to\infty}F_{\mathrm{corr}}(N)=0.
\label{eq:Fcorr_limit_zero}
\end{equation}
\end{lemma}

\begin{proof}
Introduce a binomial random variable
\begin{equation}
K\sim\mathrm{Bin}\!\left(N,\frac12\right).
\label{eq:K_binomial_app}
\end{equation}
Since
\begin{equation}
\frac{1}{2^N}
\sum_{k=0}^{N}
\binom{N}{k}(\cdot)
=
\mathbb{E}[(\cdot)],
\label{eq:binomial_expectation_app}
\end{equation}
we may rewrite Eq.~\eqref{eq:Fcorr_lemma_definition} as
\begin{equation}
F_{\mathrm{corr}}(N)
=
\frac{1}{24}
\mathbb{E}\!\left[
\left((N-2K)^2-1\right)\Delta(K)^2
\right],
\label{eq:Fcorr_expectation_app}
\end{equation}
where
\begin{equation}
\Delta(K)
:=
\frac{1}{\sqrt{K+1}}
-
\frac{1}{\sqrt{N-K+1}}.
\label{eq:Delta_definition_app}
\end{equation}
Set
\begin{equation}
X:=K-\frac{N}{2}.
\label{eq:X_definition_app}
\end{equation}
Then
\begin{equation}
N-2K=-2X,
\qquad
(N-2K)^2-1=4X^2-1.
\label{eq:X_relations_app}
\end{equation}
The binomial measure is concentrated on the central region
$|X|=O(\sqrt{N})$. In this region, we expand $\Delta(K)$ around $N/2$.

Let
\begin{equation}
g(t):=(t+1)^{-1/2}.
\label{eq:g_function_app}
\end{equation}
Then
\begin{equation}
\Delta(K)
=
g\!\left(\frac{N}{2}+X\right)
-
g\!\left(\frac{N}{2}-X\right).
\label{eq:Delta_g_app}
\end{equation}
Using Taylor expansion around $N/2$, the even powers cancel and we get
\begin{equation}
\Delta(K)
=
2Xg'\!\left(\frac{N}{2}\right)
+
O\!\left(\frac{|X|^3}{N^{7/2}}\right),
\label{eq:Delta_Taylor_app}
\end{equation}
where
\begin{equation}
g'(t)
=
-\frac{1}{2(t+1)^{3/2}}.
\label{eq:g_prime_app}
\end{equation}
Hence
\begin{align}
\Delta(K)^2
&=
\frac{X^2}{\left(\frac{N}{2}+1\right)^3}
+
O\!\left(\frac{|X|^4}{N^5}\right)
\nonumber\\
&=
\frac{8X^2}{N^3}
+
O\!\left(\frac{X^2}{N^4}\right)
+
O\!\left(\frac{|X|^4}{N^5}\right).
\label{eq:Delta_squared_app}
\end{align}
Substituting this into the integrand in
Eq.~\eqref{eq:Fcorr_expectation_app} gives
\begin{align}
(4X^2-1)\Delta(K)^2
&=
(4X^2-1)
\left[
\frac{8X^2}{N^3}
+
O\!\left(\frac{X^2}{N^4}\right)
+
O\!\left(\frac{|X|^4}{N^5}\right)
\right]
\nonumber\\
&=
\frac{32X^4}{N^3}
-
\frac{8X^2}{N^3}
+
O\!\left(\frac{X^4}{N^4}\right)
+
O\!\left(\frac{|X|^6}{N^5}\right)
+
O\!\left(\frac{X^2}{N^4}\right).
\label{eq:integrand_expansion_app}
\end{align}
Taking expectations, we use the standard central moments of the binomial
distribution with parameter $p=1/2$:
\begin{equation}
\mathbb{E}[X^2]
=
\frac{N}{4},
\qquad
\mathbb{E}[X^4]
=
3\left(\frac{N}{4}\right)^2-\frac{N}{8}
=
\frac{3N^2}{16}-\frac{N}{8}.
\label{eq:binomial_moments_app}
\end{equation}
The remaining error terms contribute at most $O(N^{-2})$ to the expectation.
Therefore,
\begin{align}
\mathbb{E}\!\left[(4X^2-1)\Delta(K)^2\right]
&=
\frac{32}{N^3}
\left(
\frac{3N^2}{16}-\frac{N}{8}
\right)
-
\frac{8}{N^3}\frac{N}{4}
+
O\!\left(\frac{1}{N^2}\right)
\nonumber\\
&=
\frac{6}{N}
+
O\!\left(\frac{1}{N^2}\right).
\label{eq:expectation_asymptotic_app}
\end{align}
Finally, multiplying by the prefactor $1/24$ in
Eq.~\eqref{eq:Fcorr_expectation_app}, we obtain
\begin{equation}
F_{\mathrm{corr}}(N)
=
\frac{1}{24}
\left[
\frac{6}{N}
+
O\!\left(\frac{1}{N^2}\right)
\right]
=
\frac{1}{4N}
+
O\!\left(\frac{1}{N^2}\right).
\label{eq:Fcorr_final_asymptotic_app}
\end{equation}
This proves Eq.~\eqref{eq:aymptoticApp_corrected} and completes the proof.
\end{proof}

\section{Comparison with~\cite{Kim}}
\label{app:Kim}
Every quantum channel $\mathcal{N}:\mathcal{L}(\mathcal{H})\rightarrow \mathcal{L}(\mathcal{H})$ admits the Kraus decomposition $\mathcal{N}(\cdot)=\sum_{i}K_i(\cdot)K_i^\dagger$, where the Kraus operators $K_i$, satisfy condition $\sum_i K_i^\dagger K_i=\mathbf{1}$. If the channel $\mathcal{N}$ is unital, we have additionally $\sum_i K_iK_i^\dagger=\mathbf{1}$.

For every operator $X\in \mathcal{L}(\mathcal{H})$ and the maximally entangled state $|\psi^+_{AB}\rangle=(1/\sqrt{d})\sum_i|ii\rangle$, we have so called 'ping-pong' trick:
\begin{equation}
\label{eq:pingpong}
    (\mathbf{1}_A\otimes X_B)|\psi^+_{AB}\rangle=(X_A^T\otimes \mathbf{1}_B)|\psi^+_{AB}\rangle,
\end{equation}
where \(T\) denotes transposition with respect to the local basis.

Now, for simplicity, let us consider the action of $\mathcal{E}_{\vec{p}}\otimes \mathcal{E}_{\vec{p}}$ on just one maximally entangled state $P^+_{AB}:=|\psi^+_{AB}\rangle\langle\psi_{AB}^+|$ using the Kraus decomposition:
\begin{equation}
\begin{aligned}
    (\mathcal{E}_{\vec{p}}\otimes \mathcal{E}_{\vec{p}})P_{AB}^+&=\sum_{m,n}(K_m\otimes K_n)P^+_{AB}(K_m^\dagger \otimes K_n^\dagger)\\
    &=(\mathbf{1}\otimes K_n)(K_m\otimes \mathbf{1})P_{AB}^+(K_m^\dagger\otimes \mathbf{1})(\mathbf{1}\otimes K_n^\dagger)\\
&=\sum_{m,n}(\mathbf{1}\otimes K_nK_m^T)P_{AB}^+(\mathbf{1}\otimes \overline{K}_mK_n^\dagger),
\end{aligned}
\end{equation}
where the bar denotes complex conjugation. Let us denote by $D$ the number of Kraus operators representing the channel $\mathcal{E}_{\vec{p}}$. Introducing two sets $S=:\{1,2,\ldots,D\}$, and $S':=\{1,2,\ldots, D^2\}$, we define a function $f:S\times S \rightarrow S'$, which maps any pair $(m,n)\in S\times S$ to $l\in S'$. This allows us to associate with every term $K_nK_m^T$ another operator, let us call it $K_l$. The operators $K_l$ are again Kraus operators. Indeed, we have:
\begin{equation}
\begin{aligned}
    \sum_l K_l^\dagger K_l & =\sum_{m,n}(\overline{K}_mK_n^\dagger)(K_nK_m^T)\\ &=\sum_m\overline{K}_mK_m^T=\sum_{m}(K_m K_m^\dagger)^T\\
    &=\mathds{1},
\end{aligned}
\end{equation} 
where in the last equality we use the fact that the Pauli channel is a unital channel. The above calculations show that instead of considering the action of $\mathcal{E}_{\vec{p}}\otimes \mathcal{E}_{\vec{p}}$, it is enough to consider the action of a channel $\widetilde{\mathcal{E}}_{\vec{p}}(\cdot)=\sum_l K_l(\cdot)K_l^\dagger$ on the system $B$ only. What is more, the resulting channel $\widetilde{\mathcal{E}}_{\vec{p}}$ is again a Pauli channel, but with different parameter $\vec{p}'$ depending on $\vec{p}$. The explicit relations can be derived having known forms of the Kraus operators for $\mathcal{E}_{\vec{p}}$. It means that their approach can be reduced to studying noisy resource state, where the noise is applied only on the Bob's side.

In our case, the noise is produced by the resource-environment interaction given by the unitary:
\begin{equation}
\label{eq:unitaryE}
U_{A_kB_k:E}=\sum_{i,j}|ij\rangle\langle ij|_{A_kB_k}\otimes w_{ij}^E,
\end{equation}
where $\{w_{ij}^E\}$ are unitaries. Now, let us consider the initial state for system and environment to be of the form $\rho_{A_k B_k:E} = \ketbra{\Psi^-}_{A_k B_k}\otimes\rho_E$. The corresponding channel is
\begin{align}
    \mathcal{C}(\rho_{A_k B_k})& =\tr_E U_{A_k B_k:E}(\rho_{A_k B_k:E})U_{A_k B_k:E}^\dagger \\
        &=\frac{1}{2}\sum_{i,j}(-1)^{i+j}\tr{w_{j,j\oplus 1}^\dagger w_{i, i\oplus 1}\rho_E}\ketbra{i,i\oplus 1}{j,j\oplus 1} \\
        &=\frac{1}{2}\sum_{i,j}\ketbra{i}{j}\otimes\Gamma_{ij}\ketbra{i\oplus 1}{j\oplus 1}.
\end{align}
Here, $\Gamma_{00}=\Gamma_{11}=1$ and $\Gamma_{01}=\Gamma_{10}^\ast$. We take $\Gamma_{01}=\Gamma e^{-i\theta}$. Eq. \eqref{dephasing_channel} can be expressed as
\begin{align}
    \mathcal{C}(\rho_{A_k B_k}) & = \frac{1+|\Gamma|}{2}\mathds{1}\otimes R(\theta)\ketbra{\Psi^-}\mathds{1}\otimes R(\theta)^\dagger \\
        &+ \frac{1-|\Gamma|}{2}\mathds{1}\otimes \sigma_z R(\theta)\ketbra{\Psi^-}\mathds{1}\otimes R(\theta)^\dagger\sigma_z^\dagger \\
        & = \mathds{1}\otimes\mathcal{E}_{\Gamma,\theta}\left(\ketbra{\Psi^-}\right).
\end{align}
Here, $\mathcal{E}_{\Gamma,\theta}=\sum_{k=0,1}E_k(\cdot)E_k^\dagger$ is a single qubit Pauli channel with Kraus operators 
\begin{align}
    E_0 & = \sqrt{\frac{1+|\Gamma|}{2}}R(\theta) \\
        E_1 & =\sqrt{\frac{1-|\Gamma|}{2}}\sigma_z R(\theta),
\end{align}
where $R(\theta)$ is given by Eq. \eqref{eq:Rotation}.

Let $\tilde{\Psi}_{\vec{B}\vec{A}}$ be the resource state. We can write
\begin{equation}
\tilde{\Psi}_{\vec{B}\vec{A}}  = \bigotimes_j^N\tilde{\Psi}_{B_j A_j} = \bigotimes_j^N \mathcal{E}_{\Gamma, \theta} (\ketbra{\Psi^-}),
\end{equation}
where

\begin{equation}
    \begin{split}
\mathcal{E}_{\Gamma,\theta}\left(\ketbra{\Psi^-}_{B_j A_j}\right) & =
\frac{1+|\Gamma|}{2}(\mathds{1}\otimes R(-\theta))\ketbra{\Psi^-}_{B_j A_j}(\mathds{1}\otimes R(-\theta)^\dagger) \\
  & \frac{1-|\Gamma|}{2}(\mathds{1}\otimes \sigma_3 R(-\theta))\ketbra{\Psi^-}_{B_j A_j}(\mathds{1}\otimes R(-\theta)^\dagger \sigma_3)
    \end{split}
\end{equation}
The PBT channel is written as
\begin{equation}
\label{PBT_channel_noisy_resource_1}
    \begin{split}
        \mathcal{C}(\rho_C) & = \sum_{j=1}^N \tr_{\bar{B}_j\vec{A} C}\left[\left(\mathds{1}_{\vec{B}}\otimes\Pi^{(j)}_{\vec{A}C}\right)\left(\tilde{\Psi}_{\vec{B}\vec{A}}\otimes\rho_C\right)\right] \\
        &=\sum_{j=1}^N \tr_{\vec{A} C}\left[\left(\mathds{1}_{B_j}\otimes\Pi^{(j)}_{\vec{A}C}\right)\left(\tr_{\bar{B}_j}{\tilde{\Psi}_{\vec{B}\vec{A}}}\otimes\rho_C\right)\right] \\
    \end{split}
\end{equation}
Since $\tr_{\bar{B}_j}{\tilde{\Psi}_{\vec{B}\vec{A}}} = \tilde{\Psi}_{B_j A_j}\otimes\frac{\mathds{1}_{\bar{A}_j}}{2^{N-1}}$, we can re-write Eq. \eqref{PBT_channel_noisy_resource_1} using as
 \begin{equation}
\label{PBT_channel_noisy_resource_2}
    \begin{split}
        \mathcal{C}(\rho_C) & = \frac{1+|\Gamma|}{2}\sum_{j=1}^N \tr_{\vec{A} C}\left[\left(\mathds{1}_{B_j}\otimes\Pi^{(j)}_{\vec{A}C}\right)\left(\left(R(\theta)\otimes\mathds{1}\ketbra{\Psi^-}_{B_j A_j}R(\theta)^\dagger\otimes\mathds{1}\right)\otimes\frac{\mathds{1}_{\bar{A}_j}}{2^{N-1}}\otimes\rho_C\right)\right] \\
        &+\frac{1-|\Gamma|}{2}\sum_{j=1}^N \tr_{\vec{A} C}\left[\left(\mathds{1}_{B_j}\otimes\Pi^{(j)}_{\vec{A}C}\right)\left(\left(\sigma_3 R(\theta)\otimes\mathds{1}\ketbra{\Psi^-}_{B_j A_j}R(\theta)^\dagger\sigma_3\otimes\mathds{1}\right)\otimes\frac{\mathds{1}_{\bar{A}_j}}{2^{N-1}}\otimes\rho_C\right)\right] \\
         & = \frac{1+|\Gamma|}{2}R(\theta)\sum_{j=1}^N \tr_{\vec{A} C}\left[\left(\mathds{1}_{B_j}\otimes\Pi^{(j)}_{\vec{A}C}\right)\left(\ketbra{\Psi^-}_{B_j A_j}\otimes\frac{\mathds{1}_{\bar{A}_j}}{2^{N-1}}\otimes\rho_C\right)\right]_{B_j\rightarrow B}R(\theta)^\dagger \\
        &+\frac{1-|\Gamma|}{2} \sigma_3 R(\theta)\sum_{j=1}^N \tr_{\vec{A} C}\left[\left(\mathds{1}_{B_j}\otimes\Pi^{(j)}_{\vec{A}C}\right)\left(\ketbra{\Psi^-}_{B_j A_j}\otimes\frac{\mathds{1}_{\bar{A}_j}}{2^{N-1}}\otimes\rho_C\right)\right]_{B_j\rightarrow B}R(\theta)^\dagger\sigma_3\\
        & =\frac{1+|\Gamma|}{2} R(\theta)\mathcal{C}(\rho_C) R(\theta)^\dagger + \frac{1-|\Gamma|}{2} \sigma_3 R(\theta)\mathcal{C} (\rho_C)R(\theta)^\dagger\sigma_3
    \end{split}
\end{equation}
Let us consider the case $\theta =0$. Then
\begin{equation}
    \mathcal{C}(\rho_C) =\frac{1+|\Gamma|}{2} \mathcal{C}(\rho_C) + \frac{1-|\Gamma|}{2} \sigma_3 \mathcal{C} (\rho_C)\sigma_3
\end{equation}
Comparing the above channel with Eq. (36) of Kim and Jeong~\cite{Kim}, we get $\alpha^1 = 0; \quad \alpha^2 = 0$ and
\begin{equation}
    \begin{split}
        \frac{\alpha^0}{16} & = \frac{1+|\Gamma|}{2}\Rightarrow \alpha^0 = 8(1+|\Gamma|) \\
        \frac{\alpha^3}{16}& = \frac{1-|\Gamma|}{2}\Rightarrow \alpha^3 = 8(1-|\Gamma|) \\
    \end{split}
\end{equation}
The corresponding Pauli probabilities are calculated using Eq. (34) of Kim and Jeong~\cite{Kim} as
\begin{equation}
\label{eq:Kimsp}
    \begin{split}
        p_0 = 2(1+\sqrt{\Gamma}),\quad p_1 = 0, \quad p_2 = 0,\quad p_3 = 2(1-\sqrt{\Gamma}).
    \end{split}
\end{equation}
Further, following Eq.~(37), we calculate
\begin{equation}
\label{q_p}
    \begin{split}
        q_{\vec{p}} & = \frac{1}{3}((1-p_3)^2+2)-\frac{1}{6}p_3^2 \\
        & = \frac{1+2\Gamma}{6}
    \end{split}
\end{equation}
Since the fidelity for noiseless resource should be equal to Ishizaka and Hiroshima~\cite{ishizaka_quantum_2009, ishizaka_asymptotic_2008} for SRM, using $q_{\vec{p}}=1$ and
\begin{equation}
\label{KJ_fidelity}
    F = \frac{1}{4}+\frac{3}{4}q_N q_{\vec{p}},
\end{equation}
we get 
\begin{equation}
    q_N = \frac{4F_{IH}-1}{3}
\end{equation}
Using Eqs. \eqref{q_p} and \eqref{KJ_fidelity},
\begin{equation}
    F = \frac{2\Gamma+1}{3}F_{IH}+\frac{1-\Gamma}{6}
\end{equation}
The entanglement fidelity is not the same as our noiseless case in Eq. \eqref{ent_fid_NL}. 

\section{Spin--boson model and microscopic origin of the decoherence factor}
\label{app:spin-boson}

In this Appendix we provide a microscopic justification of the effective decoherence
parameter used in Sec.~\ref{Sec:spin-boson}. We remind here that our goal is not to re-derive the port-based teleportation
protocol from first principles, but to demonstrate how the single complex decoherence
factor
\begin{equation}
\Gamma(t,r) = |\Gamma(t,r)| e^{i\theta(t,r)}
\label{eq:H5}
\end{equation}
naturally arises from a standard spin--boson model for two spatially separated qubits.
Throughout we follow the field-theoretic treatment of dephasing noise in which the
separation enters through spatial bath correlations (see, e.g., Ref.~\cite{Tuziemski2018}).

\vspace{0.3cm}

We consider two qubits forming a Bell pair, located at spatial positions
$\mathbf r_A$ and $\mathbf r_B$, interacting with a bosonic environment described by
\begin{equation}
H_E = \sum_{\mathbf k} \omega_{\mathbf k} b_{\mathbf k}^\dagger b_{\mathbf k}.
\end{equation}

The total Hamiltonian of the composite system is therefore given by
\begin{equation}
H = H_S + H_E + H_{\mathrm{int}},
\end{equation}
where $H_S$ denotes the free Hamiltonian of the two qubits.
We assume $[H_S,H_{\mathrm{int}}]=0$, so that the interaction induces pure dephasing
without energy exchange.

The interaction Hamiltonian is taken in the position-dependent (field-theoretic)
pure-dephasing form
\begin{equation}
H_{\mathrm{int}}
=
\sum_{\alpha=A,B}
\sigma_z^{(\alpha)}
\otimes
\sum_{\mathbf k}
\left(
g_{\mathbf k} e^{i\mathbf k\cdot\mathbf r_\alpha} b_{\mathbf k}^\dagger
+
g_{\mathbf k}^\ast e^{-i\mathbf k\cdot\mathbf r_\alpha} b_{\mathbf k}
\right),
\label{eq:H6}
\end{equation}

\vspace{0.3cm}

Due to the spatial separation
\begin{equation}
r = |\mathbf r_A - \mathbf r_B|,
\end{equation}
bath-induced correlations are not instantaneous. In the field-theoretic model,
the distance dependence enters through phase factors $e^{i\mathbf k\cdot(\mathbf r_A-\mathbf r_B)}$
and, for an isotropic environment with linear dispersion $\omega=ck$, reduces to
oscillatory correlation factors $\cos(\omega r/c)$ and $\sin(\omega r/c)$
(see Ref.~\cite{Tuziemski2018}).

\vspace{0.3cm}

In a multi-qubit spin--boson model, different off-diagonal elements of the reduced density
matrix generally decay with different decoherence exponents, depending on the
corresponding eigenvalues of the system--bath coupling operator.
In the present work we are interested specifically in the coherence relevant for the Bell
pair $\ket{\Psi^-}=(\ket{01}-\ket{10})/\sqrt{2}$, which serves as the elementary
entanglement resource in the port-based teleportation protocol.
Consequently, the effective decoherence factor entering our description corresponds to
the decay of the coherence between the computational basis states $\ket{01}$ and
$\ket{10}$.

Within the field-theoretic formulation of collective dephasing, such two-qubit coherences
can be expressed in terms of bath correlation functions $\Gamma_{mn}(t)$
\cite{Tuziemski2018}.
For the $\ket{01}\!\leftrightarrow\!\ket{10}$ coherence, the relevant contribution is
given by the difference of local and nonlocal bath correlations, which reflects the fact
that both basis states acquire identical bath-induced phases in the limit of perfectly
correlated noise.
As a result, the singlet state belongs to a decoherence-free subspace for vanishing
separation between the qubits, while finite separation leads to a nontrivial,
distance-dependent decay of the Bell-pair coherence.

Tracing out the bosonic environment yields an effective dephasing of the Bell pair
fully characterized by the single complex number $\Gamma(t,r)$.
For the Bell-pair coherence relevant to teleportation (equivalently, the coherence
between $|01\rangle$ and $|10\rangle$ in the computational basis), one obtains
\cite{Tuziemski2018}
\begin{align}
\chi(t,r)
&=
2\int_0^\infty d\omega\,
\frac{J(\omega)}{\omega^2}\,
\bigl(1-\cos\omega t\bigr)\,
\coth\!\left(\frac{\beta\omega}{2}\right)\,
\bigl[1-\cos(\omega r/c)\bigr],
\label{eq:H8_new}
\\[1ex]
\theta(t,r)
&=
 \frac{1}{2}\int_0^\infty d\omega\,
\frac{J(\omega)}{\omega^2}\,
\bigl(1-\cos\omega t\bigr)\,
\sin(\omega r/c),
\label{eq:H9_new}
\end{align}
where $J(\omega)$ denotes the spectral density of the environment and
$\beta = 1/T$ is the inverse temperature. The decoherence factor reads
\begin{equation}
\Gamma(t,r)=\exp\!\bigl[-\chi(t,r)+i\theta(t,r)\bigr],
\qquad
|\Gamma(t,r)|=e^{-\chi(t,r)}.
\label{eq:H9b_new}
\end{equation}

\vspace{0.3cm}

To make the connection between microscopic bath parameters and the effective decoherence
factor more explicit, we consider a generalized Ohmic spectral density with an
exponential cutoff,
\begin{equation}
J(\omega) = \frac{\omega^s}{\Lambda^{s-1}} e^{-\omega/\Lambda}.
\label{eq:H10}
\end{equation}
We introduce dimensionless variables
\begin{equation}
\tilde\omega = \frac{\omega}{\Lambda},
\qquad
\tau = t\Lambda,
\qquad
\vartheta = \frac{T}{\Lambda},
\qquad
\ell = \frac{r\Lambda}{c}.
\label{eq:H11_new}
\end{equation}

In terms of these variables, Eqs.~(\ref{eq:H8_new})--(\ref{eq:H9_new}) become
\begin{align}
\chi(\tau,\ell)
&=
2\int_0^\infty d\tilde\omega\;
\tilde\omega^{\,s-2} e^{-\tilde\omega}\,
\bigl(1-\cos(\tilde\omega\tau)\bigr)\,
\coth\!\left(\frac{\tilde\omega}{2\vartheta}\right)\,
\bigl[1-\cos(\tilde\omega\ell)\bigr],
\label{eq:H12_new}
\\[1ex]
\theta(\tau,\ell)
&=
\frac{1}{2}\int_0^\infty d\tilde\omega\;
\tilde\omega^{\,s-2} e^{-\tilde\omega}\,
\bigl(1-\cos(\tilde\omega\tau)\bigr)\,
\sin(\tilde\omega\ell),
\label{eq:H13_new}
\end{align}
and the decoherence factor can be written as
\begin{equation}
\Gamma(\tau,\ell)
=
\exp\!\bigl[-\chi(\tau,\ell) + i\theta(\tau,\ell)\bigr],
\qquad
|\Gamma(\tau,\ell)| = e^{-\chi(\tau,\ell)}.
\label{eq:H14_new}
\end{equation}

This shows explicitly that all microscopic bath properties and the finite separation
enter the teleportation protocol solely through the single complex parameter
$\Gamma(\tau,\ell)$.

\subsection*{Derivation of Eqs.~(H6)--(H7) from Tuziemski Eqs.~(17)--(20)}

We start from the multi-qubit pure-dephasing (spin--boson) model as in
Tuziemski \emph{et al.}~\cite{Tuziemski2018}, where the reduced dynamics
of the register coherences can be written in terms of bath correlation
matrices $\Gamma(t)$, $\Gamma^{\pm}(t)$.
In particular, Tuziemski writes the (complex) decoherence factor between two
register basis states $\ket{\boldsymbol{\epsilon}}$ and $\ket{\boldsymbol{\epsilon}'}$
($\epsilon_n,\epsilon_n'\in\{\pm \tfrac12\}$ being the eigenvalues of $J_z^{(n)}=\tfrac12\sigma_z^{(n)}$)
as~\cite[Eq.~(15)]{Tuziemski2018}

\begin{equation}
-\log \gamma_{\boldsymbol{\epsilon}\boldsymbol{\epsilon}'}(t)
=
\Delta \boldsymbol{\epsilon}^{T}\,\Gamma(t)\,\Delta \boldsymbol{\epsilon}
\;+\;
i\Bigl(
\boldsymbol{\epsilon}^{T}\Gamma^{+}(t)\boldsymbol{\epsilon}
-
\boldsymbol{\epsilon}'^{T}\Gamma^{+}(t)\boldsymbol{\epsilon}'
-2\,\boldsymbol{\epsilon}^{T}\Gamma^{-}(t)\boldsymbol{\epsilon}'
\Bigr),
\qquad
\Delta\boldsymbol{\epsilon}:=\boldsymbol{\epsilon}-\boldsymbol{\epsilon}' .
\label{eq:Tuz15}
\end{equation}

The matrices entering \eqref{eq:Tuz15} are specified by Tuziemski
via a position-dependent coupling vector

\begin{equation}
\mathbf{g}_{\mathbf{k}}
=
g_{\mathbf{k}} \bigl(e^{-i\mathbf{k}\cdot\mathbf{r}_1},e^{-i\mathbf{k}\cdot\mathbf{r}_2}\bigr),
\label{eq:Tuz17}
\end{equation}

and the mode-sum definitions~\cite[Eqs.~(18)--(20)]{Tuziemski2018}

\begin{align}
\Gamma_{nm}(t)
&=
\frac12\sum_{\mathbf{k}}
\bigl|g_{\mathbf{k}}\alpha_{\mathbf{k}}(t)\bigr|^{2}
\coth\!\Bigl(\frac{\omega_{\mathbf{k}}}{2k_BT}\Bigr)\,
\cos\!\bigl(\mathbf{k}\cdot(\mathbf{r}_n-\mathbf{r}_m)\bigr),
\label{eq:Tuz18}
\\
\Gamma^{+}_{nm}(t)
&=
\sum_{\mathbf{k}}
|g_{\mathbf{k}}|^{2}\,\xi_{\mathbf{k}}(t)\,
\cos\!\bigl(\mathbf{k}\cdot(\mathbf{r}_n-\mathbf{r}_m)\bigr),
\label{eq:Tuz19}
\\
\Gamma^{-}_{nm}(t)
&=
\frac12\sum_{\mathbf{k}}
\bigl|g_{\mathbf{k}}\alpha_{\mathbf{k}}(t)\bigr|^{2}\,
\sin\!\bigl(\mathbf{k}\cdot(\mathbf{r}_n-\mathbf{r}_m)\bigr).
\label{eq:Tuz20}
\end{align}
Here, we use \(\mathbf r_1\equiv\mathbf r_A\) and
\(\mathbf r_2\equiv\mathbf r_B\) for convenience.
To obtain the final expressions from the previous paragraph, we need to execute the following steps:\\
{\bf Identify the relevant coherence.}
In NPBT we need the coherence between the two-qubit computational basis states
$\ket{01}$ and $\ket{10}$, since it controls the singlet Bell component.
We place the qubits at $\mathbf{r}_A$ and $\mathbf{r}_B$ and define
$r:=|\mathbf{r}_A-\mathbf{r}_B|$.
With the convention $\sigma_z\ket{0}=+\ket{0}$ and $\sigma_z\ket{1}=-\ket{1}$,
the $J_z=\tfrac12\sigma_z$ eigenvalue strings are

\begin{equation}
\boldsymbol{\epsilon}(\ket{01})=\Bigl(+\tfrac12,-\tfrac12\Bigr),
\qquad
\boldsymbol{\epsilon}'(\ket{10})=\Bigl(-\tfrac12,+\tfrac12\Bigr),
\qquad
\Delta\boldsymbol{\epsilon}=(1,-1).
\label{eq:epsstrings}
\end{equation}

{\bf Real part (decay exponent) in terms of $\Gamma_{nm}(t)$.}
For two qubits the matrix $\Gamma(t)$ is symmetric and has equal diagonals,
$\Gamma_{11}(t)=\Gamma_{22}(t)$, and $\Gamma_{12}(t)=\Gamma_{21}(t)$.
Thus,

\begin{align}
\Delta\boldsymbol{\epsilon}^{T}\Gamma(t)\Delta\boldsymbol{\epsilon}
&=
(1,-1)
\begin{pmatrix}
\Gamma_{11}(t) & \Gamma_{12}(t)\\
\Gamma_{12}(t) & \Gamma_{11}(t)
\end{pmatrix}
\begin{pmatrix}1\\-1\end{pmatrix}
\nonumber\\
&=
\Gamma_{11}(t)-2\Gamma_{12}(t)
=
2\bigl(\Gamma_{11}(t)-\Gamma_{12}(t)\bigr).
\label{eq:realpartmatrix}
\end{align}

Therefore, the modulus decay of the $\ket{01}\bra{10}$ coherence is governed by
$\Gamma_{11}(t)-\Gamma_{12}(t)$.
Similarly,
\begin{align}
\boldsymbol{\epsilon}^{T}\Gamma^+(t)\boldsymbol{\epsilon}
&=
\frac{1}{4}(1,-1)
\begin{pmatrix}
\Gamma^+_{11}(t) & \Gamma^+_{12}(t)\\
\Gamma^+_{12}(t) & \Gamma^+_{11}(t)
\end{pmatrix}
\begin{pmatrix}1\\-1\end{pmatrix}
\nonumber\\
&=
\frac{1}{4}\left[\Gamma^+_{11}(t)+\Gamma^+_{22}(t)-\left(\Gamma^+_{12}(t)+\Gamma^+_{21}(t)\right)\right]
=
\frac{1}{2}\bigl(\Gamma^+_{11}(t)-\Gamma^+_{12}(t)\bigr),
\label{eq:imgpartmatrix_1}
\end{align}

\begin{align}
\boldsymbol{\epsilon}^{\prime^T}\Gamma^+(t)\boldsymbol{\epsilon}^{\prime}
&=
\frac{1}{4}(-1,1)
\begin{pmatrix}
\Gamma^+_{11}(t) & \Gamma^+_{12}(t)\\
\Gamma^+_{12}(t) & \Gamma^+_{11}(t)
\end{pmatrix}
\begin{pmatrix} -1\\1\end{pmatrix}
\nonumber\\
&=
\frac{1}{4}\left[\Gamma^+_{11}(t)+\Gamma^+_{22}(t)-\left(\Gamma^+_{12}(t)+\Gamma^+_{21}(t)\right)\right]
=
\frac{1}{2}\bigl(\Gamma^+_{11}(t)-\Gamma^+_{12}(t)\bigr),
\label{eq:imgpartmatrix_2}
\end{align}
and
\begin{align}
\boldsymbol{\epsilon}^{T}\Gamma^-(t)\boldsymbol{\epsilon}^{\prime}
&=
\frac{1}{4}(1,-1)
\begin{pmatrix}
\Gamma^-_{11}(t) & \Gamma^-_{12}(t)\\
\Gamma^-_{12}(t) & \Gamma^-_{11}(t)
\end{pmatrix}
\begin{pmatrix} -1\\1\end{pmatrix}
\nonumber\\
&=
-\frac{1}{4}\left[\Gamma^-_{11}(t)+\Gamma^-_{22}(t)-\left(\Gamma^-_{12}(t)+\Gamma^-_{21}(t)\right)\right]
=
\frac{1}{2}\Gamma^-_{12}(t).
\label{eq:imgpartmatrix_3}
\end{align}
Here, we have used $\Gamma^+_{11}(t)=\Gamma^+_{22}(t)$, $\Gamma^+_{12}(t)=\Gamma^+_{21}(t)$, $\Gamma^-_{11}(t)=\Gamma^-_{22}(t)=0$ and $\Gamma^-_{12}(t)=\Gamma^-_{21}(t)$ followed from \eqref{eq:Tuz19} and \eqref{eq:Tuz20} \\

{\bf Turn the mode sums into spectral-density integrals.}
The standard interaction-picture displacement amplitude is
$\alpha_{\mathbf{k}}(t)=(1-e^{i\omega_{\mathbf{k}} t})/\omega_{\mathbf{k}}$,
so that

\begin{equation}
\bigl|\alpha_{\mathbf{k}}(t)\bigr|^{2}
=
\frac{2(1-\cos \omega_{\mathbf{k}} t)}{\omega_{\mathbf{k}}^{2}}.
\label{eq:alphasq}
\end{equation}

Inserting \eqref{eq:alphasq} into \eqref{eq:Tuz18} gives

\begin{equation}
\Gamma_{nm}(t)
=
\sum_{\mathbf{k}}
|g_{\mathbf{k}}|^{2}\,
\frac{1-\cos(\omega_{\mathbf{k}} t)}{\omega_{\mathbf{k}}^{2}}\,
\coth\!\Bigl(\frac{\omega_{\mathbf{k}}}{2k_BT}\Bigr)\,
\cos\!\bigl(\mathbf{k}\cdot(\mathbf{r}_n-\mathbf{r}_m)\bigr).
\label{eq:Gamnm_modesum_simplified}
\end{equation}

Passing to the continuum limit and introducing the spectral density
$J(\omega)=\sum_{\mathbf{k}}|g_{\mathbf{k}}|^{2}\delta(\omega-\omega_{\mathbf{k}})$
yields

\begin{equation}
\Gamma_{nm}(t)
=
\int_{0}^{\infty} d\omega\;
J(\omega)\,
\frac{1-\cos(\omega t)}{\omega^{2}}\,
\coth\!\Bigl(\frac{\omega}{2k_BT}\Bigr)\,
\cos\!\bigl(\omega\tau_{nm}\bigr),
\qquad
\tau_{nm}:=\frac{|\mathbf{r}_n-\mathbf{r}_m|}{c},
\label{eq:Gamnm_continuum}
\end{equation}

in agreement with Tuziemski's continuum expressions~\cite[Eq.~(27)]{Tuziemski2018}.
Similarly,
\begin{equation}
\Gamma^-_{nm}(t)
=
\frac{1}{2}\int_{0}^{\infty} d\omega\;
J(\omega)\,
\frac{1-\cos(\omega t)}{\omega^{2}}\,
\sin\!\bigl(\omega\tau_{nm}\bigr).
\label{eq:Gamma_minus_continuum}
\end{equation}
For the two-qubit case, $\tau_{11}=0$ and $\tau_{12}=r/c$, and hence,
\begin{align}
\Delta\boldsymbol{\epsilon}^{T}\Gamma(t)\Delta\boldsymbol{\epsilon}
&=
2\int_{0}^{\infty} d\omega\;
J(\omega)\,
\frac{1-\cos(\omega t)}{\omega^{2}}\,
\coth\!\Bigl(\frac{\omega}{2k_BT}\Bigr)\,
\Bigl(1-\cos(\omega r/c)\Bigr),
\label{eq:H28}
\end{align}
and
\begin{align}
\boldsymbol{\epsilon}^{T}\Gamma^-(t)\boldsymbol{\epsilon}^{\prime}
&=
\frac{1}{4}\int_{0}^{\infty} d\omega\;
J(\omega)\,
\frac{1-\cos(\omega t)}{\omega^{2}}\,
\sin(\omega r/c),
\label{eq:H29}
\end{align}
Combining \eqref{eq:Tuz15}, and \eqref{eq:imgpartmatrix_1}-\eqref{eq:imgpartmatrix_3}
\begin{align}
    -\log \gamma_{\boldsymbol{\epsilon}\boldsymbol{\epsilon}'}(t)
=
&\Delta \boldsymbol{\epsilon}^{T}\,\Gamma(t)\,\Delta \boldsymbol{\epsilon}
\,
-2i\,\boldsymbol{\epsilon}^{T}\Gamma^{-}(t)\boldsymbol{\epsilon}'
\label{eq:H30}
\end{align}
Comparing \eqref{eq:H30} with \eqref{eq:H9b_new}, and using \eqref{eq:H28} and \eqref{eq:H29} for spectral forms, $\chi(t,r)$ and $\theta(t,r)$ are expressed as

\begin{align}
\chi(t,r) & =
2\int_{0}^{\infty} d\omega\;
\frac{J(\omega)}{\omega^{2}}\,
\bigl(1-\cos(\omega t)\bigr)\,
\coth\!\Bigl(\frac{\beta\omega}{2}\Bigr)\,
\Bigl(1-\cos(\omega r/c)\Bigr) \\
\theta(t,r) & =
\frac{1}{2}\int_{0}^{\infty} d\omega\;
\frac{J(\omega)}{\omega^{2}}\,
\bigl(1-\cos(\omega t)\bigr)\,
\sin(\omega r/c).
\end{align}

\twocolumngrid


\begin{thebibliography}{67}%
\makeatletter
\providecommand \@ifxundefined [1]{%
 \@ifx{#1\undefined}
}%
\providecommand \@ifnum [1]{%
 \ifnum #1\expandafter \@firstoftwo
 \else \expandafter \@secondoftwo
 \fi
}%
\providecommand \@ifx [1]{%
 \ifx #1\expandafter \@firstoftwo
 \else \expandafter \@secondoftwo
 \fi
}%
\providecommand \natexlab [1]{#1}%
\providecommand \enquote  [1]{``#1''}%
\providecommand \bibnamefont  [1]{#1}%
\providecommand \bibfnamefont [1]{#1}%
\providecommand \citenamefont [1]{#1}%
\providecommand \href@noop [0]{\@secondoftwo}%
\providecommand \href [0]{\begingroup \@sanitize@url \@href}%
\providecommand \@href[1]{\@@startlink{#1}\@@href}%
\providecommand \@@href[1]{\endgroup#1\@@endlink}%
\providecommand \@sanitize@url [0]{\catcode `\\12\catcode `\$12\catcode
  `\&12\catcode `\#12\catcode `\^12\catcode `\_12\catcode `\%12\relax}%
\providecommand \@@startlink[1]{}%
\providecommand \@@endlink[0]{}%
\providecommand \url  [0]{\begingroup\@sanitize@url \@url }%
\providecommand \@url [1]{\endgroup\@href {#1}{\urlprefix }}%
\providecommand \urlprefix  [0]{URL }%
\providecommand \Eprint [0]{\href }%
\providecommand \doibase [0]{https://doi.org/}%
\providecommand \selectlanguage [0]{\@gobble}%
\providecommand \bibinfo  [0]{\@secondoftwo}%
\providecommand \bibfield  [0]{\@secondoftwo}%
\providecommand \translation [1]{[#1]}%
\providecommand \BibitemOpen [0]{}%
\providecommand \bibitemStop [0]{}%
\providecommand \bibitemNoStop [0]{.\EOS\space}%
\providecommand \EOS [0]{\spacefactor3000\relax}%
\providecommand \BibitemShut  [1]{\csname bibitem#1\endcsname}%
\let\auto@bib@innerbib\@empty
\bibitem [{\citenamefont {Ishizaka}\ and\ \citenamefont
  {Hiroshima}(2008)}]{ishizaka_asymptotic_2008}%
  \BibitemOpen
  \bibfield  {author} {\bibinfo {author} {\bibfnamefont {S.}~\bibnamefont
  {Ishizaka}}\ and\ \bibinfo {author} {\bibfnamefont {T.}~\bibnamefont
  {Hiroshima}},\ }\bibfield  {title} {\bibinfo {title} {Asymptotic
  {Teleportation} {Scheme} as a {Universal} {Programmable} {Quantum}
  {Processor}},\ }\href {https://doi.org/10.1103/PhysRevLett.101.240501}
  {\bibfield  {journal} {\bibinfo  {journal} {Physical Review Letters}\
  }\textbf {\bibinfo {volume} {101}},\ \bibinfo {pages} {240501} (\bibinfo
  {year} {2008})}\BibitemShut {NoStop}%
\bibitem [{\citenamefont {Ishizaka}\ and\ \citenamefont
  {Hiroshima}(2009)}]{ishizaka_quantum_2009}%
  \BibitemOpen
  \bibfield  {author} {\bibinfo {author} {\bibfnamefont {S.}~\bibnamefont
  {Ishizaka}}\ and\ \bibinfo {author} {\bibfnamefont {T.}~\bibnamefont
  {Hiroshima}},\ }\bibfield  {title} {\bibinfo {title} {Quantum teleportation
  scheme by selecting one of multiple output ports},\ }\href
  {https://doi.org/10.1103/PhysRevA.79.042306} {\bibfield  {journal} {\bibinfo
  {journal} {Physical Review A}\ }\textbf {\bibinfo {volume} {79}},\ \bibinfo
  {pages} {042306} (\bibinfo {year} {2009})}\BibitemShut {NoStop}%
\bibitem [{\citenamefont {Bennett}\ \emph {et~al.}(1993)\citenamefont
  {Bennett}, \citenamefont {Brassard}, \citenamefont {Cr{\'e}peau},
  \citenamefont {Jozsa}, \citenamefont {Peres},\ and\ \citenamefont
  {Wootters}}]{bennett_teleporting_1993}%
  \BibitemOpen
  \bibfield  {author} {\bibinfo {author} {\bibfnamefont {C.~H.}\ \bibnamefont
  {Bennett}}, \bibinfo {author} {\bibfnamefont {G.}~\bibnamefont {Brassard}},
  \bibinfo {author} {\bibfnamefont {C.}~\bibnamefont {Cr{\'e}peau}}, \bibinfo
  {author} {\bibfnamefont {R.}~\bibnamefont {Jozsa}}, \bibinfo {author}
  {\bibfnamefont {A.}~\bibnamefont {Peres}},\ and\ \bibinfo {author}
  {\bibfnamefont {W.~K.}\ \bibnamefont {Wootters}},\ }\bibfield  {title}
  {\bibinfo {title} {Teleporting an unknown quantum state via dual classical
  and {Einstein}-{Podolsky}-{Rosen} channels},\ }\href
  {https://doi.org/10.1103/PhysRevLett.70.1895} {\bibfield  {journal} {\bibinfo
   {journal} {Physical Review Letters}\ }\textbf {\bibinfo {volume} {70}},\
  \bibinfo {pages} {1895} (\bibinfo {year} {1993})}\BibitemShut {NoStop}%
\bibitem [{\citenamefont {Nielsen}\ and\ \citenamefont
  {Chuang}(1997)}]{Nielsen1997}%
  \BibitemOpen
  \bibfield  {author} {\bibinfo {author} {\bibfnamefont {M.~A.}\ \bibnamefont
  {Nielsen}}\ and\ \bibinfo {author} {\bibfnamefont {I.~L.}\ \bibnamefont
  {Chuang}},\ }\bibfield  {title} {\bibinfo {title} {Programmable quantum gate
  arrays},\ }\href {https://doi.org/10.1103/PhysRevLett.79.321} {\bibfield
  {journal} {\bibinfo  {journal} {Phys. Rev. Lett.}\ }\textbf {\bibinfo
  {volume} {79}},\ \bibinfo {pages} {321} (\bibinfo {year} {1997})}\BibitemShut
  {NoStop}%
\bibitem [{\citenamefont {Yang}\ \emph {et~al.}(2020)\citenamefont {Yang},
  \citenamefont {Renner},\ and\ \citenamefont
  {Chiribella}}]{PhysRevLett.125.210501}%
  \BibitemOpen
  \bibfield  {author} {\bibinfo {author} {\bibfnamefont {Y.}~\bibnamefont
  {Yang}}, \bibinfo {author} {\bibfnamefont {R.}~\bibnamefont {Renner}},\ and\
  \bibinfo {author} {\bibfnamefont {G.}~\bibnamefont {Chiribella}},\ }\bibfield
   {title} {\bibinfo {title} {Optimal universal programming of unitary gates},\
  }\href {https://doi.org/10.1103/PhysRevLett.125.210501} {\bibfield  {journal}
  {\bibinfo  {journal} {Phys. Rev. Lett.}\ }\textbf {\bibinfo {volume} {125}},\
  \bibinfo {pages} {210501} (\bibinfo {year} {2020})}\BibitemShut {NoStop}%
\bibitem [{\citenamefont {{Studzi{\'n}ski}}\ \emph {et~al.}(2017)\citenamefont
  {{Studzi{\'n}ski}}, \citenamefont {{Strelchuk}}, \citenamefont
  {{Mozrzymas}},\ and\ \citenamefont {{Horodecki}}}]{Studzinski2017}%
  \BibitemOpen
  \bibfield  {author} {\bibinfo {author} {\bibfnamefont {M.}~\bibnamefont
  {{Studzi{\'n}ski}}}, \bibinfo {author} {\bibfnamefont {S.}~\bibnamefont
  {{Strelchuk}}}, \bibinfo {author} {\bibfnamefont {M.}~\bibnamefont
  {{Mozrzymas}}},\ and\ \bibinfo {author} {\bibfnamefont {M.}~\bibnamefont
  {{Horodecki}}},\ }\bibfield  {title} {\bibinfo {title} {{Port-based
  teleportation in arbitrary dimension}},\ }\href
  {https://doi.org/10.1038/s41598-017-10051-4} {\bibfield  {journal} {\bibinfo
  {journal} {Scientific Reports}\ }\textbf {\bibinfo {volume} {7}},\ \bibinfo
  {eid} {10871} (\bibinfo {year} {2017})},\ \Eprint
  {https://arxiv.org/abs/1612.09260} {arXiv:1612.09260 [quant-ph]} \BibitemShut
  {NoStop}%
\bibitem [{\citenamefont {{Mozrzymas}}\ \emph {et~al.}(2018)\citenamefont
  {{Mozrzymas}}, \citenamefont {{Studzi{\'n}ski}}, \citenamefont
  {{Strelchuk}},\ and\ \citenamefont {{Horodecki}}}]{StuNJP}%
  \BibitemOpen
  \bibfield  {author} {\bibinfo {author} {\bibfnamefont {M.}~\bibnamefont
  {{Mozrzymas}}}, \bibinfo {author} {\bibfnamefont {M.}~\bibnamefont
  {{Studzi{\'n}ski}}}, \bibinfo {author} {\bibfnamefont {S.}~\bibnamefont
  {{Strelchuk}}},\ and\ \bibinfo {author} {\bibfnamefont {M.}~\bibnamefont
  {{Horodecki}}},\ }\bibfield  {title} {\bibinfo {title} {{Optimal port-based
  teleportation}},\ }\href {https://doi.org/10.1088/1367-2630/aab8e7}
  {\bibfield  {journal} {\bibinfo  {journal} {New Journal of Physics}\ }\textbf
  {\bibinfo {volume} {20}},\ \bibinfo {eid} {053006} (\bibinfo {year}
  {2018})},\ \Eprint {https://arxiv.org/abs/1707.08456} {arXiv:1707.08456
  [quant-ph]} \BibitemShut {NoStop}%
\bibitem [{\citenamefont {Wang}\ and\ \citenamefont
  {Braunstein}(2016)}]{wang_higher-dimensional_2016}%
  \BibitemOpen
  \bibfield  {author} {\bibinfo {author} {\bibfnamefont {Z.-W.}\ \bibnamefont
  {Wang}}\ and\ \bibinfo {author} {\bibfnamefont {S.~L.}\ \bibnamefont
  {Braunstein}},\ }\bibfield  {title} {{\selectlanguage {en}\bibinfo {title}
  {Higher-dimensional performance of port-based teleportation}},\ }\href
  {https://doi.org/10.1038/srep33004} {\bibfield  {journal} {\bibinfo
  {journal} {Scientific Reports}\ }\textbf {\bibinfo {volume} {6}},\ \bibinfo
  {pages} {33004} (\bibinfo {year} {2016})}\BibitemShut {NoStop}%
\bibitem [{\citenamefont {Christandl}\ \emph {et~al.}(2020)\citenamefont
  {Christandl}, \citenamefont {Leditzky}, \citenamefont {Majenz}, \citenamefont
  {Smith}, \citenamefont {Speelman},\ and\ \citenamefont {Walter}}]{majenz2}%
  \BibitemOpen
  \bibfield  {author} {\bibinfo {author} {\bibfnamefont {M.}~\bibnamefont
  {Christandl}}, \bibinfo {author} {\bibfnamefont {F.}~\bibnamefont
  {Leditzky}}, \bibinfo {author} {\bibfnamefont {C.}~\bibnamefont {Majenz}},
  \bibinfo {author} {\bibfnamefont {G.}~\bibnamefont {Smith}}, \bibinfo
  {author} {\bibfnamefont {F.}~\bibnamefont {Speelman}},\ and\ \bibinfo
  {author} {\bibfnamefont {M.}~\bibnamefont {Walter}},\ }\bibfield  {title}
  {\bibinfo {title} {Asymptotic performance of port-based teleportation},\
  }\bibfield  {journal} {\bibinfo  {journal} {Communications in Mathematical
  Physics}\ }\href {https://doi.org/10.1007/s00220-020-03884-0}
  {10.1007/s00220-020-03884-0} (\bibinfo {year} {2020})\BibitemShut {NoStop}%
\bibitem [{\citenamefont {Leditzky}(2022)}]{leditzky2020optimality}%
  \BibitemOpen
  \bibfield  {author} {\bibinfo {author} {\bibfnamefont {F.}~\bibnamefont
  {Leditzky}},\ }\bibfield  {title} {\bibinfo {title} {Optimality of the pretty
  good measurement for port-based teleportation},\ }\href
  {https://doi.org/10.1007/s11005-022-01592-5} {\bibfield  {journal} {\bibinfo
  {journal} {Letters in Mathematical Physics}\ }\textbf {\bibinfo {volume}
  {112}} (\bibinfo {year} {2022})}\BibitemShut {NoStop}%
\bibitem [{\citenamefont {Taranto}\ \emph {et~al.}(2025)\citenamefont
  {Taranto}, \citenamefont {Milz}, \citenamefont {Murao}, \citenamefont
  {Quintino},\ and\ \citenamefont
  {Modi}}]{taranto2025higherorderquantumoperations}%
  \BibitemOpen
  \bibfield  {author} {\bibinfo {author} {\bibfnamefont {P.}~\bibnamefont
  {Taranto}}, \bibinfo {author} {\bibfnamefont {S.}~\bibnamefont {Milz}},
  \bibinfo {author} {\bibfnamefont {M.}~\bibnamefont {Murao}}, \bibinfo
  {author} {\bibfnamefont {M.~T.}\ \bibnamefont {Quintino}},\ and\ \bibinfo
  {author} {\bibfnamefont {K.}~\bibnamefont {Modi}},\ }\href
  {https://arxiv.org/abs/2503.09693} {\bibinfo {title} {Higher-order quantum
  operations}} (\bibinfo {year} {2025}),\ \Eprint
  {https://arxiv.org/abs/2503.09693} {arXiv:2503.09693 [quant-ph]} \BibitemShut
  {NoStop}%
\bibitem [{\citenamefont {Pereira}\ \emph {et~al.}(2021)\citenamefont
  {Pereira}, \citenamefont {Banchi},\ and\ \citenamefont {Pirandola}}]{sim}%
  \BibitemOpen
  \bibfield  {author} {\bibinfo {author} {\bibfnamefont {J.}~\bibnamefont
  {Pereira}}, \bibinfo {author} {\bibfnamefont {L.}~\bibnamefont {Banchi}},\
  and\ \bibinfo {author} {\bibfnamefont {S.}~\bibnamefont {Pirandola}},\
  }\bibfield  {title} {\bibinfo {title} {Characterising port-based
  teleportation as universal simulator of qubit channels},\ }\href
  {https://doi.org/10.1088/1751-8121/abe67a} {\bibfield  {journal} {\bibinfo
  {journal} {Journal of Physics A: Mathematical and Theoretical}\ }\textbf
  {\bibinfo {volume} {54}},\ \bibinfo {pages} {205301} (\bibinfo {year}
  {2021})}\BibitemShut {NoStop}%
\bibitem [{\citenamefont {Muguruza}\ and\ \citenamefont
  {Speelman}(2024)}]{Muguruza2024portbasedstate}%
  \BibitemOpen
  \bibfield  {author} {\bibinfo {author} {\bibfnamefont {G.}~\bibnamefont
  {Muguruza}}\ and\ \bibinfo {author} {\bibfnamefont {F.}~\bibnamefont
  {Speelman}},\ }\bibfield  {title} {\bibinfo {title} {Port-{B}ased {S}tate
  {P}reparation and {A}pplications},\ }\href
  {https://doi.org/10.22331/q-2024-12-18-1573} {\bibfield  {journal} {\bibinfo
  {journal} {{Quantum}}\ }\textbf {\bibinfo {volume} {8}},\ \bibinfo {pages}
  {1573} (\bibinfo {year} {2024})}\BibitemShut {NoStop}%
\bibitem [{\citenamefont {Brzić}\ \emph {et~al.}(2025)\citenamefont {Brzić},
  \citenamefont {Yoshida}, \citenamefont {Murao},\ and\ \citenamefont
  {Quintino}}]{brzic2025}%
  \BibitemOpen
  \bibfield  {author} {\bibinfo {author} {\bibfnamefont {V.}~\bibnamefont
  {Brzić}}, \bibinfo {author} {\bibfnamefont {S.}~\bibnamefont {Yoshida}},
  \bibinfo {author} {\bibfnamefont {M.}~\bibnamefont {Murao}},\ and\ \bibinfo
  {author} {\bibfnamefont {M.~T.}\ \bibnamefont {Quintino}},\ }\href
  {https://arxiv.org/abs/2510.20530} {\bibinfo {title} {Higher-order quantum
  computing with known input states}} (\bibinfo {year} {2025}),\ \Eprint
  {https://arxiv.org/abs/2510.20530} {arXiv:2510.20530 [quant-ph]} \BibitemShut
  {NoStop}%
\bibitem [{\citenamefont {Yoshida}\ \emph {et~al.}(2024)\citenamefont
  {Yoshida}, \citenamefont {Koizumi}, \citenamefont {Studziński},
  \citenamefont {Quintino},\ and\ \citenamefont {Murao}}]{yoshida2024onetoone}%
  \BibitemOpen
  \bibfield  {author} {\bibinfo {author} {\bibfnamefont {S.}~\bibnamefont
  {Yoshida}}, \bibinfo {author} {\bibfnamefont {Y.}~\bibnamefont {Koizumi}},
  \bibinfo {author} {\bibfnamefont {M.}~\bibnamefont {Studziński}}, \bibinfo
  {author} {\bibfnamefont {M.~T.}\ \bibnamefont {Quintino}},\ and\ \bibinfo
  {author} {\bibfnamefont {M.}~\bibnamefont {Murao}},\ }\href
  {https://arxiv.org/abs/2408.11902} {\bibinfo {title} {One-to-one
  correspondence between deterministic port-based teleportation and unitary
  estimation}} (\bibinfo {year} {2024}),\ \Eprint
  {https://arxiv.org/abs/2408.11902} {arXiv:2408.11902 [quant-ph]} \BibitemShut
  {NoStop}%
\bibitem [{\citenamefont {Quintino}\ \emph
  {et~al.}(2019{\natexlab{a}})\citenamefont {Quintino}, \citenamefont {Dong},
  \citenamefont {Shimbo}, \citenamefont {Soeda},\ and\ \citenamefont
  {Murao}}]{PhysRevLett.123.210502}%
  \BibitemOpen
  \bibfield  {author} {\bibinfo {author} {\bibfnamefont {M.~T.}\ \bibnamefont
  {Quintino}}, \bibinfo {author} {\bibfnamefont {Q.}~\bibnamefont {Dong}},
  \bibinfo {author} {\bibfnamefont {A.}~\bibnamefont {Shimbo}}, \bibinfo
  {author} {\bibfnamefont {A.}~\bibnamefont {Soeda}},\ and\ \bibinfo {author}
  {\bibfnamefont {M.}~\bibnamefont {Murao}},\ }\bibfield  {title} {\bibinfo
  {title} {Reversing unknown quantum transformations: Universal quantum circuit
  for inverting general unitary operations},\ }\href
  {https://doi.org/10.1103/PhysRevLett.123.210502} {\bibfield  {journal}
  {\bibinfo  {journal} {Phys. Rev. Lett.}\ }\textbf {\bibinfo {volume} {123}},\
  \bibinfo {pages} {210502} (\bibinfo {year} {2019}{\natexlab{a}})}\BibitemShut
  {NoStop}%
\bibitem [{\citenamefont {Quintino}\ \emph
  {et~al.}(2019{\natexlab{b}})\citenamefont {Quintino}, \citenamefont {Dong},
  \citenamefont {Shimbo}, \citenamefont {Soeda},\ and\ \citenamefont
  {Murao}}]{PhysRevA.100.062339}%
  \BibitemOpen
  \bibfield  {author} {\bibinfo {author} {\bibfnamefont {M.~T.}\ \bibnamefont
  {Quintino}}, \bibinfo {author} {\bibfnamefont {Q.}~\bibnamefont {Dong}},
  \bibinfo {author} {\bibfnamefont {A.}~\bibnamefont {Shimbo}}, \bibinfo
  {author} {\bibfnamefont {A.}~\bibnamefont {Soeda}},\ and\ \bibinfo {author}
  {\bibfnamefont {M.}~\bibnamefont {Murao}},\ }\bibfield  {title} {\bibinfo
  {title} {Probabilistic exact universal quantum circuits for transforming
  unitary operations},\ }\href {https://doi.org/10.1103/PhysRevA.100.062339}
  {\bibfield  {journal} {\bibinfo  {journal} {Phys. Rev. A}\ }\textbf {\bibinfo
  {volume} {100}},\ \bibinfo {pages} {062339} (\bibinfo {year}
  {2019}{\natexlab{b}})}\BibitemShut {NoStop}%
\bibitem [{\citenamefont {Quintino}\ and\ \citenamefont
  {Ebler}(2022)}]{Quintino2022deterministic}%
  \BibitemOpen
  \bibfield  {author} {\bibinfo {author} {\bibfnamefont {M.~T.}\ \bibnamefont
  {Quintino}}\ and\ \bibinfo {author} {\bibfnamefont {D.}~\bibnamefont
  {Ebler}},\ }\bibfield  {title} {\bibinfo {title} {Deterministic
  transformations between unitary operations: {E}xponential advantage with
  adaptive quantum circuits and the power of indefinite causality},\ }\href
  {https://doi.org/10.22331/q-2022-03-31-679} {\bibfield  {journal} {\bibinfo
  {journal} {{Quantum}}\ }\textbf {\bibinfo {volume} {6}},\ \bibinfo {pages}
  {679} (\bibinfo {year} {2022})}\BibitemShut {NoStop}%
\bibitem [{\citenamefont {{Sedl{\'a}k}}\ \emph {et~al.}(2019)\citenamefont
  {{Sedl{\'a}k}}, \citenamefont {{Bisio}},\ and\ \citenamefont
  {{Ziman}}}]{Stroing}%
  \BibitemOpen
  \bibfield  {author} {\bibinfo {author} {\bibfnamefont {M.}~\bibnamefont
  {{Sedl{\'a}k}}}, \bibinfo {author} {\bibfnamefont {A.}~\bibnamefont
  {{Bisio}}},\ and\ \bibinfo {author} {\bibfnamefont {M.}~\bibnamefont
  {{Ziman}}},\ }\bibfield  {title} {\bibinfo {title} {{Optimal Probabilistic
  Storage and Retrieval of Unitary Channels}},\ }\href
  {https://doi.org/10.1103/PhysRevLett.122.170502} {\bibfield  {journal}
  {\bibinfo  {journal} {\prl}\ }\textbf {\bibinfo {volume} {122}},\ \bibinfo
  {eid} {170502} (\bibinfo {year} {2019})},\ \Eprint
  {https://arxiv.org/abs/1809.04552} {arXiv:1809.04552 [quant-ph]} \BibitemShut
  {NoStop}%
\bibitem [{\citenamefont {Bisio}\ \emph {et~al.}(2010)\citenamefont {Bisio},
  \citenamefont {Chiribella}, \citenamefont {D'Ariano}, \citenamefont
  {Facchini},\ and\ \citenamefont {Perinotti}}]{PhysRevA.81.032324}%
  \BibitemOpen
  \bibfield  {author} {\bibinfo {author} {\bibfnamefont {A.}~\bibnamefont
  {Bisio}}, \bibinfo {author} {\bibfnamefont {G.}~\bibnamefont {Chiribella}},
  \bibinfo {author} {\bibfnamefont {G.~M.}\ \bibnamefont {D'Ariano}}, \bibinfo
  {author} {\bibfnamefont {S.}~\bibnamefont {Facchini}},\ and\ \bibinfo
  {author} {\bibfnamefont {P.}~\bibnamefont {Perinotti}},\ }\bibfield  {title}
  {\bibinfo {title} {Optimal quantum learning of a unitary transformation},\
  }\href {https://doi.org/10.1103/PhysRevA.81.032324} {\bibfield  {journal}
  {\bibinfo  {journal} {Phys. Rev. A}\ }\textbf {\bibinfo {volume} {81}},\
  \bibinfo {pages} {032324} (\bibinfo {year} {2010})}\BibitemShut {NoStop}%
\bibitem [{\citenamefont {Beigi}\ and\ \citenamefont
  {K{\"o}nig}(2011)}]{beigi_konig}%
  \BibitemOpen
  \bibfield  {author} {\bibinfo {author} {\bibfnamefont {S.}~\bibnamefont
  {Beigi}}\ and\ \bibinfo {author} {\bibfnamefont {R.}~\bibnamefont
  {K{\"o}nig}},\ }\bibfield  {title} {\bibinfo {title} {Simplified
  instantaneous non-local quantum computation with applications to
  position-based cryptography},\ }\href
  {https://doi.org/10.1088/1367-2630/13/9/093036} {\bibfield  {journal}
  {\bibinfo  {journal} {New Journal of Physics}\ }\textbf {\bibinfo {volume}
  {13}},\ \bibinfo {pages} {093036} (\bibinfo {year} {2011})}\BibitemShut
  {NoStop}%
\bibitem [{\citenamefont {May}(2022)}]{may2022complexity}%
  \BibitemOpen
  \bibfield  {author} {\bibinfo {author} {\bibfnamefont {A.}~\bibnamefont
  {May}},\ }\bibfield  {title} {\bibinfo {title} {Complexity and entanglement
  in non-local computation and holography},\ }\href@noop {} {\bibfield
  {journal} {\bibinfo  {journal} {Quantum}\ }\textbf {\bibinfo {volume} {6}},\
  \bibinfo {pages} {864} (\bibinfo {year} {2022})}\BibitemShut {NoStop}%
\bibitem [{\citenamefont {Buhrman}\ \emph {et~al.}(2016)\citenamefont
  {Buhrman}, \citenamefont {Czekaj}, \citenamefont {Grudka}, \citenamefont
  {Horodecki}, \citenamefont {Horodecki}, \citenamefont {Markiewicz},
  \citenamefont {Speelman},\ and\ \citenamefont
  {Strelchuk}}]{buhrman_quantum_2016}%
  \BibitemOpen
  \bibfield  {author} {\bibinfo {author} {\bibfnamefont {H.}~\bibnamefont
  {Buhrman}}, \bibinfo {author} {\bibfnamefont {{\L}.}~\bibnamefont {Czekaj}},
  \bibinfo {author} {\bibfnamefont {A.}~\bibnamefont {Grudka}}, \bibinfo
  {author} {\bibfnamefont {M.}~\bibnamefont {Horodecki}}, \bibinfo {author}
  {\bibfnamefont {P.}~\bibnamefont {Horodecki}}, \bibinfo {author}
  {\bibfnamefont {M.}~\bibnamefont {Markiewicz}}, \bibinfo {author}
  {\bibfnamefont {F.}~\bibnamefont {Speelman}},\ and\ \bibinfo {author}
  {\bibfnamefont {S.}~\bibnamefont {Strelchuk}},\ }\bibfield  {title}
  {{\selectlanguage {en}\bibinfo {title} {Quantum communication complexity
  advantage implies violation of a {Bell} inequality}},\ }\href
  {https://doi.org/10.1073/pnas.1507647113} {\bibfield  {journal} {\bibinfo
  {journal} {Proceedings of the National Academy of Sciences}\ }\textbf
  {\bibinfo {volume} {113}},\ \bibinfo {pages} {3191} (\bibinfo {year}
  {2016})}\BibitemShut {NoStop}%
\bibitem [{\citenamefont {Pereira}\ \emph {et~al.}(2023)\citenamefont
  {Pereira}, \citenamefont {Banchi},\ and\ \citenamefont
  {Pirandola}}]{pereira2023continuous}%
  \BibitemOpen
  \bibfield  {author} {\bibinfo {author} {\bibfnamefont {J.~L.}\ \bibnamefont
  {Pereira}}, \bibinfo {author} {\bibfnamefont {L.}~\bibnamefont {Banchi}},\
  and\ \bibinfo {author} {\bibfnamefont {S.}~\bibnamefont {Pirandola}},\
  }\bibfield  {title} {\bibinfo {title} {Continuous variable port-based
  teleportation},\ }\href@noop {} {\bibfield  {journal} {\bibinfo  {journal}
  {arXiv:2302.08522}\ } (\bibinfo {year} {2023})}\BibitemShut {NoStop}%
\bibitem [{\citenamefont {{Pirandola}}\ \emph {et~al.}(2019)\citenamefont
  {{Pirandola}}, \citenamefont {{Laurenza}}, \citenamefont {{Lupo}},\ and\
  \citenamefont {{Pereira}}}]{limit}%
  \BibitemOpen
  \bibfield  {author} {\bibinfo {author} {\bibfnamefont {S.}~\bibnamefont
  {{Pirandola}}}, \bibinfo {author} {\bibfnamefont {R.}~\bibnamefont
  {{Laurenza}}}, \bibinfo {author} {\bibfnamefont {C.}~\bibnamefont {{Lupo}}},\
  and\ \bibinfo {author} {\bibfnamefont {J.~L.}\ \bibnamefont {{Pereira}}},\
  }\bibfield  {title} {\bibinfo {title} {{Fundamental limits to quantum channel
  discrimination}},\ }\href {https://doi.org/10.1038/s41534-019-0162-y}
  {\bibfield  {journal} {\bibinfo  {journal} {npj Quantum Information}\
  }\textbf {\bibinfo {volume} {5}},\ \bibinfo {eid} {50} (\bibinfo {year}
  {2019})},\ \Eprint {https://arxiv.org/abs/1803.02834} {arXiv:1803.02834
  [quant-ph]} \BibitemShut {NoStop}%
\bibitem [{\citenamefont {Chitambar}\ and\ \citenamefont
  {Leditzky}(2024)}]{10315956}%
  \BibitemOpen
  \bibfield  {author} {\bibinfo {author} {\bibfnamefont {E.}~\bibnamefont
  {Chitambar}}\ and\ \bibinfo {author} {\bibfnamefont {F.}~\bibnamefont
  {Leditzky}},\ }\bibfield  {title} {\bibinfo {title} {On the duality of
  teleportation and dense coding},\ }\href
  {https://doi.org/10.1109/TIT.2023.3331821} {\bibfield  {journal} {\bibinfo
  {journal} {IEEE Transactions on Information Theory}\ }\textbf {\bibinfo
  {volume} {70}},\ \bibinfo {pages} {3529} (\bibinfo {year}
  {2024})}\BibitemShut {NoStop}%
\bibitem [{\citenamefont {Grinko}\ \emph {et~al.}(2023)\citenamefont {Grinko},
  \citenamefont {Burchardt},\ and\ \citenamefont
  {Ozols}}]{grinko2023gelfandtsetlinbasispartiallytransposed}%
  \BibitemOpen
  \bibfield  {author} {\bibinfo {author} {\bibfnamefont {D.}~\bibnamefont
  {Grinko}}, \bibinfo {author} {\bibfnamefont {A.}~\bibnamefont {Burchardt}},\
  and\ \bibinfo {author} {\bibfnamefont {M.}~\bibnamefont {Ozols}},\ }\href
  {https://arxiv.org/abs/2310.02252} {\bibinfo {title} {Gelfand-tsetlin basis
  for partially transposed permutations, with applications to quantum
  information}} (\bibinfo {year} {2023}),\ \Eprint
  {https://arxiv.org/abs/2310.02252} {arXiv:2310.02252 [quant-ph]} \BibitemShut
  {NoStop}%
\bibitem [{\citenamefont {Fei}\ \emph {et~al.}(2023)\citenamefont {Fei},
  \citenamefont {Timmerman},\ and\ \citenamefont
  {Hayden}}]{fei2023efficientquantumalgorithmportbased}%
  \BibitemOpen
  \bibfield  {author} {\bibinfo {author} {\bibfnamefont {J.}~\bibnamefont
  {Fei}}, \bibinfo {author} {\bibfnamefont {S.}~\bibnamefont {Timmerman}},\
  and\ \bibinfo {author} {\bibfnamefont {P.}~\bibnamefont {Hayden}},\ }\href
  {https://arxiv.org/abs/2310.01637} {\bibinfo {title} {Efficient quantum
  algorithm for port-based teleportation}} (\bibinfo {year} {2023}),\ \Eprint
  {https://arxiv.org/abs/2310.01637} {arXiv:2310.01637 [quant-ph]} \BibitemShut
  {NoStop}%
\bibitem [{\citenamefont {Grinko}\ \emph {et~al.}(2024)\citenamefont {Grinko},
  \citenamefont {Burchardt},\ and\ \citenamefont {Ozols}}]{Grinko2024}%
  \BibitemOpen
  \bibfield  {author} {\bibinfo {author} {\bibfnamefont {D.}~\bibnamefont
  {Grinko}}, \bibinfo {author} {\bibfnamefont {A.}~\bibnamefont {Burchardt}},\
  and\ \bibinfo {author} {\bibfnamefont {M.}~\bibnamefont {Ozols}},\ }\href
  {https://arxiv.org/abs/2312.03188} {\bibinfo {title} {Efficient quantum
  circuits for port-based teleportation}} (\bibinfo {year} {2024}),\ \Eprint
  {https://arxiv.org/abs/2312.03188} {arXiv:2312.03188 [quant-ph]} \BibitemShut
  {NoStop}%
\bibitem [{\citenamefont {Wills}\ \emph {et~al.}(2024)\citenamefont {Wills},
  \citenamefont {Hsieh},\ and\ \citenamefont
  {Strelchuk}}]{PRXQuantum.5.030354}%
  \BibitemOpen
  \bibfield  {author} {\bibinfo {author} {\bibfnamefont {A.}~\bibnamefont
  {Wills}}, \bibinfo {author} {\bibfnamefont {M.-H.}\ \bibnamefont {Hsieh}},\
  and\ \bibinfo {author} {\bibfnamefont {S.}~\bibnamefont {Strelchuk}},\
  }\bibfield  {title} {\bibinfo {title} {Efficient algorithms for all
  port-based teleportation protocols},\ }\href
  {https://doi.org/10.1103/PRXQuantum.5.030354} {\bibfield  {journal} {\bibinfo
   {journal} {PRX Quantum}\ }\textbf {\bibinfo {volume} {5}},\ \bibinfo {pages}
  {030354} (\bibinfo {year} {2024})}\BibitemShut {NoStop}%
\bibitem [{\citenamefont {Kim}\ and\ \citenamefont {Jeong}(2024)}]{Kim}%
  \BibitemOpen
  \bibfield  {author} {\bibinfo {author} {\bibfnamefont {H.~E.}\ \bibnamefont
  {Kim}}\ and\ \bibinfo {author} {\bibfnamefont {K.}~\bibnamefont {Jeong}},\
  }\bibfield  {title} {\bibinfo {title} {Port-based entanglement teleportation
  via noisy resource states},\ }\href
  {https://doi.org/10.1088/1402-4896/ad22c6} {\bibfield  {journal} {\bibinfo
  {journal} {Physica Scripta}\ }\textbf {\bibinfo {volume} {99}},\ \bibinfo
  {pages} {035105} (\bibinfo {year} {2024})}\BibitemShut {NoStop}%
\bibitem [{\citenamefont
  {Schlosshauer}(2007{\natexlab{a}})}]{schlosshauer2007}%
  \BibitemOpen
  \bibfield  {author} {\bibinfo {author} {\bibfnamefont {M.}~\bibnamefont
  {Schlosshauer}},\ }\href {https://doi.org/10.1007/978-3-540-35775-9} {\emph
  {\bibinfo {title} {Decoherence and the Quantum-To-Classical Transition}}}\
  (\bibinfo  {publisher} {Springer Science \& Business Media},\ \bibinfo {year}
  {2007})\BibitemShut {NoStop}%
\bibitem [{Note1()}]{Note1}%
  \BibitemOpen
  \bibinfo {note} {In the literature, square-root measurements (SRMs) are
  sometimes referred to as pretty-good measurements (PGMs). We will use the two
  terms interchangeably.}\BibitemShut {Stop}%
\bibitem [{\citenamefont {Palma}\ \emph {et~al.}(1996)\citenamefont {Palma},
  \citenamefont {Suominen},\ and\ \citenamefont
  {Ekert}}]{10.1098/rspa.1996.0029}%
  \BibitemOpen
  \bibfield  {author} {\bibinfo {author} {\bibfnamefont {G.~M.}\ \bibnamefont
  {Palma}}, \bibinfo {author} {\bibfnamefont {K.-a.}\ \bibnamefont
  {Suominen}},\ and\ \bibinfo {author} {\bibfnamefont {A.}~\bibnamefont
  {Ekert}},\ }\bibfield  {title} {\bibinfo {title} {Quantum computers and
  dissipation},\ }\href {https://doi.org/10.1098/rspa.1996.0029} {\bibfield
  {journal} {\bibinfo  {journal} {Proceedings of the Royal Society A:
  Mathematical, Physical and Engineering Sciences}\ }\textbf {\bibinfo {volume}
  {452}},\ \bibinfo {pages} {567} (\bibinfo {year} {1996})},\ \Eprint
  {https://arxiv.org/abs/https://royalsocietypublishing.org/rspa/article-pdf/452/1946/567/998471/rspa.1996.0029.pdf}
  {https://royalsocietypublishing.org/rspa/article-pdf/452/1946/567/998471/rspa.1996.0029.pdf}
  \BibitemShut {NoStop}%
\bibitem [{\citenamefont {Reina}\ \emph {et~al.}(2002)\citenamefont {Reina},
  \citenamefont {Quiroga},\ and\ \citenamefont {Johnson}}]{PhysRevA.65.032326}%
  \BibitemOpen
  \bibfield  {author} {\bibinfo {author} {\bibfnamefont {J.~H.}\ \bibnamefont
  {Reina}}, \bibinfo {author} {\bibfnamefont {L.}~\bibnamefont {Quiroga}},\
  and\ \bibinfo {author} {\bibfnamefont {N.~F.}\ \bibnamefont {Johnson}},\
  }\bibfield  {title} {\bibinfo {title} {Decoherence of quantum registers},\
  }\href {https://doi.org/10.1103/PhysRevA.65.032326} {\bibfield  {journal}
  {\bibinfo  {journal} {Phys. Rev. A}\ }\textbf {\bibinfo {volume} {65}},\
  \bibinfo {pages} {032326} (\bibinfo {year} {2002})}\BibitemShut {NoStop}%
\bibitem [{\citenamefont
  {Schlosshauer}(2007{\natexlab{b}})}]{Schlosshauer2007-tx}%
  \BibitemOpen
  \bibfield  {author} {\bibinfo {author} {\bibfnamefont {M.}~\bibnamefont
  {Schlosshauer}},\ }\href@noop {} {\emph {\bibinfo {title} {Decoherence: And
  the {Quantum-To-Classical} Transition}}}\ (\bibinfo  {publisher} {Springer},\
  \bibinfo {year} {2007})\BibitemShut {NoStop}%
\bibitem [{\citenamefont {Ruggiero}\ \emph {et~al.}(2006)\citenamefont
  {Ruggiero}, \citenamefont {Delsing}, \citenamefont {Granata}, \citenamefont
  {Pashkin},\ and\ \citenamefont {Silvestrini}}]{ruggiero2006quantum}%
  \BibitemOpen
  \bibfield  {author} {\bibinfo {author} {\bibfnamefont {B.}~\bibnamefont
  {Ruggiero}}, \bibinfo {author} {\bibfnamefont {P.}~\bibnamefont {Delsing}},
  \bibinfo {author} {\bibfnamefont {C.}~\bibnamefont {Granata}}, \bibinfo
  {author} {\bibfnamefont {Y.~A.}\ \bibnamefont {Pashkin}},\ and\ \bibinfo
  {author} {\bibfnamefont {P.}~\bibnamefont {Silvestrini}},\ }\bibfield
  {title} {\bibinfo {title} {Quantum computing in solid state systems}\
  }(\bibinfo  {publisher} {Springer},\ \bibinfo {address} {New York, NY},\
  \bibinfo {year} {2006})\BibitemShut {NoStop}%
\bibitem [{\citenamefont {Predojevi{\'c}}\ and\ \citenamefont
  {Mitchell}(2015)}]{predojevic2015engineering}%
  \BibitemOpen
  \bibfield  {author} {\bibinfo {author} {\bibfnamefont {A.}~\bibnamefont
  {Predojevi{\'c}}}\ and\ \bibinfo {author} {\bibfnamefont {M.}~\bibnamefont
  {Mitchell}},\ }\href {https://books.google.pl/books?id=Qj8wCgAAQBAJ} {\emph
  {\bibinfo {title} {Engineering the Atom-Photon Interaction: Controlling
  Fundamental Processes with Photons, Atoms and Solids}}},\ Nano-Optics and
  Nanophotonics\ (\bibinfo  {publisher} {Springer International Publishing},\
  \bibinfo {year} {2015})\BibitemShut {NoStop}%
\bibitem [{\citenamefont {Gaitan}(2008)}]{gaitan2008quantum}%
  \BibitemOpen
  \bibfield  {author} {\bibinfo {author} {\bibfnamefont {F.}~\bibnamefont
  {Gaitan}},\ }\href@noop {} {\emph {\bibinfo {title} {Quantum Error Correction
  and Fault Tolerant Quantum Computing}}}\ (\bibinfo  {publisher} {CRC Press},\
  \bibinfo {address} {Boca Raton, FL},\ \bibinfo {year} {2008})\BibitemShut
  {NoStop}%
\bibitem [{\citenamefont {Weiss}(2021)}]{weiss2021quantum}%
  \BibitemOpen
  \bibfield  {author} {\bibinfo {author} {\bibfnamefont {U.}~\bibnamefont
  {Weiss}},\ }\href@noop {} {\emph {\bibinfo {title} {Quantum Dissipative
  Systems}}},\ \bibinfo {edition} {5th}\ ed.\ (\bibinfo  {publisher} {World
  Scientific Publishing Company},\ \bibinfo {address} {Singapore},\ \bibinfo
  {year} {2021})\BibitemShut {NoStop}%
\bibitem [{\citenamefont {Horodecki}\ \emph {et~al.}(1999)\citenamefont
  {Horodecki}, \citenamefont {Horodecki},\ and\ \citenamefont
  {Horodecki}}]{HorodeckiMPRFidelity}%
  \BibitemOpen
  \bibfield  {author} {\bibinfo {author} {\bibfnamefont {M.}~\bibnamefont
  {Horodecki}}, \bibinfo {author} {\bibfnamefont {P.}~\bibnamefont
  {Horodecki}},\ and\ \bibinfo {author} {\bibfnamefont {R.}~\bibnamefont
  {Horodecki}},\ }\bibfield  {title} {\bibinfo {title} {General teleportation
  channel, singlet fraction, and quasidistillation},\ }\href
  {https://doi.org/10.1103/PhysRevA.60.1888} {\bibfield  {journal} {\bibinfo
  {journal} {Phys. Rev. A}\ }\textbf {\bibinfo {volume} {60}},\ \bibinfo
  {pages} {1888} (\bibinfo {year} {1999})}\BibitemShut {NoStop}%
\bibitem [{\citenamefont {Yu}\ and\ \citenamefont
  {Eberly}(2004)}]{PhysRevLett.93.140404}%
  \BibitemOpen
  \bibfield  {author} {\bibinfo {author} {\bibfnamefont {T.}~\bibnamefont
  {Yu}}\ and\ \bibinfo {author} {\bibfnamefont {J.~H.}\ \bibnamefont
  {Eberly}},\ }\bibfield  {title} {\bibinfo {title} {Finite-time
  disentanglement via spontaneous emission},\ }\href
  {https://doi.org/10.1103/PhysRevLett.93.140404} {\bibfield  {journal}
  {\bibinfo  {journal} {Phys. Rev. Lett.}\ }\textbf {\bibinfo {volume} {93}},\
  \bibinfo {pages} {140404} (\bibinfo {year} {2004})}\BibitemShut {NoStop}%
\bibitem [{\citenamefont {Dür}\ \emph {et~al.}(2004)\citenamefont {Dür},
  \citenamefont {Simon},\ and\ \citenamefont {Cirac}}]{dur2004stability}%
  \BibitemOpen
  \bibfield  {author} {\bibinfo {author} {\bibfnamefont {W.}~\bibnamefont
  {Dür}}, \bibinfo {author} {\bibfnamefont {C.}~\bibnamefont {Simon}},\ and\
  \bibinfo {author} {\bibfnamefont {J.~I.}\ \bibnamefont {Cirac}},\ }\bibfield
  {title} {\bibinfo {title} {Stability of macroscopic entanglement under
  decoherence},\ }\href {https://doi.org/10.1103/PhysRevLett.92.180403}
  {\bibfield  {journal} {\bibinfo  {journal} {Physical Review Letters}\
  }\textbf {\bibinfo {volume} {92}},\ \bibinfo {pages} {180403} (\bibinfo
  {year} {2004})}\BibitemShut {NoStop}%
\bibitem [{\citenamefont {Zurek}(2003)}]{zurek2003decoherence}%
  \BibitemOpen
  \bibfield  {author} {\bibinfo {author} {\bibfnamefont {W.~H.}\ \bibnamefont
  {Zurek}},\ }\bibfield  {title} {\bibinfo {title} {Decoherence, einselection,
  and the quantum origins of the classical},\ }\href
  {https://doi.org/10.1103/RevModPhys.75.715} {\bibfield  {journal} {\bibinfo
  {journal} {Reviews of Modern Physics}\ }\textbf {\bibinfo {volume} {75}},\
  \bibinfo {pages} {715} (\bibinfo {year} {2003})}\BibitemShut {NoStop}%
\bibitem [{\citenamefont {Schlosshauer}(2019)}]{SCHLOSSHAUER20191}%
  \BibitemOpen
  \bibfield  {author} {\bibinfo {author} {\bibfnamefont {M.}~\bibnamefont
  {Schlosshauer}},\ }\bibfield  {title} {\bibinfo {title} {Quantum
  decoherence},\ }\href
  {https://doi.org/https://doi.org/10.1016/j.physrep.2019.10.001} {\bibfield
  {journal} {\bibinfo  {journal} {Physics Reports}\ }\textbf {\bibinfo {volume}
  {831}},\ \bibinfo {pages} {1} (\bibinfo {year} {2019})},\ \bibinfo {note}
  {quantum decoherence}\BibitemShut {NoStop}%
\bibitem [{\citenamefont {Roszak}\ and\ \citenamefont
  {Korbicz}(2023)}]{Roszak2023purifying}%
  \BibitemOpen
  \bibfield  {author} {\bibinfo {author} {\bibfnamefont {K.}~\bibnamefont
  {Roszak}}\ and\ \bibinfo {author} {\bibfnamefont {J.~K.}\ \bibnamefont
  {Korbicz}},\ }\bibfield  {title} {\bibinfo {title} {Purifying
  teleportation},\ }\href {https://doi.org/10.22331/q-2023-02-16-923}
  {\bibfield  {journal} {\bibinfo  {journal} {{Quantum}}\ }\textbf {\bibinfo
  {volume} {7}},\ \bibinfo {pages} {923} (\bibinfo {year} {2023})}\BibitemShut
  {NoStop}%
\bibitem [{\citenamefont {Tuziemski}\ \emph {et~al.}(2019)\citenamefont
  {Tuziemski}, \citenamefont {Lampo}, \citenamefont {Lewenstein},\ and\
  \citenamefont {Korbicz}}]{Tuziemski2018}%
  \BibitemOpen
  \bibfield  {author} {\bibinfo {author} {\bibfnamefont {J.}~\bibnamefont
  {Tuziemski}}, \bibinfo {author} {\bibfnamefont {A.}~\bibnamefont {Lampo}},
  \bibinfo {author} {\bibfnamefont {M.}~\bibnamefont {Lewenstein}},\ and\
  \bibinfo {author} {\bibfnamefont {J.~K.}\ \bibnamefont {Korbicz}},\
  }\bibfield  {title} {\bibinfo {title} {Decoherence of spin registers
  revisited},\ }\href {https://doi.org/10.1103/PhysRevA.99.022122} {\bibfield
  {journal} {\bibinfo  {journal} {Physical Review A}\ }\textbf {\bibinfo
  {volume} {99}},\ \bibinfo {pages} {022122} (\bibinfo {year} {2019})},\
  \bibinfo {note} {arXiv:1811.01808},\ \Eprint
  {https://arxiv.org/abs/1811.01808} {arXiv:1811.01808 [quant-ph]} \BibitemShut
  {NoStop}%
\bibitem [{\citenamefont {Barnum}\ and\ \citenamefont
  {Knill}(2002)}]{barnum2002reversing}%
  \BibitemOpen
  \bibfield  {author} {\bibinfo {author} {\bibfnamefont {H.}~\bibnamefont
  {Barnum}}\ and\ \bibinfo {author} {\bibfnamefont {E.}~\bibnamefont {Knill}},\
  }\bibfield  {title} {\bibinfo {title} {Reversing quantum dynamics with
  near-optimal quantum and classical fidelity},\ }\href
  {https://doi.org/10.1063/1.1459754} {\bibfield  {journal} {\bibinfo
  {journal} {Journal of Mathematical Physics}\ }\textbf {\bibinfo {volume}
  {43}},\ \bibinfo {pages} {2097} (\bibinfo {year} {2002})}\BibitemShut
  {NoStop}%
\bibitem [{\citenamefont {Montanaro}(2008)}]{Montanaro2019}%
  \BibitemOpen
  \bibfield  {author} {\bibinfo {author} {\bibfnamefont {A.}~\bibnamefont
  {Montanaro}},\ }\bibfield  {title} {\bibinfo {title} {A lower bound on the
  probability of error in quantum state discrimination},\ }in\ \href
  {https://doi.org/10.1109/ITW.2008.4578690} {\emph {\bibinfo {booktitle} {2008
  IEEE Information Theory Workshop}}}\ (\bibinfo {year} {2008})\ pp.\ \bibinfo
  {pages} {378--380}\BibitemShut {NoStop}%
\bibitem [{\citenamefont {Audenaert}\ and\ \citenamefont
  {Mosonyi}(2014)}]{Audenaert2014}%
  \BibitemOpen
  \bibfield  {author} {\bibinfo {author} {\bibfnamefont {K.~M.~R.}\
  \bibnamefont {Audenaert}}\ and\ \bibinfo {author} {\bibfnamefont
  {M.}~\bibnamefont {Mosonyi}},\ }\bibfield  {title} {\bibinfo {title} {Upper
  bounds on the error probabilities and asymptotic error exponents in quantum
  multiple state discrimination},\ }\href {https://doi.org/10.1063/1.4898559}
  {\bibfield  {journal} {\bibinfo  {journal} {Journal of Mathematical Physics}\
  }\textbf {\bibinfo {volume} {55}},\ \bibinfo {pages} {102201} (\bibinfo
  {year} {2014})},\ \Eprint
  {https://arxiv.org/abs/https://pubs.aip.org/aip/jmp/article-pdf/doi/10.1063/1.4898559/14758713/102201\_1\_online.pdf}
  {https://pubs.aip.org/aip/jmp/article-pdf/doi/10.1063/1.4898559/14758713/102201\_1\_online.pdf}
  \BibitemShut {NoStop}%
\bibitem [{\citenamefont {Harris}\ \emph {et~al.}(2020)\citenamefont {Harris},
  \citenamefont {Millman}, \citenamefont {van~der Walt}, \citenamefont
  {Gommers}, \citenamefont {Virtanen}, \citenamefont {Cournapeau},
  \citenamefont {Wieser}, \citenamefont {Taylor}, \citenamefont {Berg},
  \citenamefont {Smith} \emph {et~al.}}]{harris2020numpy}%
  \BibitemOpen
  \bibfield  {author} {\bibinfo {author} {\bibfnamefont {C.~R.}\ \bibnamefont
  {Harris}}, \bibinfo {author} {\bibfnamefont {K.~J.}\ \bibnamefont {Millman}},
  \bibinfo {author} {\bibfnamefont {S.~J.}\ \bibnamefont {van~der Walt}},
  \bibinfo {author} {\bibfnamefont {R.}~\bibnamefont {Gommers}}, \bibinfo
  {author} {\bibfnamefont {P.}~\bibnamefont {Virtanen}}, \bibinfo {author}
  {\bibfnamefont {D.}~\bibnamefont {Cournapeau}}, \bibinfo {author}
  {\bibfnamefont {E.}~\bibnamefont {Wieser}}, \bibinfo {author} {\bibfnamefont
  {J.}~\bibnamefont {Taylor}}, \bibinfo {author} {\bibfnamefont
  {S.}~\bibnamefont {Berg}}, \bibinfo {author} {\bibfnamefont {N.~J.}\
  \bibnamefont {Smith}}, \emph {et~al.},\ }\bibfield  {title} {\bibinfo {title}
  {Array programming with numpy},\ }\href
  {https://doi.org/10.1038/s41586-020-2649-2} {\bibfield  {journal} {\bibinfo
  {journal} {Nature}\ }\textbf {\bibinfo {volume} {585}},\ \bibinfo {pages}
  {357} (\bibinfo {year} {2020})}\BibitemShut {NoStop}%
\bibitem [{\citenamefont {Eldar}\ \emph {et~al.}(2004)\citenamefont {Eldar},
  \citenamefont {Megretski},\ and\ \citenamefont
  {Verghese}}]{eldar2004optimal}%
  \BibitemOpen
  \bibfield  {author} {\bibinfo {author} {\bibfnamefont {Y.~C.}\ \bibnamefont
  {Eldar}}, \bibinfo {author} {\bibfnamefont {A.}~\bibnamefont {Megretski}},\
  and\ \bibinfo {author} {\bibfnamefont {G.~C.}\ \bibnamefont {Verghese}},\
  }\bibfield  {title} {\bibinfo {title} {Optimal detection of symmetric mixed
  quantum states},\ }\href {https://doi.org/10.1109/TIT.2004.828070} {\bibfield
   {journal} {\bibinfo  {journal} {IEEE Transactions on Information Theory}\
  }\textbf {\bibinfo {volume} {50}},\ \bibinfo {pages} {1198} (\bibinfo {year}
  {2004})}\BibitemShut {NoStop}%
\bibitem [{\citenamefont {Assalini}\ \emph {et~al.}(2010)\citenamefont
  {Assalini}, \citenamefont {Cariolaro},\ and\ \citenamefont
  {Pierobon}}]{PhysRevA.81.012315}%
  \BibitemOpen
  \bibfield  {author} {\bibinfo {author} {\bibfnamefont {A.}~\bibnamefont
  {Assalini}}, \bibinfo {author} {\bibfnamefont {G.}~\bibnamefont
  {Cariolaro}},\ and\ \bibinfo {author} {\bibfnamefont {G.}~\bibnamefont
  {Pierobon}},\ }\bibfield  {title} {\bibinfo {title} {Efficient optimal
  minimum error discrimination of symmetric quantum states},\ }\href
  {https://doi.org/10.1103/PhysRevA.81.012315} {\bibfield  {journal} {\bibinfo
  {journal} {Phys. Rev. A}\ }\textbf {\bibinfo {volume} {81}},\ \bibinfo
  {pages} {012315} (\bibinfo {year} {2010})}\BibitemShut {NoStop}%
\bibitem [{\citenamefont {Iten}\ \emph {et~al.}(2017)\citenamefont {Iten},
  \citenamefont {Renes},\ and\ \citenamefont {Sutter}}]{ItenRenesSutter}%
  \BibitemOpen
  \bibfield  {author} {\bibinfo {author} {\bibfnamefont {R.}~\bibnamefont
  {Iten}}, \bibinfo {author} {\bibfnamefont {J.~M.}\ \bibnamefont {Renes}},\
  and\ \bibinfo {author} {\bibfnamefont {D.}~\bibnamefont {Sutter}},\
  }\bibfield  {title} {\bibinfo {title} {Pretty good measures in quantum
  information theory},\ }\href {https://doi.org/10.1109/TIT.2016.2639521}
  {\bibfield  {journal} {\bibinfo  {journal} {IEEE Transactions on Information
  Theory}\ }\textbf {\bibinfo {volume} {63}},\ \bibinfo {pages} {1270}
  (\bibinfo {year} {2017})}\BibitemShut {NoStop}%
\bibitem [{\citenamefont {Helstrom}(1976)}]{helstrom1976quantum}%
  \BibitemOpen
  \bibfield  {author} {\bibinfo {author} {\bibfnamefont {C.~W.}\ \bibnamefont
  {Helstrom}},\ }\href@noop {} {\emph {\bibinfo {title} {Quantum Detection and
  Estimation Theory}}}\ (\bibinfo  {publisher} {Academic Press},\ \bibinfo
  {address} {New York, NY, USA},\ \bibinfo {year} {1976})\BibitemShut {NoStop}%
\bibitem [{\citenamefont {Horodecki}\ \emph {et~al.}(2003)\citenamefont
  {Horodecki}, \citenamefont {Shor},\ and\ \citenamefont
  {Ruskai}}]{ent_breaking_channels}%
  \BibitemOpen
  \bibfield  {author} {\bibinfo {author} {\bibfnamefont {M.}~\bibnamefont
  {Horodecki}}, \bibinfo {author} {\bibfnamefont {P.~W.}\ \bibnamefont
  {Shor}},\ and\ \bibinfo {author} {\bibfnamefont {M.~B.}\ \bibnamefont
  {Ruskai}},\ }\bibfield  {title} {\bibinfo {title} {Entanglement breaking
  channels},\ }\href {https://doi.org/10.1142/S0129055X03001709} {\bibfield
  {journal} {\bibinfo  {journal} {Reviews in Mathematical Physics}\ }\textbf
  {\bibinfo {volume} {15}},\ \bibinfo {pages} {629} (\bibinfo {year} {2003})},\
  \Eprint {https://arxiv.org/abs/https://doi.org/10.1142/S0129055X03001709}
  {https://doi.org/10.1142/S0129055X03001709} \BibitemShut {NoStop}%
\bibitem [{\citenamefont {Breuer}\ \emph {et~al.}(2009)\citenamefont {Breuer},
  \citenamefont {Laine},\ and\ \citenamefont {Piilo}}]{PhysRevLett.103.210401}%
  \BibitemOpen
  \bibfield  {author} {\bibinfo {author} {\bibfnamefont {H.-P.}\ \bibnamefont
  {Breuer}}, \bibinfo {author} {\bibfnamefont {E.-M.}\ \bibnamefont {Laine}},\
  and\ \bibinfo {author} {\bibfnamefont {J.}~\bibnamefont {Piilo}},\ }\bibfield
   {title} {\bibinfo {title} {Measure for the degree of non-markovian behavior
  of quantum processes in open systems},\ }\href
  {https://doi.org/10.1103/PhysRevLett.103.210401} {\bibfield  {journal}
  {\bibinfo  {journal} {Phys. Rev. Lett.}\ }\textbf {\bibinfo {volume} {103}},\
  \bibinfo {pages} {210401} (\bibinfo {year} {2009})}\BibitemShut {NoStop}%
\bibitem [{\citenamefont {Gily\'en}\ \emph {et~al.}(2022)\citenamefont
  {Gily\'en}, \citenamefont {Lloyd}, \citenamefont {Marvian}, \citenamefont
  {Quek},\ and\ \citenamefont {Wilde}}]{PhysRevLett.128.220502}%
  \BibitemOpen
  \bibfield  {author} {\bibinfo {author} {\bibfnamefont {A.}~\bibnamefont
  {Gily\'en}}, \bibinfo {author} {\bibfnamefont {S.}~\bibnamefont {Lloyd}},
  \bibinfo {author} {\bibfnamefont {I.}~\bibnamefont {Marvian}}, \bibinfo
  {author} {\bibfnamefont {Y.}~\bibnamefont {Quek}},\ and\ \bibinfo {author}
  {\bibfnamefont {M.~M.}\ \bibnamefont {Wilde}},\ }\bibfield  {title} {\bibinfo
  {title} {Quantum algorithm for petz recovery channels and pretty good
  measurements},\ }\href {https://doi.org/10.1103/PhysRevLett.128.220502}
  {\bibfield  {journal} {\bibinfo  {journal} {Phys. Rev. Lett.}\ }\textbf
  {\bibinfo {volume} {128}},\ \bibinfo {pages} {220502} (\bibinfo {year}
  {2022})}\BibitemShut {NoStop}%
\bibitem [{\citenamefont {Gily{\'e}n}\ \emph {et~al.}(2019)\citenamefont
  {Gily{\'e}n}, \citenamefont {Su}, \citenamefont {Low},\ and\ \citenamefont
  {Wiebe}}]{Gilyen2019QSVT}%
  \BibitemOpen
  \bibfield  {author} {\bibinfo {author} {\bibfnamefont {A.}~\bibnamefont
  {Gily{\'e}n}}, \bibinfo {author} {\bibfnamefont {Y.}~\bibnamefont {Su}},
  \bibinfo {author} {\bibfnamefont {G.~H.}\ \bibnamefont {Low}},\ and\ \bibinfo
  {author} {\bibfnamefont {N.}~\bibnamefont {Wiebe}},\ }\bibfield  {title}
  {\bibinfo {title} {Quantum singular value transformation and beyond},\
  }\href@noop {} {\bibfield  {journal} {\bibinfo  {journal} {Proceedings of the
  51st Annual ACM SIGACT Symposium on Theory of Computing}\ ,\ \bibinfo {pages}
  {193}} (\bibinfo {year} {2019})},\ \Eprint {https://arxiv.org/abs/1806.01838}
  {arXiv:1806.01838 [quant-ph]} \BibitemShut {NoStop}%
\bibitem [{\citenamefont {Low}\ and\ \citenamefont
  {Chuang}(2019)}]{Low2019Hamiltonian}%
  \BibitemOpen
  \bibfield  {author} {\bibinfo {author} {\bibfnamefont {G.~H.}\ \bibnamefont
  {Low}}\ and\ \bibinfo {author} {\bibfnamefont {I.~L.}\ \bibnamefont
  {Chuang}},\ }\bibfield  {title} {\bibinfo {title} {Hamiltonian simulation by
  qubitization},\ }\href@noop {} {\bibfield  {journal} {\bibinfo  {journal}
  {Quantum}\ }\textbf {\bibinfo {volume} {3}},\ \bibinfo {pages} {163}
  (\bibinfo {year} {2019})},\ \Eprint {https://arxiv.org/abs/1610.06546}
  {arXiv:1610.06546 [quant-ph]} \BibitemShut {NoStop}%
\bibitem [{\citenamefont {Bernien}\ \emph {et~al.}(2013)\citenamefont
  {Bernien}, \citenamefont {Hensen}, \citenamefont {Pfaff}, \citenamefont
  {Koolstra}, \citenamefont {Blok}, \citenamefont {Robledo}, \citenamefont
  {Taminiau}, \citenamefont {Markham}, \citenamefont {Twitchen}, \citenamefont
  {Childress},\ and\ \citenamefont {Hanson}}]{Bernien2013Heralded}%
  \BibitemOpen
  \bibfield  {author} {\bibinfo {author} {\bibfnamefont {H.}~\bibnamefont
  {Bernien}}, \bibinfo {author} {\bibfnamefont {B.}~\bibnamefont {Hensen}},
  \bibinfo {author} {\bibfnamefont {W.}~\bibnamefont {Pfaff}}, \bibinfo
  {author} {\bibfnamefont {G.}~\bibnamefont {Koolstra}}, \bibinfo {author}
  {\bibfnamefont {M.~S.}\ \bibnamefont {Blok}}, \bibinfo {author}
  {\bibfnamefont {L.}~\bibnamefont {Robledo}}, \bibinfo {author} {\bibfnamefont
  {T.~H.}\ \bibnamefont {Taminiau}}, \bibinfo {author} {\bibfnamefont
  {M.}~\bibnamefont {Markham}}, \bibinfo {author} {\bibfnamefont {D.~J.}\
  \bibnamefont {Twitchen}}, \bibinfo {author} {\bibfnamefont {L.}~\bibnamefont
  {Childress}},\ and\ \bibinfo {author} {\bibfnamefont {R.}~\bibnamefont
  {Hanson}},\ }\bibfield  {title} {\bibinfo {title} {Heralded entanglement
  between solid-state qubits separated by three metres},\ }\href
  {https://doi.org/10.1038/nature12016} {\bibfield  {journal} {\bibinfo
  {journal} {Nature}\ }\textbf {\bibinfo {volume} {497}},\ \bibinfo {pages}
  {86} (\bibinfo {year} {2013})}\BibitemShut {NoStop}%
\bibitem [{\citenamefont {Childress}\ and\ \citenamefont
  {Hanson}(2013)}]{ChildressHanson2013NV}%
  \BibitemOpen
  \bibfield  {author} {\bibinfo {author} {\bibfnamefont {L.}~\bibnamefont
  {Childress}}\ and\ \bibinfo {author} {\bibfnamefont {R.}~\bibnamefont
  {Hanson}},\ }\bibfield  {title} {{\selectlanguage {English}\bibinfo {title}
  {Diamond nv centers for quantum computing and quantum networks}},\ }\href
  {https://doi.org/10.1557/mrs.2013.20} {\bibfield  {journal} {\bibinfo
  {journal} {MRS Bulletin}\ }\textbf {\bibinfo {volume} {38}},\ \bibinfo
  {pages} {134} (\bibinfo {year} {2013})}\BibitemShut {NoStop}%
\bibitem [{\citenamefont {Bradley}\ \emph {et~al.}(2022)\citenamefont
  {Bradley}, \citenamefont {de~Bone}, \citenamefont {Möller}, \citenamefont
  {Baier}, \citenamefont {Degen}, \citenamefont {Loenen}, \citenamefont
  {Bartling}, \citenamefont {Markham}, \citenamefont {Twitchen}, \citenamefont
  {Hanson}, \citenamefont {Elkouss},\ and\ \citenamefont
  {Taminiau}}]{Bradley2022Robust}%
  \BibitemOpen
  \bibfield  {author} {\bibinfo {author} {\bibfnamefont {C.~E.}\ \bibnamefont
  {Bradley}}, \bibinfo {author} {\bibfnamefont {S.~W.}\ \bibnamefont
  {de~Bone}}, \bibinfo {author} {\bibfnamefont {P.~F.~W.}\ \bibnamefont
  {Möller}}, \bibinfo {author} {\bibfnamefont {S.}~\bibnamefont {Baier}},
  \bibinfo {author} {\bibfnamefont {M.~J.}\ \bibnamefont {Degen}}, \bibinfo
  {author} {\bibfnamefont {S.~J.~H.}\ \bibnamefont {Loenen}}, \bibinfo {author}
  {\bibfnamefont {H.~P.}\ \bibnamefont {Bartling}}, \bibinfo {author}
  {\bibfnamefont {M.}~\bibnamefont {Markham}}, \bibinfo {author} {\bibfnamefont
  {D.~J.}\ \bibnamefont {Twitchen}}, \bibinfo {author} {\bibfnamefont
  {R.}~\bibnamefont {Hanson}}, \bibinfo {author} {\bibfnamefont
  {D.}~\bibnamefont {Elkouss}},\ and\ \bibinfo {author} {\bibfnamefont {T.~H.}\
  \bibnamefont {Taminiau}},\ }\bibfield  {title} {\bibinfo {title} {Robust
  quantum-network memory based on spin qubits in isotopically engineered
  diamond},\ }\href {https://doi.org/10.1038/s41534-022-00637-w} {\bibfield
  {journal} {\bibinfo  {journal} {npj Quantum Information}\ }\textbf {\bibinfo
  {volume} {8}},\ \bibinfo {pages} {122} (\bibinfo {year} {2022})}\BibitemShut
  {NoStop}%
\bibitem [{\citenamefont {Siddiqi}(2021)}]{Siddiqi2021Engineering}%
  \BibitemOpen
  \bibfield  {author} {\bibinfo {author} {\bibfnamefont {I.}~\bibnamefont
  {Siddiqi}},\ }\bibfield  {title} {\bibinfo {title} {Engineering
  high-coherence superconducting qubits},\ }\href
  {https://doi.org/10.1038/s41578-021-00370-4} {\bibfield  {journal} {\bibinfo
  {journal} {Nature Reviews Materials}\ }\textbf {\bibinfo {volume} {6}},\
  \bibinfo {pages} {875} (\bibinfo {year} {2021})}\BibitemShut {NoStop}%
\bibitem [{\citenamefont {Anton}\ \emph {et~al.}(2012)\citenamefont {Anton},
  \citenamefont {M\"uller}, \citenamefont {Birenbaum}, \citenamefont
  {O'Kelley}, \citenamefont {Fefferman}, \citenamefont {Golubev}, \citenamefont
  {Hilton}, \citenamefont {Cho}, \citenamefont {Irwin}, \citenamefont
  {Wellstood}, \citenamefont {Sch\"on}, \citenamefont {Shnirman},\ and\
  \citenamefont {Clarke}}]{PhysRevB.85.224505}%
  \BibitemOpen
  \bibfield  {author} {\bibinfo {author} {\bibfnamefont {S.~M.}\ \bibnamefont
  {Anton}}, \bibinfo {author} {\bibfnamefont {C.}~\bibnamefont {M\"uller}},
  \bibinfo {author} {\bibfnamefont {J.~S.}\ \bibnamefont {Birenbaum}}, \bibinfo
  {author} {\bibfnamefont {S.~R.}\ \bibnamefont {O'Kelley}}, \bibinfo {author}
  {\bibfnamefont {A.~D.}\ \bibnamefont {Fefferman}}, \bibinfo {author}
  {\bibfnamefont {D.~S.}\ \bibnamefont {Golubev}}, \bibinfo {author}
  {\bibfnamefont {G.~C.}\ \bibnamefont {Hilton}}, \bibinfo {author}
  {\bibfnamefont {H.-M.}\ \bibnamefont {Cho}}, \bibinfo {author} {\bibfnamefont
  {K.~D.}\ \bibnamefont {Irwin}}, \bibinfo {author} {\bibfnamefont {F.~C.}\
  \bibnamefont {Wellstood}}, \bibinfo {author} {\bibfnamefont {G.}~\bibnamefont
  {Sch\"on}}, \bibinfo {author} {\bibfnamefont {A.}~\bibnamefont {Shnirman}},\
  and\ \bibinfo {author} {\bibfnamefont {J.}~\bibnamefont {Clarke}},\
  }\bibfield  {title} {\bibinfo {title} {Pure dephasing in flux qubits due to
  flux noise with spectral density scaling as $1/{f}^{\ensuremath{\alpha}}$},\
  }\href {https://doi.org/10.1103/PhysRevB.85.224505} {\bibfield  {journal}
  {\bibinfo  {journal} {Phys. Rev. B}\ }\textbf {\bibinfo {volume} {85}},\
  \bibinfo {pages} {224505} (\bibinfo {year} {2012})}\BibitemShut {NoStop}%
\bibitem [{\citenamefont {Bylander}\ \emph {et~al.}(2011)\citenamefont
  {Bylander}, \citenamefont {Gustavsson}, \citenamefont {Yan}, \citenamefont
  {Yoshihara}, \citenamefont {Harrabi}, \citenamefont {Fitch}, \citenamefont
  {Cory}, \citenamefont {Nakamura}, \citenamefont {Tsai},\ and\ \citenamefont
  {Oliver}}]{Bylander2011Noise}%
  \BibitemOpen
  \bibfield  {author} {\bibinfo {author} {\bibfnamefont {J.}~\bibnamefont
  {Bylander}}, \bibinfo {author} {\bibfnamefont {S.}~\bibnamefont
  {Gustavsson}}, \bibinfo {author} {\bibfnamefont {F.}~\bibnamefont {Yan}},
  \bibinfo {author} {\bibfnamefont {F.}~\bibnamefont {Yoshihara}}, \bibinfo
  {author} {\bibfnamefont {K.}~\bibnamefont {Harrabi}}, \bibinfo {author}
  {\bibfnamefont {G.}~\bibnamefont {Fitch}}, \bibinfo {author} {\bibfnamefont
  {D.~G.}\ \bibnamefont {Cory}}, \bibinfo {author} {\bibfnamefont
  {Y.}~\bibnamefont {Nakamura}}, \bibinfo {author} {\bibfnamefont {J.-S.}\
  \bibnamefont {Tsai}},\ and\ \bibinfo {author} {\bibfnamefont {W.~D.}\
  \bibnamefont {Oliver}},\ }\bibfield  {title} {\bibinfo {title} {Noise
  spectroscopy through dynamical decoupling with a superconducting flux
  qubit},\ }\href {https://doi.org/10.1038/nphys1994} {\bibfield  {journal}
  {\bibinfo  {journal} {Nature Physics}\ }\textbf {\bibinfo {volume} {7}},\
  \bibinfo {pages} {565} (\bibinfo {year} {2011})}\BibitemShut {NoStop}%
\bibitem [{\citenamefont {Bruzewicz}\ \emph {et~al.}(2019)\citenamefont
  {Bruzewicz}, \citenamefont {Chiaverini}, \citenamefont {McConnell},\ and\
  \citenamefont {Sage}}]{Bruzewicz2019TrappedIons}%
  \BibitemOpen
  \bibfield  {author} {\bibinfo {author} {\bibfnamefont {C.~D.}\ \bibnamefont
  {Bruzewicz}}, \bibinfo {author} {\bibfnamefont {J.}~\bibnamefont
  {Chiaverini}}, \bibinfo {author} {\bibfnamefont {R.}~\bibnamefont
  {McConnell}},\ and\ \bibinfo {author} {\bibfnamefont {J.~M.}\ \bibnamefont
  {Sage}},\ }\bibfield  {title} {\bibinfo {title} {Trapped-ion quantum
  computing: Progress and challenges},\ }\href
  {https://doi.org/10.1063/1.5088164} {\bibfield  {journal} {\bibinfo
  {journal} {Applied Physics Reviews}\ }\textbf {\bibinfo {volume} {6}},\
  \bibinfo {pages} {021314} (\bibinfo {year} {2019})}\ \Eprint
  {https://arxiv.org/abs/https://pubs.aip.org/aip/apr/article-pdf/doi/10.1063/1.5088164/19742554/021314_1_online.pdf}{}
  \BibitemShut {NoStop}%
\end{thebibliography}
\end{document}